\documentclass[a4paper,UKenglish,cleveref, autoref, thm-restate]{lipics-v2021}
\nolinenumbers

\usepackage{lipsum}
\usepackage{todonotes}

\usepackage[T1]{fontenc}
\usepackage{dsfont}
\usepackage{hyperref}
\usepackage{amsmath}
\usepackage{mathtools}
\usepackage{amssymb}
\usepackage{graphicx}
\usepackage{url}
\usepackage{nicefrac}
\usepackage{enumitem}
\usepackage{tikz} 
\usepackage{xspace}
\usetikzlibrary{arrows,decorations.markings}
\usepackage[normalem]{ulem}
\usepackage[noend]{algpseudocode}
\usepackage{comment}
\usepackage[ruled,vlined]{algorithm2e}

\usepackage{thmtools} 
\usepackage{thm-restate}

\usepackage{xcolor}
\usepackage{graphicx}

\title{Improved Approximations for Flexible Network Design}
\author{Dylan Hyatt-Denesik}{Eindhoven University of Technology, Eindhoven, Netherlands,}{d.v.p.hyatt-denesi@tue.nl}{}{}
\author{Afrouz Jabal-Ameli}{Eindhoven University of Technology, Eindhoven, Netherlands,}{a.jabal.ameli@tue.nl}{}{}
\author{Laura Sanit\`a}{Bocconi University, Milan, Italy}{laura.sanita@unibocconi.itg}{}{}
\authorrunning{D. Hyatt-Denesik et al.}
\Copyright{CC-BY} %TODO mandatory, please use full first names. LIPIcs license is "CC-BY";  http://creativecommons.org/licenses/by/3.0/

\ccsdesc[500]{Theory of computation~Approximation algorithms analysis} %TODO mandatory: Please choose ACM 2012 classifications from https://dl.acm.org/ccs/ccs_flat.cfm 

\keywords{Approximation Algorithms \and Network Design \and Flexible Connectivity.}

\EventEditors{John Q. Open and Joan R. Access}
\EventNoEds{2}
\EventLongTitle{42nd Conference on Very Important Topics (CVIT 2016)}
\EventShortTitle{CVIT 2016}
\EventAcronym{CVIT}
\EventYear{2016}
\EventDate{December 24--27, 2016}
\EventLocation{Little Whinging, United Kingdom}
\EventLogo{}
\SeriesVolume{42}
\ArticleNo{23}
\begin{document}

% \affil[1]{Eindhoven University of Technology, Eindhoven, The Netherlands, \{}
% \institute{Eindhoven University of Technology, Eindhoven, The Netherlands,}
% \email{\{d.v.p.hyatt-denesik,a.jabal.ameli\}@tue.nl}\\
% Bocconi University, Milan, Italy,\\
% \email{laura.sanita@unibocconi.it}}
%
\maketitle              % typeset the header of the contribution
\begin{abstract}
    % \todo{IPCO: \\
    % 1)15 pages, excluding references and appendices.\\
    % 2) LIPIcs style\\
    % 3)The first page should contain the title, the authors’ names with their affiliations, and a short abstract.}
    Flexible network design deals with building a network that guarantees some connectivity requirements between its vertices, even when some of its elements (like vertices or edges) fail. In particular, the set of edges (resp. vertices) of a given graph are here partitioned into \emph{safe} and \emph{unsafe}. The goal is to identify a minimum size subgraph that is 2-edge-connected (resp. 2-vertex-connected), and stay so whenever any of the unsafe elements gets removed.  

    In this paper, we provide improved approximation algorithms for flexible network design problems, considering both edge-connectivity and vertex-connectivity, as well as connectivity values higher than 2.  For the vertex-connectivity variant, in particular, our algorithm is the first with approximation factor strictly better than 2. 

\end{abstract}
%
%
%
% \listoftodos
\section{Introduction}
\label{sec:intro}
Survivable network design is an important area of combinatorial optimization. In a classical setting, we are given a network represented as a graph, and the goal is to select the cheapest subset of edges that guarantee some connectivity requirements among its vertices, even when some of its elements (like vertices or edges) may fail.

Two fundamental problems in this area are the \emph{$2$-edge-connected spanning subgraph} (2ECSS) and the \emph{$2$-vertex-connected spanning subgraph} (2VCSS). These problems aim at building a network resilient to a possible failure of an edge or of a vertex, respectively. More in detail, in 2ECSS, we are given in input a graph $G=(V,E)$, and the goal is to select a subset $F \subseteq E$ of minimum-cardinality 
such that the graph $(V,F)$ is 2-edge-connected: that is, $(V,F)$ contains 2 edge-disjoint paths between any pair of vertices.  
In 2VCSS, the input is the same, but here we require that our selected set $F$ is such that  $(V,F)$ is 2-vertex connected, i.e., it contains 2 vertex-disjoint paths between any pair of vertices. Both 2ECSS and 2VCSS have been extensively studied in the literature. They are APX-hard(see~\cite{czumaj1999approximability})
but admit constant-factor approximation algorithms (see e.g.~\cite{ameli20234,garg2023improved,KV94, sebHo2014shorter}). Currently, the best approximation for  2ECSS is 
$\frac{118}{89}$~\cite{garg2023improved}\footnote{Recently a $1.3$-approximation was presented for $2$ECSS (see~\cite{DBLP:conf/isaac/0001N23}) with the assumption that there exists a polynomial time algorithm for maximum triangle-free $2$-matching problem.}, while the best approximation for 2VCSS is $\frac{4}{3}$~\cite{ameli20234}. The problems have been investigated also in the \emph{weighted} setting. That is, when the input graph $G$ comes equipped with an edge-weight function $w \in \mathbb R_+$, and the goal is to minimize the total weight of the selected set $F$, rather than its cardinality. For this more general case, nothing better than a 2-approximation is known~\cite{jain2001factor, khuller1995improved}.

In recent years, an interesting generalization has been introduced by Adjiashvili et al.~\cite{adjiashvili2022flexible}, which soon received a lot of attention in the network design community. This generalization is called \emph{flexible} graph connectivity and is the focus of this paper. Specifically, the authors in~\cite{adjiashvili2022flexible} considered a scenario in which not all edges are subject to potential failures. The set of edges is partitioned into \emph{safe} and \emph{unsafe}, and the goal is to construct a network resilient to the failure of unsafe edges. A formal definition is given below.

\begin{definition}[Flexible Graph Connectivity Problem (FGC)]
    Given a graph $G=(V,E)$ and a partition of $E$ into safe edges $E_S$ and unsafe edges $E_U$, the goal is to find the smallest subset $E'$ of $E$ such that (1) $H=(V,E')$ is connected, (2) and for every edge $e\in E_U \cap E'$, the graph $(V, E'\setminus{e})$ is connected. 
\end{definition}

Next, we define the vertex-connectivity version of the problem, investigated first in~\cite{bansal2022extensions}. Recall that a \emph{cut-vertex} of a graph is a vertex $u$ such that if we remove $u$ and all its incident edges the number of connected components of the graph increases.
\begin{definition}[Flexible Vertex Connectivity Problem (FVC)]
    Given a graph $G=(V,E)$ and a partition of $V$ into safe vertices $V_S$ and unsafe vertices $V_U$, the goal is to find the smallest subset $E'$ of $E$ such that (1) $H=(V,E')$ is connected, (2) and for every vertex $u\in V_U$, $u$ is not a cut-vertex of $H$. 
\end{definition}

Note that, when $E_S=\emptyset$, FGC reduces to 2ECSS. Similarly, when $U_S=\emptyset$, FVC reduces to 2VCSS. Hence, these problems
are at least as hard as 2ECSS and 2VCSS, respectively. 

For FGC and some further generalizations, several approximation results have been developed in the past few years (see e.g.~\cite{adjiashvili2022flexible, adjiashvili2020fault, bansal2023constant, bansal2022extensions, bansal2023improved, bentert2023complexity, boyd2023approximation,  chekuri2023approximation, nutov2023improved}, and application of flexible graph based techniques). This list shows a growing interest in this problem and its variants in the literature. 

The current best known approximation factor for FGC is given in~\cite{boyd2023approximation}, which is based on the best 2ECSS approximation ratio. At the moment, using the best approximation result available in the literature for 2ECSS, \cite{garg2023improved}, this translates into a bound of $\frac{472}{385}\approx 1.45$.\footnote{If one applies the aforementioned 1.3-approximation found in \cite{DBLP:conf/isaac/0001N23}, then this approximation becomes $1.\bar 4$.}
%This is better than a 2-approximation, known for the weighted version~\cite{}. 

Of course, a first natural question is the following:
\begin{itemize}
\item \emph{Can the approximation factor for FGC be improved?}
\end{itemize}
Problems involving vertex-connectivity rather than edge-connectivity often turn out to be more challenging to address, and in fact
approximation results on FVC are more scarce.  The authors in~\cite{bansal2022extensions} studied the problem in the \emph{weighted} setting, and give a 2-approximation under the assumption to have at least one safe vertex (they target a more general setting of removing $k$ unsafe vertices).  
Note that a 2-approximation is known for the weighted setting of FGC as well, however as mentioned before, for the standard FGC a significantly better approximation is known. This raises the question of whether a better-than-2 approximation is possible also for FVC.

\begin{itemize}
\item \emph{Is there an algorithm for FVC with approximation factor better than 2?}
\end{itemize}

Finally, we observe that both 2ECSS and 2VCSS have been studied in a generalized setting, called $k$ECSS and $k$VCSS respectively, where one requires $k$-edge-disjoint paths (resp. $k$-vertex-disjoint paths) between each pair of vertices. 
It turns out that for higher values of connectivity requirements, much better approximations can be developed. In particular,~\cite{cheriyan2000approximating,gabow2009approximating} show an approximation factor of $1+ O(\frac{1}{k})$ for $k$ECSS and $k$VCSS. It is natural to ask what happens if we generalize FGC and FVC in a similar way. 
Specifically, we could aim at building graphs that are $k$-edge-connected (resp. $k$-vertex-connected) after the removal of any unsafe edge. 
The last question we are interested in is then the following:
\begin{itemize}
\item \emph{Can we obtain similar approximation bounds as observed for $k$ECSS or $k$VCSS, if we consider FGC and VFC with a higher value of $k$ as connectivity requirement?}
\end{itemize}

\paragraph{Our Results}
In this paper, we focus on the above questions and answer them in a positive way. 
% \todo{we say we answer the in the affirmative, however, we don't for $k$FVC.}

In Section~\ref{sec:FGC}, we prove Theorem~\ref{thm:FECapx}, which improves the best known approximation factor for FGC. 
The theorem is proved by giving a refined analysis of the algorithm developed in~\cite{boyd2023approximation}. Their algorithm relies on an approximation algorithm
for 2ECSS as a subroutine,  whose analysis is used mainly when the size of the optimal solution is large enough compared to the number of vertices $n$. Our improvement stems from realizing that 2ECSS can be approximated better than the current best known factor whenever the optimal solution is large compared to $n$.
\begin{theorem}
\label{thm:FECapx}
    There is a polynomial time $\frac{10}{7}$-approximation algorithm for FGC.
\end{theorem}

In Section~\ref{sec:11fvc}  we prove Theorem~\ref{thm:FVCapx}, which yields the first approximation
algorithm for FVC with an approximation factor strictly better than 2.
Its proof constitutes the most technical part of our paper. It combines two different main algorithms, which rely on  
non-trivial combinatorial ingredients, like ear-decompositions and matroid intersection.

\begin{theorem}
\label{thm:FVCapx}
    There is a polynomial time $\frac{11}{7}$-approximation algorithm for FVC.
\end{theorem}

Finally, in Section~\ref{sec:1kFEC} we answer the last question we raised, for FGC specifically. In particular, we consider the FGC problem for higher values of $k$, formally defined below.
\begin{definition}[$k$-Flexible Graph Connectivity Problem ($k$-FGC)]
    Given a graph $G=(V,E)$ and a partition of $E$ into safe edges $E_S$ and unsafe edges $E_U$, the goal is to find the smallest subset $E'$ of $E$ such that (1) $H=(V,E')$ is $1$-edge-connected, (2) and for every $k$ unsafe edges $\{e_1,\dots, e_k\} \subseteq E_U \cap E'$, the graph $(V, E'\setminus \{e_1,\dots, e_k\})$ is $1$-edge-connected. 
\end{definition}
% \begin{definition}[$k$-Flexible Graph Connectivity Problem ($k$-FGC)]
%     Given a graph $G=(V,E)$ and a partition of $E$ into safe edges $E_S$ and unsafe edges $E_U$, the goal is to find the smallest subset $E'$ of $E$ such that (1) $H=(V,E')$ is $k$-edge-connected, (2) and for every edge $e\in E_U \cap E'$, the graph $(V, E'\setminus{e})$ is $k$-edge-connected. 
% \end{definition}

We prove Theorem~\ref{thm:1kFECapx}, which shows that also $k$-FGC becomes somewhat easier to approximate when $k$ grows, as it happens for $k$ECSS.\footnote{We note that a (stronger) approximability of $1 + O \left( \frac{1}{k} \right)$ for $k$-FGC was previously claimed by \cite{adjiashvili2022flexible}, but their proof is flawed as we explain in Appendix~\ref{sec:bug}.}

\begin{theorem}
\label{thm:1kFECapx}
    There is a polynomial time $1 + O \left( \frac{1}{\sqrt{k}} \right)$-approximation algorithm for $k$-FGC.
\end{theorem}

We were not able to extend our arguments to FVC with higher values of connectivity requirement $k$ in a similar manner to $k$-FGC. We leave this as an open question.

\subsection{Preliminaries}
Here, we begin with some preliminaries to define and construct some tools that will be helpful later on in the paper. 

%%%
% contraction
%%%

Given a graph $G=(V,E)$ and a subset of vertices $U\subseteq V$. We denote by $G/U$ the (multi-)graph obtained by first replacing the vertices of $U$ by a single vertex $\hat U$, and then for every edge $uv \in E$ such that $u\in U$ and $v\in V\setminus U$
we add an edge $\hat U v$ to $E$ (Note that there can be multiple copies of the same edge $\hat U v$ in $G/U$, and there will be no loops on $\hat U$).
We sometimes refer to this as \emph{contracting} $G$ by the vertices $U$. We can define a similar operation on a subset of edges $F\subseteq E$, by taking the vertex sets of connected components of $(V, F)$, and contracting $G$ by these sets one by one. We denote this operation by $G/F$. 

%%%
% induced subgraph
%%%

Given a graph $G=(V,E)$, and a subset of vertices $U\subseteq V$, we define $G[U] = (U, \{u_1u_2\in E| u_1,u_2\in U\})$, the \emph{induced} subgraph of $U$. Similarly, we can define an induced subgraph on a subset of edges $F\subseteq E$, and by abuse of notation we use the same notation $G[F] = (\{v\in V| \exists u\in V, uv\in F\},F)$. We also say a graph $H=(V',E')$ is a spanning subgraph of $G$ if $H$ is a subgraph of $G$ and $V'=V$.

%%%
% Ear decompositions
%%%
\begin{definition}[Ear-Decomposition]
    Let $G=(V,E)$ be a graph. An \emph{ear-decomposition} is a sequence \( P_1, \ldots, P_k\), where \(P_1\) is a cycle of $G$, and for each \(i \in \{2, \ldots, k\}\), \(P_i\)  is either:
    \begin{itemize}
        \item a path sharing exactly its two endpoints with \(V(P_1) \cup \ldots \cup V(P_{i-1})\), or;
        \item a cycle that shares exactly one vertex with \(V(P_1) \cup \ldots \cup V(P_{i-1})\).
    \end{itemize}
    \(P_1, \ldots, P_k\) are called \emph{ears}. $P_i$ is an \emph{open} ear if it is a path. An ear-decomposition is \emph{open} if for every $i\in \{2,...,k\}$, $P_i$ is an open ear. We refer to $|E(P)|$ as the \emph{length} of $P$. 
    
    Given an open ear-decomposition \( P_1, \ldots, P_k\), we say that $P'$ is a potential open ear of \( P_1, \ldots, P_k\) if \( P_1, \ldots, P_k,P'\) is itself an open ear-decomposition. \( P_1, \ldots, P_k\) is an open ear-decomposition of a graph $G=(V,E)$ if \( (V(P_0) \cup \ldots \cup V(P_k)=V\) and for each $i$, $E(P_i)\subseteq E$. 
\end{definition}
We will often abbreviate the adjective `2-vertex-connected' with `2VC'. 
The following well known result on open ear-decompositions can be found in Chapter 4 of \cite{west2001introduction}.
\begin{lemma}[\cite{west2001introduction}]
\label{lem:2VCandOpenEars}
    A graph $G$ is $2$VC if and only if it has an open-ear decomposition.
\end{lemma}

%It is called \emph{open} if all ears except \(P_1\) are open.

%%%
% Blocks of 1VC graph
%%%

An important tool for the analysis of our algorithm in Section~\ref{sec:11/9fvc} is the following, which will let us characterize the vertex connectivity of a graph.
\begin{definition}[Blocks]
    Let $G=(V,E)$ be a graph such that $|V(G)|\ge 2$. A block of $G$ is a maximal connected subgraph of $G$ that has at least one edge and has no cut vertex. Therefore if $G$ has no self-loops, a block is either an induced connected subgraph on two vertices or it is a maximal $2$VC subgraph on at least three vertices. 
\end{definition}

We end this section by introducing some useful properties of blocks. For proofs and more details on these, we refer the reader to Chapter 4 of~\cite{west2001introduction}.
\begin{lemma}[\cite{west2001introduction}]
\label{lem:blockbound}
    Let $G=(V,E)$ be a graph, and let $\{B_1,\dots, B_k\}$ be the set of all blocks of $G$. The following properties hold
    \begin{enumerate}
        \item Two blocks share at most one vertex.
        \item The blocks $B_1,\dots, B_k$ of $G$ partition $E$, that is $E = \cup_{i=1}^k E(B_i)$, and $E(B_i) \cap E(B_j) = \emptyset$, for $i\neq j$.
        \item For two distinct edges  $e_1$ and $e_2$, $e_1$ and $e_2$ belong to the same block $B_i$ if and only if there is a cycle in $B_i$ that contains $e_1$ and $e_2$.
       % For any block $B_i$, $i=1,\dots, k$ with at least two edges $e_1,e_2$ in $E(B_i)$, there is a cycle in $B_i$ that contains $e_1$ and $e_2$.
     
        \item If $G$ is connected, $G$ has at most $|V|-1$ blocks.
    \end{enumerate}
\end{lemma}

 We also use the following useful Lemma, whose proof can be found in Appendix~\ref{sec:reduceblocks}.
\begin{lemma}
\label{lem:reduceblocks}
    Let $G=(V,E)$ be a connected graph with $B(G)$ many blocks. Let $H$ be a connected, spanning subgraph of $H$ with $B(H) > B(G)$ many blocks. In polynomial time, we can find an edge $e \in E(G)\backslash E(H)$ such that $H'\coloneqq (V(H),E(H)\cup \{e\})$ has fewer than $B(H)$ blocks.
\end{lemma}

\section{\texorpdfstring{$\frac{11}{7}$}{}-Approximation for FVC}
\label{sec:11fvc}
In this section we provide a $\frac{11}{7}$-approximation algorithm for Flexible Vertex Connectivity Problem. We assume that our given graph $G=(V,E)$, is a simple graph, that is, $G$ does not contain any loops or parallel edges (note that even if $G$ were not simple, parallel edges and loops would not help in finding a feasible solution), and has $n \coloneqq |V|$ vertices.
%We can assume this as for any feasible FVC solution $F$ containing loops or parallel edges, we can remove the loops and parallel edges of $F$ with no loss in feasibility. 

We begin this section by showing that in order to obtain a $\beta$-approximation for $\beta \geq 1$ one can assume that the input graph (i.e. $G$) has some additional properties, such as not containing specific subgraphs that we call \textit{forbidden cycles}.
\begin{definition}[Forbidden Cycle]
    We say that a $4$-cycle $C$ in $G$ is a forbidden cycle if $C$ has two vertices $w$ and $z$ such that $wz\notin E(C)$ and $\deg_G(w)=\deg_G(z)=2$.
\end{definition}
The following Lemma allows us to assume without loss of generality that our input graph $G$ is $2$VC and does not contain a forbidden cycle. Its proof can be found in Appendix~\ref{sec:removeforbiddencycles}.
\begin{lemma}
\label{lem:removeforiddencycles}
    For $\beta\ge 1$, if there is a $\beta$-approximation for FVC instances that are $2$VC and do not contain forbidden cycles, then there is a $\beta$-approximation for FVC. 
\end{lemma}

Furthermore, using the next Lemma we assume throughout Section~\ref{sec:11fvc} that $OPT\ge n$ as otherwise we can solve the problem optimally in polynomial time. The proof of this Lemma can be found in Appendix~\ref{sec:optnottree}
\begin{lemma}
\label{lem:optnottree}
    Let $OPT$ be an optimal solution to FVC instance $G=(V,E)$.  We have $|OPT| \geq n-1$. Furthermore, if $|OPT| = n-1$, then we can find such a solution in polynomial time.
\end{lemma}

\subsection{Approximation 1}
\label{sec:5/3}

% \textcolor{red}{I suggest we include a paragraph describing the notation used in the figures. For instance we can say in all the upcoming figures black vertices are safe, others are unsafe, edges in the solution are solid and the rest are}
By applying Lemma~\ref{lem:removeforiddencycles}, we can assume without loss of generality that $G$ is $2$VC and does not contain a forbidden cycle. We also assume that $G$ has at least 5 vertices, as we can handle smaller instances by enumeration. Our algorithm relies on the construction of a certain open ear-decomposition which we outline here. 
    Start with $D$ being a cycle of length at least four. We remark that such a cycle must exist, as $n\ge 4$, and $G$ is $2$VC (see~\cite{west2001introduction}). Moreover one can find such a cycle in polynomial time. Now, until there exists a potential open ear of $D$ that has a length of at least 4 for $D$, we find and add one such potential open ear to $D$.
 %   continually find and add potential open ears of $D$ that have length of at least 4, until there are no more potential open ears of $D$ that are the length of at least $4$. 
    %return D
We observe by Lemma~\ref{lem:2VCandOpenEars} that $D$ is $2$VC.
% \begin{algorithm}[ht]
%         {
%           %  Remove \textit{Forbidden Cycles} from G. (Lemmas~\ref{lem:forbiddencyclesUnsafe} and~\ref{lem:forbiddencycles})\\
%             Define $D\leftarrow C$, for some cycle $C$ of length at least 4 in $G$\\
%             \While{there is a path $C$ of length at least 4 in $G$ such with endpoints in $D$ and inner vertices not in $D$}
%             {   
%                 $V(D)\leftarrow V(D)\cup V(C)$, $E(D)\leftarrow E(D)\cup E(C)$. \\
%             }
%         }
%     Return \(D\)
%     \caption{Partial Ear Decomposition}
%     \label{alg:FVCapprox1}
% \end{algorithm} 
To show that this procedure terminates in polynomial time we rely on the following simple claim, proven in Appendix~\ref{sec:eardecompinpolytime}.
\begin{claim}
\label{lem:eardecompinpolytime}
    Let $D$ be an open-ear decomposition. If there exists a potential open ear of $D$ with a length of at least $4$, it can be detected in polynomial time. 
\end{claim}

The following Lemma, proven in Appendix~\ref{sec:decompositioninvariant}, provides an upper bound on the number of edges in $D$.
\begin{lemma}
\label{lem:decompositioninvariant}
    When the algorithm terminates, we have $|E(D)| \le \frac{4}{3}(|V(D)|-1)$.
\end{lemma}
% \begin{claim}
%     $|V(D)|\ge 4$ and $D$ is $2$VC.
% \end{claim}

% \begin{proof}
%     It follows as we start with a cycle of length $4$ and the fact that $D$ is an open decomposition on $G[V(D)]$.
% \end{proof}
With this open ear decomposition $D$, the following key Lemma shows that the edges of $G[V\setminus V(D)]$ have a useful structure, that will be critical to our approximation algorithms.
\begin{lemma}
\label{lem:matching}
    The edges of $G[V\setminus V(D)]$ form a (not necessarily perfect) matching.
\end{lemma} 
\begin{proof}
    For simplicity, denote $G' \coloneqq G[V \setminus V(D)]$. If $E(G') = \emptyset$, we are done. Otherwise, consider edge $uv\in E(G')$. We first wish to show that there are vertices $u'\neq v'\in V(D)$, such that $uu',vv'\in E$. That is, $uv$ is part of a potential open ear for $D$ of length three.

    Since $G$ is $2$VC there are $u' \neq v'\in V(D)$, such that there exist vertex disjoint paths $P_u$ and $P_v$ from $u$ to $u'$ and $v$ to $v'$ respectively. Since $u'\neq v'$, the path $P_u \cup \{uv\} \cup P_v$ is a potential open ear of $D$. By our construction, there are no potential open ears of length at least $4$ possible anymore. Thus we can see the length of the path $P_u \cup \{uv\} \cup P_v$ is at most $3$, and hence exactly $3$. Therefore, $P_u =\{uu'\}$ and $P_v= \{vv'\}$. 
    
    To prove the claim, we assume for contradiction $G'$ has a vertex $v$ of degree at least two. That is, we assume there are edges $uv,vw\in E(G')$ for $u\neq w$. We distinguish the following cases for the path $u,v,w$:
    \begin{figure}[t]
        \begin{center}
            \begin{tabular}{c c c}
                (a)    
\begin{tikzpicture}
    
    % Node styles
    % \tikzset{black dot/.style={draw=black, very thick, circle,minimum size=0pt, inner sep=1pt, outer sep=1pt,fill=black}}
    \tikzstyle{black dot}=[fill=black, draw=black, circle, minimum size=0pt,inner sep=2pt, outer sep=2pt]
    \tikzset{terminal/.style={draw=black,  thick,minimum size=0pt, inner sep=2.5pt, outer sep=1pt}}
    \tikzset{P node/.style={fill={rgb,255: red,20; green,154; blue,0}, draw={rgb,255: red,20; green,154; blue,0}, circle, minimum size=0pt,inner sep=1pt, outer sep=1pt}}
    \tikzset{Writing/.style={shape=circle} }
    \tikzstyle{empty circle}=[fill=white, draw=black, shape=circle]

    % Edge styles
    \tikzstyle{witness edge}=[-, draw={rgb,255: red,195; green,0; blue,3}, very thick]
    \tikzstyle{T edges}=[-, thick]
    \tikzstyle{Fat edge}=[-, ultra thick]
    \tikzstyle{new witness}=[-, draw={rgb,255: red,195; green,0; blue,3}, dashed, ultra thick]
    \tikzstyle{connected terminals}=[-, draw=black, dashed, very thick]
    \tikzstyle{P}=[-, draw={rgb,255: red,20; green,154; blue,0}, very thick]

		\node [style=Writing] (0) at (-1.75, 1.5) {};
		\node [style=Writing] (1) at (1.75, 1.5) {};
		\node [style=Writing] (2) at (1.75, -1.5) {};
		\node [style=Writing] (3) at (-1.75, -1.5) {};
		\node [style=Writing] (4) at (-1.225, 0) {$D$};
		\node (5) at (-1, 1) {};
		\node (6) at (-1, -1) {};
		\node [style=black dot] (7) at (-0.775, 0.625) {};
		\node [style=black dot] (8) at (-0.825, -0.675) {};
		\node [style=black dot] (9) at (0.5, 1) {};
		\node [style=black dot] (10) at (0.5, -1) {};
		\node [style=black dot] (11) at (1.25, 0) {};
		\node [style=Writing] (12) at (0.85, 1.25) {$u$};
		\node [style=Writing] (13) at (1.5, 0.325) {$w$};
		\node [style=Writing] (14) at (0.875, -1.025) {$v$};
		\node [style=Writing] (15) at (-0.65, -0.4) {$v'$};
		\node [style=Writing] (16) at (-0.575, 0.325) {$u'$};
		\node [style=Writing] (17) at (0.75, 0) {$C$};
		\draw [style=T edges, bend right=90] (-1, 1) to (-1, -1);
		\draw [style=T edges, bend left=90] (-1, 1) to (-1, -1);
		\draw [style=T edges] (7) to (9);
		\draw [style=T edges, bend left=15, looseness=1.25] (9) to (11);
		\draw [style=T edges, bend left, looseness=0.75] (11) to (10);
		\draw [style=T edges, bend left=15] (10) to (9);
		\draw [style=T edges] (8) to (10);
\end{tikzpicture}&
                (b)    
\begin{tikzpicture}
    
    % Node styles
    % \tikzset{black dot/.style={draw=black, very thick, circle,minimum size=0pt, inner sep=1pt, outer sep=1pt,fill=black}}
    \tikzstyle{black dot}=[fill=black, draw=black, circle, minimum size=0pt,inner sep=2pt, outer sep=2pt]
    \tikzset{terminal/.style={draw=black,  thick,minimum size=0pt, inner sep=2.5pt, outer sep=1pt}}
    \tikzset{P node/.style={fill={rgb,255: red,20; green,154; blue,0}, draw={rgb,255: red,20; green,154; blue,0}, circle, minimum size=0pt,inner sep=1pt, outer sep=1pt}}
    \tikzset{Writing/.style={shape=circle} }
    \tikzstyle{empty circle}=[fill=white, draw=black, shape=circle]

    % Edge styles
    \tikzstyle{witness edge}=[-, draw={rgb,255: red,195; green,0; blue,3}, very thick]
    \tikzstyle{T edges}=[-, thick]
    \tikzstyle{Fat edge}=[-, ultra thick]
    \tikzstyle{new witness}=[-, draw={rgb,255: red,195; green,0; blue,3}, dashed, ultra thick]
    \tikzstyle{connected terminals}=[-, draw=black, dashed, very thick]
    \tikzstyle{P}=[-, draw={rgb,255: red,20; green,154; blue,0}, very thick]

		\node [style=Writing] (0) at (-1.75, 1.5) {};
		\node [style=Writing] (1) at (1.75, 1.5) {};
		\node [style=Writing] (2) at (1.75, -1.5) {};
		\node [style=Writing] (3) at (-1.75, -1.5) {};
		\node [style=Writing] (4) at (-1.475, 0.25) {$D$};
		\node [style=Writing] (5) at (-1, 1.25) {};
		\node [style=Writing] (6) at (-1, -1.25) {};
		\node [style=black dot] (7) at (-0.675, 0.525) {};
		\node [style=black dot] (8) at (-0.925, -0.85) {};
		\node [style=black dot] (9) at (0.75, 1.25) {};
		\node [style=black dot] (10) at (0.9, -0.35) {};
		\node [style=black dot] (11) at (1, 0.5) {};
		\node [style=Writing] (12) at (1.1, 1.25) {$u$};
		\node [style=Writing] (13) at (1.25, -0.275) {$w$};
		\node [style=Writing] (14) at (1.375, 0.475) {$v$};
		\node [style=Writing] (15) at (-1.025, -0.15) {$v'$};
		\node [style=Writing] (16) at (-1.0, 0.85) {$u'=w'$};
		\node [style=black dot] (17) at (0.5, -1) {};
		\node [style=black dot] (18) at (-0.675, -0.25) {};
		\node [style=Writing] (19) at (0.875, -1) {$z$};
		\node [style=Writing] (20) at (-1.25, -0.675) {$z'$};
		\draw [style=T edges, bend right=90] (-1, 1.25) to (-1, -1.25);
		\draw [style=T edges, bend left=90] (-1, 1.25) to (-1, -1.25);
		\draw [style=T edges] (7) to (9);
		\draw [style=T edges] (7) to (10);
		\draw [style=T edges] (18) to (11);
		\draw [style=T edges] (9) to (11);
		\draw [style=T edges] (11) to (10);
		\draw [style=T edges] (10) to (17);
		\draw [style=T edges] (17) to (8);
\end{tikzpicture}&
                (c)    
\begin{tikzpicture}
    
    % Node styles
    % \tikzset{black dot/.style={draw=black, very thick, circle,minimum size=0pt, inner sep=1pt, outer sep=1pt,fill=black}}
    \tikzstyle{black dot}=[fill=black, draw=black, circle, minimum size=0pt,inner sep=2pt, outer sep=2pt]
    \tikzset{terminal/.style={draw=black,  thick,minimum size=0pt, inner sep=2.5pt, outer sep=1pt}}
    \tikzset{P node/.style={fill={rgb,255: red,20; green,154; blue,0}, draw={rgb,255: red,20; green,154; blue,0}, circle, minimum size=0pt,inner sep=1pt, outer sep=1pt}}
    \tikzset{Writing/.style={shape=circle} }
    \tikzstyle{empty circle}=[fill=white, draw=black, shape=circle]

    % Edge styles
    \tikzstyle{witness edge}=[-, draw={rgb,255: red,195; green,0; blue,3}, very thick]
    \tikzstyle{T edges}=[-, thick]
    \tikzstyle{Fat edge}=[-, ultra thick]
    \tikzstyle{new witness}=[-, draw={rgb,255: red,195; green,0; blue,3}, dashed, ultra thick]
    \tikzstyle{connected terminals}=[-, draw=black, dashed, very thick]
    \tikzstyle{P}=[-, draw={rgb,255: red,20; green,154; blue,0}, very thick]

		\node [style=Writing] (0) at (-1.75, 1.5) {};
		\node [style=Writing] (1) at (1.75, 1.5) {};
		\node [style=Writing] (2) at (1.75, -1.5) {};
		\node [style=Writing] (3) at (-1.75, -1.5) {};
		\node [style=Writing] (4) at (-1.275, 0) {$D$};
		\node [style=Writing] (5) at (-1, 1) {};
		\node [style=Writing] (6) at (-1, -1) {};
		\node [style=black dot] (7) at (-0.625, 0.25) {};
		\node [style=black dot] (8) at (-0.825, -0.675) {};
		\node [style=black dot] (9) at (0.5, 1) {};
		\node [style=black dot] (10) at (0.5, -0.5) {};
		\node [style=black dot] (11) at (1.25, 0.25) {};
		\node [style=Writing] (12) at (0.85, 1) {$u$};
		\node [style=Writing] (13) at (0.5, -0.175) {$w$};
		\node [style=Writing] (14) at (1.625, 0.225) {$v$};
		\node [style=Writing] (15) at (-0.65, -0.4) {$v'$};
		\node [style=Writing] (16) at (-0.825, 0.575) {$u'$};
		\draw [style=T edges, bend right=90] (-1, 1) to (-1, -1);
		\draw [style=T edges, bend left=90] (-1, 1) to (-1, -1);
		\draw [style=T edges] (7) to (9);
		\draw [style=T edges] (9) to (11);
		\draw [style=T edges] (11) to (10);
		\draw [style=T edges] (7) to (10);
		\draw [style=T edges, bend left=75] (11) to (8);
		\draw [style=connected terminals, bend right=45] (7) to (8);
\end{tikzpicture}
            \end{tabular}
        \end{center}
        \caption{
        (a) Here we have a triangle $C$, with $V(C)=\{u,v,w\}$, in $G'$. A potential open ear of length $4$ for $D$ is the path $uu',uw,wv, vv'$. 
        (b) $u,v,w,z$ is a path in $G'$. Each vertex in this path is adjacent to a vertex $u',v',w',z'\in V(D)$, respectively. Since $w'\neq z'$ by construction, we can find a potential open ear of $D$ of length at least $4$ for $D$.
        (c) Here the path $u,v,w$ is in neither case $a$ nor $b$. Thus, $u$ and $w$ are degree $2$ in $G$, with identical neighbourhood. So $u',u,v,w$ is a forbidden square.
        % Moreover, $v$ and $u'$ are degree at least $3$ in $G$, 
        }
        \label{fig:matchings}
    \end{figure}
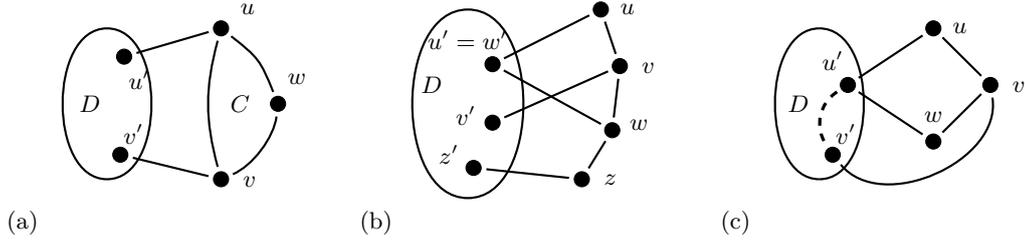
        
    \textbf{Case a:} $uw$ is an edge of $G'$. We know there are vertices $u'\neq v'\in V(D)$ with $uu',vv'\in E$. Now $\{u'u,uw,wv,vv'\}$ is the edge set of a potential open ear of $D$ of length four, contradicting our construction of $D$. See Figure~\ref{fig:matchings}.(a).

    \textbf{Case b:} $u,v,w$ is part of a path $P$ of length three in $G'$. Without loss of generality let $u,v,w,z$ be a subpath of $P$ in $G'$ ($z \notin \{u, v, w\}$). We know there are vertices $u', v', w', z'\in V(D)$, with $u'\neq v'$, $w'\neq z'$, such that $uu'$, $vv'$, $ww', zz' \in E$. Since $w'\neq z'$, at most one of $u' = w'$ and $u' = z'$ can be true. Assume for contradiction that $u' \neq w'$. Then the path $u',u,v,w,w'$  is a potential open ear for $D$ of length $4$. Contradicting our construction of $D$.
    If $u' \neq z'$, we can find a similar contradiction. See Figure~\ref{fig:matchings}.(b).

    \textbf{Case c:} $u,v,w$ is neither in case $a$ nor in case $b$. We know there are vertices $u',w'\in V(D)$, such that $uu'$, $ww'\in E$. 
    If $u'\neq w'$, then the path  $u',u,v,w,w'$ is a potential open ear of $D$ of length $4$, a contradiction. 
    So we see $u'=w'$. So this implies that $u'$ is the unique neighbour of $u$ and $w$ in $V(D)$, or again we find a potential open ear of $D$ of length $4$. Furthermore, as we assume we are not in case a or b, we can observe that $v$ and $u'$ are the only neighbours of $u$ and $w$, ie, $deg_G(w) = deg_G(u)=2$. That is, $u,v,w,u'$ is a forbidden cycle, contradicting our application of Lemma~\ref{lem:removeforiddencycles}.
    % If there is an edge $uw\in E$, then $u,v,w$ is a cycle, and we are in case $a$, contradicting the case assumption. If $u$ and $w$ are  adjacent to any vertex of $V(G')$ other than $v$, then we are in case $b$, contradicting the case assumption. Therefore, $u$ and $w$ are degree $2$ in $G$.
    % Lastly, since $u'$ and $v'$ are in $V(D)$, there is a path in $D$ between $u'$ and $v'$. Thus, $u'$ and $w$ have degree at least $3$ in $G$. Hence, $u',u,v,w,u'$ is a forbidden cycle, contradicting our application of Lemma~\ref{lem:removeforiddencycles}. See Figure~\ref{fig:matchings}.(c).
    % Thus, $G'$ does not contain a vertex of degree more than $2$, and the edges of $G'$ form a matching.
\end{proof}
Since the edges of $G[V\setminus V(D)]$ form a matching, each connected component of $G[V\setminus V(D)]$ is either a singleton or an edge.
We will now partition the vertices of $V\setminus V(D)$ into the sets $K_{1,1}, K_{1,2}, K_{2,2},$ and $K_{2,3}$. 

We define $K_{1,1}\subseteq V\setminus V(D)$ as the singletons of $G[V\setminus V(D)]$ that have a safe vertex neighbour in $V(D)$. 
We define $K_{1,2}\subseteq V\setminus V(D)$ as the singletons of $G[V\setminus V(D)]$ that do not have a safe vertex neighbour in $V(D)$. 
We define $K_{2,2}\subseteq V\setminus V(D)$ as the endpoints of edges $uv$ in $G[V\setminus V(D)]$ that satisfy one of the following: $u$ and $v$ are both adjacent to a (possibly equal) safe vertex in $V(D)$, or one of $u$ or $v$ are safe, and that safe vertex is adjacent to a safe vertex in $V(D)$. 
Finally, we define $K_{2,3} = V\setminus (V(D) \cup K_{1,1}\cup K_{1,2}\cup K_{2,2})$.

Intuitively, we imagine the set $K_{i,j}$ to represent vertices of the components of $G[V\setminus V(D)]$ with $i$ vertices, where any feasible FVC of solution of $G$ must have at least $j$ edges with endpoints in a component.
% {\color{red} this really wants a figure} 

% Let $K_{1,1} = |V_{1,1}|$, and $K_{1,2}=|V_{1,2}|$, the number of vertices of $V_{1,1}$ and $V_{1,2}$ respectively. Let $K_{2,2}= |E[V_{2,2}]|$, and $K_{2,3} = |E[V_{2,3}]|$, the edges 
% between vertices of$V_{2,2}$ and $V_{2,3}$ respectively. Note that by Lemma~\ref{lem:matching}, we have $K_{2,2}= |V_{2,2}|/2$, and $K_{2,2}= |V_{2,3}|/2$.

% Now let $K_1$ be the number of vertices of degree zero in $G':=G[V(G)\setminus V(H)]$ and let $K_2$ be $|E(G')|$. We have $n:=|V(G)|=|V(H)|+K_1+2K_2$. Note that any feasible solution has at least $2K_1+3K_2$ edges that have at least one endpoint in $V(G')$. 

We now describe our first algorithm, which will compute a feasible solution $APX_1$. Starting with $APX_1 \coloneqq \emptyset$. We first add the edges of $D$ to $APX_1$. If $V\backslash V(D)= \emptyset$, then $APX_1$ is feasible since $D$ is 2VC, and we are done. Otherwise, we buy edges in the following way to make $APX_1$ feasible. 

First, for every $v\in K_{1,1}$, buy an edge $uv\in E$ where $u\in V(D)$ is safe (such a $u$ exists, by definition of $K_{1,1}$).
Second, for every $v\in K_{1,2}$, we buy arbitrary pair of edges $uv,u'v\in E$, where $u\neq u'$ (these edges exists since $G$ is assumed to be 2VC).
Third, for every edge $uv$ of $G[K_{2,2}]$, if $u$ and $v$ are both adjacent to (potentially equal) safe vertices $u'$ and $v'$ in $V(D)$, then we buy the edges $uu'$ and $vv'$. 
If at least one of $u$ and $v$, is not adjacent to a safe vertex in $V(D)$, then by definition of $K_{2,2}$, (without loss of generality) $v$ is safe and is adjacent to a safe vertex in $v'\in V(D)$; In that case, we buy edges $uv,vv'$.
Lastly, for each edge $uv$ in $G[K_{2,3}]$, we buy edge $uv$ and arbitrary pair of edges $uu'$ and $vv'$, where $u' \neq v'$. Observe that such edges must exist again as $G$ is $2$VC and $n\ge4$.  
%Note that $K_{i,j}$

The next Lemma is proven in Appendix~\ref{sec:apx1feasible}

% Our first algorithm, denoted $APX_1$, is to buy the edges of $D$ computed in Algorithm~\ref{alg:FVCapprox1}, and buy edges to connect $G'$ to $D$ in the following way. For every singleton $v\in K_{1,1}$, $v$ is adjacent to a safe vertex in $D$, buy an arbitrary edge $uv\in E$ where $u$ is safe. For every singleton $v\in K_{1,2}$, buy two edges incident on $v$ with distinct endpoints in $D$. 
% For every edge $uv\in V_{2,2}$, if $u$ and $v$ are adjacent to (potentially equal) safe vertices $u'$ and $v'$ in $D$, buy an arbitrary pair of edges $uu'$ and $vv'$. 
% If instead $uv\in V_{2,2}$ and $u,v$ are not both adjacent to a safe vertex in $D$, then by definition (without loss of generality) $u$ is safe, and adjacent to safe vertex $u'$ in $D$, we buy $uv$ and an arbitrary such $uu'$. 
% Lastly, for each edge $uv$ in $G[V_{2,3}$, we buy edge $uv$ and arbitrary pair of edges $uu'$ and $vv'$ where $u'\neq v'$.

\begin{lemma}
\label{lem:apx1feasible}
    Our algorithm computes a feasible FVC solution, $APX_1$, in polynomial time.
\end{lemma}
By our construction of $APX_1$, $|APX_1|\le |E(D)|+ |K_{1,1}| + 2|K_{1,2}| + |K_{2,2}| + \frac{3}{2}|K_{2,3}|$. The following Lemma is clear since $|E(D)|\leq \frac{4}{3}(|V(D|-1)$, by Lemma~\ref{lem:decompositioninvariant}.
\begin{lemma}
\label{lem:earAPX}
    $|APX_1| \leq \frac{4}{3}(|V(D)| -1) + |K_{1,1}| + 2|K_{1,2}| + |K_{2,2}| + \frac{3}{2}|K_{2,3}|$
\end{lemma}
% The new lemma 
We fix an optimal solution $OPT$ to the instance $G = (V,E)$. The following Lemma provides a lower bound on $|OPT|$. The proof can be found in Appendix~\ref{sec:UpperboundSimple}.

\begin{lemma}\label{lem:UpperboundSimple}
    $|OPT|\ge \max \{ n, |K_{1,1}| + 2|K_{1,2}| + |K_{2,2}| + \frac{3}{2}|K_{2,3}|\}$
\end{lemma}
% {\color{red}
% \begin{lemma}
% \label{lem:nolargecomponent}
%     If $K_{1,2}, K_{2,3} = \emptyset$, then Algorithm~\ref{alg:FVCapprox1} is a  $\frac{4}{3}$-approximation.
% \end{lemma}}
In the next Lemma show that even with the tools we have already developed, we have a $\frac{5}{3}$-approximation, answering the question if there exists an approximation factor less than $2$ in the affirmative. We spend the remainder of the section improving on this easier approximation. The proof can be found in Appendix~\ref{sec:5/3apx}.
\begin{lemma}
\label{lem:5/3apx}
    $APX_1$ is a $\frac{5}{3}$-approximate solution for FVC. Furthermore, if $|K_{1,2}| + |K_{2,3}| \le 2$, then $APX_1$ is a  $\frac{4}{3}$-approximate solution.
    % $|K_{1,2}| + |K_{2,3}| < \frac{1}{\varepsilon}$, for some $\varepsilon >0$, then $APX_1$ is a  $(\frac{4}{3} + \varepsilon)$-approximation.
    % Let $n,|V(H)|,K_1$ and $K_2$ be non-negative integers such that $n=|V(H)|+K_1+2K_2$. Then we have:
    % $$S=\frac{\frac{4}{3}(|V(H)|-1)+2K_1+3K_2}{\max\{n,2K_1+3K_2\}}\le \frac{5}{3}$$
\end{lemma}

\subsection{Approximation 2}
\label{sec:11/9fvc}
In this section, we will provide a second algorithm that relies on the vertex sets defined in Section~\ref{sec:5/3}, that were computed by our first algorithm. Namely, we are interested in $V(D)$, $K_{1,1}$, $K_{1,2}$, $K_{2,2}$, and $K_{2,3}$. This algorithm, when combined with the algorithm of Section~\ref{sec:5/3} will achieve a $\frac{11}{7}$-approximation. Applying Lemma~\ref{lem:5/3apx}, we assume that at least one of $K_{1,2}$ or $K_{2,3}$ is non-empty, as otherwise we immediately have a $\frac{4}{3}$-approximation algorithm.

As in the previous section, we will rely on the fact that $D$ is constructed as a 2VC subgraph of $G$, and the fact that any feasible solution must take edges incident to vertices in $K_{1,1}$, $K_{1,2}$, $K_{2,2}$, and $K_{2,3}$. However this time we start by buying a minimal set of edges $E'$ incident to $K_{1,1}$, $K_{1,2}$, $K_{2,2}$, and $K_{2,3}$ that are required for feasibility and then we complement these edges with a subset of edges of $G[D]$ to obtain a feasible solution $APX_2$. %In order to choose these edges we would take $E'$ among all possible sets.

Before we describe our approximation algorithm, we will define the following tools that the algorithm relies on. 
\begin{definition}[Maximum Rainbow Connection Problem]
    Given a (multi-)graph $G$ and a coloring $c:E\rightarrow\mathbb{N}$ of the edges, find a spanning subgraph of $G$ that minimizes the number of components, while choosing exactly one edge from each colour.
    % find the maximum number of edges with pairwise distinct colors that form a forest.
\end{definition}

This problem can be solved to optimality using matroid intersection between the graphic matroid, and the partition matroid on the colour classes. This is shown in the following Lemma which is proven  Appendix~\ref{sec:RainbowForest}.

\begin{lemma}\label{lem:RainbowForest}
    Given an instance of the  Maximum Rainbow Connection Problem with (multi)-graph $G = (V,E)$, with coloring $c: E \rightarrow \mathbb{N}$. We can find in polynomial time, an optimal solution $P$ such that the number of isolated vertices (vertices of degree 0) in $(V,P)$ is minimal with respect to replacing an edge with another edge of the same colour class.
\end{lemma}
%Let $D$,  $K_{1,1}$, $K_{1,2}$, $K_{2,2}$ and $K_{2,3}$ be defined as in Section~\ref{sec:5/3}.
Our goal with the Maximum Rainbow Connection problem is to more cleverly find a minimal set $E'$ than simply taking an \textit{arbitrary}  minimal set of edges incident on at least one vertex of $V\setminus V(D)$. Our goal is to select such an $E'$ that also minimizes the number of connected components of $(V, E')$. To do this we use the edges of $E$ that are incident on $K_{1,1}$, $K_{1,2}$, $K_{2,2}$, and $K_{2,3}$ to create an instance of Maximum Rainbow Connection Problem (and solve it with Lemma~\ref{lem:RainbowForest}). 

% The idea is that in order to reduce the size of $APX_2$, we do not initiate it with an \textit{arbitrary}  minimal set $E'$ of edges incident on at least one vertex of $V\setminus V(D)$. 
% Instead we pick $E'$ in a way that $(V,E')$ has as few connected components as possible. To obtain such set $E'$ we create an equivalent instance of Maximum Rainbow Connection Problem and solve it using Lemma~\ref{lem:RainbowForest}.  

We define a set of so-called \emph{pseudo-edges} $\tilde E$  with endpoints in $V(D)$ (named pseudo-edges in order to distinguish them from the "real" edges, $E$), and assign to each pseudo-edge a unique colour indexed by vertices in $K_{1,2}$ and pairs of adjacent vertices in $K_{2,3}$. 
For every $u\in K_{1,2}$, and  every distinct pair of edges $uv_1$ and $uv_{2}\in E$, we add pseudo-edge $v_1v_2$ to $\tilde E$. Assign $v_1v_2$ the colour $c_{u}$. 
Intuitively, a pseudo-edge $xy$ with colour $c_v$ (for example) corresponds to a path in $E$ from $x$ to $v$ to $y$.

For every ordered pair $(u,v) \in K_{2,3}\times K_{2,3}$ such that $uv\in E$, we add pseudo-edges to $\tilde E$ in the following way: 
(1) for every pair of edges $uu'$ and $vv'$ with $u'\neq v'$, we add pseudo-edge $u'v'$ to $\tilde E$. Assign $u'v'$ the colour $c_{uv}$. 
(2) If $u$ is adjacent to safe vertex $u'\in V(D)$, and for any two neighbours $v_1',v_2' \in V(D)$ of $v$ (if $v$ has at least two neighbours in $V(D)$), we add pseudo-edge $v_1'v_2'$ to $\tilde E$. We assign $v_1'v_2'$ the colours $c_{uv}$ (note $v$ is not adjacent to a safe vertex since $u,v\notin K_{2,2}$). 
(3) If $u$ is safe, then for every  distinct pair of edges $uv_1$ and $uv_{2}$, we add pseudo-edge $v_1v_2$ to $\tilde E$. Assign $v_1v_2$ the colour $c_{uv}$. 

Again, a pseudo-edge $xy$ with colour $c_{uv}$  corresponds to a path from $x$ to $y$ created by a minimum selection of the edges incident on $x$ and $y$ of a feasible FVC solution.

Notice that with $(V(D), \tilde E)$ (along with the corresponding edge colours we provide) describe an instance of the Maximum Rainbow Connection problem.
We use Lemma~\ref{lem:RainbowForest} to compute a solution to the Maximum Rainbow Connection problem, which we denote by $P$. Say that $P$ has $\alpha$ many components, and $\alpha_{large}$ many components with at least two vertices.
Then we use the following three algorithms one by one to obtain a feasible solution.
We have one last tool we need to provide before we can define for the next step of our algorithm, which is inspired by techniques employed in \cite{garg2023improved}. This tool will be useful for letting us decide which edges of $\tilde E$ to buy.
\begin{definition}[Good Cycles]
    Let \(\Pi = (V_1, \ldots, V_k)\), \(k \geq 2\), be a partition of the vertex-set of a 2VC simple graph \(G\). 
    
    A good cycle $C$ of $\Pi$ is a subset of edges with endpoints in distinct subsets of $\Pi$ such that:  (1) $C$ is a cycle of length at least $2$ in the graph obtained from $G$ by contracting each $V_i$ into a single vertex one by one ( that is, $G/V_1/V_2/\dots /V_k$); (2) given any two edges of $C$ incident to some $V_i$, these edges are incident to distinct vertices of \(V_i\) unless \(|V_i| = 1\); (3) $C$ has an edge incident to at least one $V_i$ with $|V_i| \geq 2$, and; (4) $|C| = 2$ only if both $V_i$ and $V_j$ incident to $C$ have $|V_i| , |V_j| \geq 2$.
\end{definition}

The following Lemma allows us to compute good cycles in polynomial time. Its proof can be found in Appendix~\ref{sec:findgoodcycles}. 
\begin{lemma}
\label{lem:findgoodcycles}
    Let \(\Pi = \{V_1, \ldots, V_k \} \), \(k \geq 2\), be a partition of the vertex-set of a 2VC simple graph \(G\) with the following conditions: $G[V_i]$ is connected for all $i=1,\dots,k$, with at least one set of size at least $2$.  Furthermore, if there is exactly one $V_i\in \Pi$ with $|V_i| \geq 2$, then there are at least two singletons in $\Pi$ that are adjacent in $G$.
    % Let \(\Pi = (V_1, \ldots, V_k)\), \(k \geq 2\), be a partition of the vertex-set of a 2VC simple graph \(G\), with at least one set of size at least $2$. 

    Then in polynomial time, one can compute a good cycle $C$ of $\Pi$.
\end{lemma}

%phase 1 description:
In our algorithm we will distinguish between connected components that have one vertex (i.e. singletons) and connected components with at least two vertices. Given a graph $H$, we say that a connected component is \textit{large} if it has at least two vertices.
Now we have all the ingredients required to finalize the description of our approximation algorithm. After computation of $P$, we initialize the set that will eventually be our solution as $APX_2\leftarrow \emptyset$. 

Our plan is to gradually build $APX_2$. We have three steps, Algorithm~\ref{alg:phase1},~\ref{alg:phase2} and~\ref{alg:phase3}, which we apply one after another. These algorithms return edge sets $S_1$, $S_2$ and $S_3$, respectively that will be part of our solution $APX_2$. We now take time to describe these algorithms in more details.

%It should be clear that any connected component that is not large, is a singleton. 

First we use Algorithm~\ref{alg:phase1},  which takes the pseudo-edges $P$, and find the large components and  a subset $X_1$ of singletons of $(V(D),P)$,
buying a subset of edges, $S_1$, to connect them into a single component $A$ in $(V(D), S_1\cup P)$ using Lemma~\ref{lem:findgoodcycles}. Upon termination of this algorithm, we will obtain the additional property that $V(D)\backslash V(A)$ is an independent set in $G$.

\begin{algorithm}[ht]
    {
        \textbf{Input:} pseudo-edges $P$\\
        % Given pseudo-edges $P$ that do the pseudo-edge thing\\
        Let $Y$ denote the singletons of $(V(D),P)$\\
        $S_1 \leftarrow \emptyset$\\
        \While{There is a good cycle $C$ in $G[V(D)]$ with connected components of $(V(D), P\cup S_1)$ being the vertex partitioning }
        {
            % Find a good cycle $C$ in $(V(D), P\cup S_1)$\\
            $S_1\leftarrow S_1\cup E(C)$\\
        }
        % Find good cycles until we have at most one large component, $A$, in $D$ and \textcolor{blue}{$V(D)\setminus V(A)$ is an independent set of vertices.}\\
        $A\leftarrow$ unique large component of $(V(D), S_1\cup P)$ \\
        Let $X_1 \leftarrow Y\cap V(A)$\\

        \textbf{Return} $(X_1, Y, S_1, A)$
    }
    \caption{Buying Good Cycles}
    \label{alg:phase1}
\end{algorithm} 

In Algorithm~\ref{alg:phase2}, our goal is to buy a minimal set $S_2$ of edges between $V(A)$ and compute a subset of $V(D) \backslash V(A)$ (which we denote by  $X_2$), such that $(V(A)\cup X_2,P\cup S_1\cup S_2)$ has one block. We remark that after termination of this algorithm any vertex of $X_3:=V(D)\setminus V(A)\cup X_2)$ is a singleton in $(V(D),P\cup S_1\cup S_2)$
\begin{algorithm}[ht]
    {
        \textbf{Input:} Edges $S_1\subseteq E$, pseudo-edges $P$, large component $A$\\
        $X_2,S_2 \leftarrow \emptyset$\\
        % \While{There is a an edge $e\in G[V(A)]$, such that adding $e$ to $S_2$ decreases the number of blocks in $(V(A), P\cup S_1 \cup S_2)$}
        % {
        %     $S_2 \leftarrow S_2 \cup \{e\}$\\
        % }
        % As long as t, then .{\color{red} should be made a while loop} \\
        \While{$G[V(A)\cup X_2] \cup P$ has more than one block}
        {
            Find $v\in V(D)\backslash (V(A) \cup X_2)$ such that $G[V(A)\cup X_2\cup\{v\}]\cup P$ has fewer blocks than $G[V(A)\cup X_2]\cup P$\\
            Find edges $e_1=vu$, $e_2=vw$ ,$v\neq w \in V(A)$, such that $(V(A) \cup X_2 \cup \{v\}, P \cup S_1\cup S_2 \cup \{e_1, e_2\})$ has fewer blocks than $(V(A) \cup X_2 , P \cup S_1 \cup S_2 )$ \\
            $X_2 \leftarrow X_2 \cup \{v\}$\\
            $S_2 \leftarrow S_2 \cup \{e_1,e_2\}$\\
            
        }
        \While{There exists an edge $e_3 = uz \in E\backslash S$ such that $(V(A) \cup X_2,  P \cup S_1\cup S_2 \cup \{e_3\})$ has fewer blocks than $(V(A) \cup X_2,  P \cup S_1\cup S_2 )$ }
            {
                $S_2 \leftarrow S_2\cup \{e_3\}$\\
            }
        \textbf{Return} $(X_2, S_2)$
    }
    \caption{Making Large Component 2VC}
    \label{alg:phase2}
\end{algorithm} 

In Algorithm~\ref{alg:phase3}, the goal is to buy a subset $S_3$ of edges such that $(V(D),S_1\cup S_2\cup S_3\cup P)$ is a feasible FVC solution on $V(D)$ (i.e. $(V(D),S_1\cup S_2\cup S_3\cup P)$ is connected and has no unsafe cut-vertices).  For every vertex in $v\in X_3\coloneqq V(D)\backslash (V(A) \cup X_2) $ that has a safe neighbour in $V(A)\cup X_2$, we buy one edge from $v$ to one of its safe neighbour in $V(A)\cup X_2$. For any other vertex in $X_3$  we buy two distinct edges from it to $V(A)\cup X_2$. We define $\alpha_1'$ as the number of vertices of $X_3$ that have a safe neighbour, and $\alpha_2'$ is the number of vertices of $X_3$ that do not have a safe neighbour. Thus $|X_3| = \alpha_1' + \alpha_2'$. To maintain a consistent notation for number of components we define $\alpha' \coloneqq |X_3| = \alpha_1' + \alpha_2'$. 
Note that this implies, $\alpha = \alpha_{large} + |X_1| + |X_2| + |X_3| =  \alpha_{large} + |X_1| + |X_2| + \alpha'.$

\begin{algorithm}[ht]
    {
        \textbf{Input:} Singletons $X_2$, large component $A$\\
        $S_3\leftarrow \emptyset$\\
        $X_3 \leftarrow  V(D)\backslash (V(A)\cup X_2)$\\
        $\alpha_1' \leftarrow 0$, $\alpha_2' \leftarrow 0$\\
        \For{every $v\in X_3$ }
        {
            \If{$v$ is adjacent to safe vertex $u\in V(A)\cup X_2$ }
            {
                $S_3 \leftarrow S_3 \cup \{uv\}$, $\alpha_1' \leftarrow \alpha_1' +1$\\
            }
            \Else
            {
                Find $u,w\in V(A)\cup X_2$ adjacent to $v$ \\
                $S_3\leftarrow S_3 \cup \{uv, uw\}$, $\alpha_2' \leftarrow \alpha_2' +1$\\
            }
        }
        \textbf{Return} $(X_3, S_3, \alpha_1', \alpha_2')$
    }
    \caption{Making Solution Feasible}
    \label{alg:phase3}
\end{algorithm} 

% Given pseudo-edges $P$ computed by Lemma~\ref{lem:RainbowForest} from $\tilde E$, we build a solution $S$ in the following way. 

% First we choose edges of $K$, set $S_P\coloneqq E(K)$\todo{$S_P$ isn ot defined really} and let $P$ be their corresponding pseudo-edges, and say that $(V(D),P)$ has $\alpha$ many components (hence $(V,S)$ has $\alpha$ many components), and $\alpha_{large}$ many components with at least two vertices. 
% We run Algorithm~\ref{alg:phase1},~\ref{alg:phase2}, and~\ref{alp3}, as described to compute verteix$X_1,X_2$, 

% with $P$, to find $X_1$ and $S_1$. We run Algorithm~\ref{alg:phase2} with $S_1$, to find $X_2$ and $S_2$. Finally, we run Algorithm~\ref{alg:phase3} with $S_2$ to find $X_3,$ $S_3$, $\alpha_1'$, and $\alpha_2'$. 

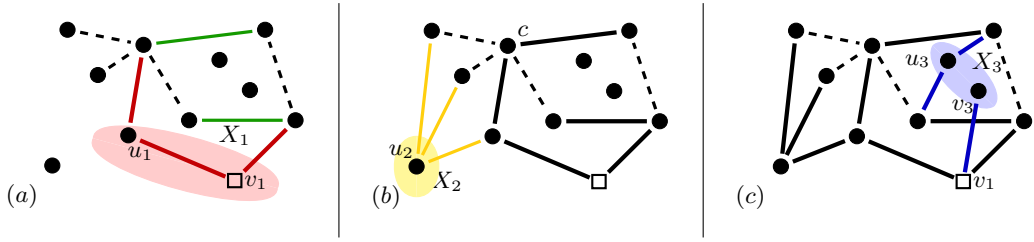
\begin{figure}[t]
    \begin{center}
        \begin{tabular}{c | c |c }
                
\begin{tikzpicture}[scale=0.8]
    
    % Node styles
    % \tikzset{black dot/.style={draw=black, very thick, circle,minimum size=0pt, inner sep=1pt, outer sep=1pt,fill=black}}
    \tikzstyle{black dot}=[fill=black, draw=black, circle, minimum size=0pt,inner sep=2pt, outer sep=2pt]
    \tikzset{terminal/.style={draw=black,  thick,minimum size=0pt, inner sep=2.5pt, outer sep=1pt}}
    \tikzset{P node/.style={fill={rgb,255: red,20; green,154; blue,0}, draw={rgb,255: red,20; green,154; blue,0}, circle, minimum size=0pt,inner sep=1pt, outer sep=1pt}}
    \tikzset{Writing/.style={shape=circle} }
    \tikzstyle{empty circle}=[fill=white, draw=black, shape=circle]

    % Edge styles
    \tikzstyle{witness edge}=[-, draw={rgb,255: red,195; green,0; blue,3}, very thick]
    \tikzstyle{T edges}=[-, thick]
    \tikzstyle{Fat edge}=[-, ultra thick]
    \tikzstyle{new witness}=[-, draw={rgb,255: red,195; green,0; blue,3}, dashed, ultra thick]
    \tikzstyle{connected terminals}=[-, draw=black, dashed, very thick]
    \tikzstyle{P}=[-, draw={rgb,255: red,20; green,154; blue,0}, very thick]
    \tikzstyle{Strong witness}=[-, draw={rgb,255: red,195; green,0; blue,3}, ultra thick]
    \tikzstyle{Blue edge}=[-, draw={rgb,255: red,3; green,0; blue,195}]
    \tikzstyle{arrow}=[<-, very thick]
    \tikzstyle{Marked Edges}=[-, draw={rgb,255: red,255; green,207; blue,14}, very thick]

        \node [style=Writing] (17) at (1.45, -1.375) {};
		\node [style=Writing] (21) at (-1.325, -0.55) {};
        \node [style=Writing] (16) at (-1.25, -0.25) {};
		\node [style=Writing] (22) at (1.475, -1.55) {};
		\filldraw [color=red!20,  bend left=105, looseness=0.50] (1.7, -1.375) to (-1.325, -0.55);
		\filldraw [color=red!20,  bend right=105, looseness=0.50] (1.7, -1.375) to(-1.325, -0.55);
		\node [style=Writing] (19) at (1, -0.5) {$X_1$};

		\node [style=Writing] (0) at (-2.25, 1.5) {};
		\node [style=Writing] (1) at (2.25, 1.5) {};
		\node [style=Writing] (2) at (2.25, -1.5) {};
		\node [style=Writing] (3) at (-2.5, -1.5) {$(a)$};
		\node [style=black dot] (4) at (-1.75, 1.25) {};
		\node [style=black dot] (5) at (-0.5, 1) {};
		\node [style=black dot] (6) at (-2, -1) {};
		\node [style=black dot] (7) at (0.25, -0.25) {};
		\node [style=black dot] (8) at (2, -0.25) {};
		\node [style=black dot] (10) at (1.5, 1.25) {};
		\node [style=terminal] (11) at (1, -1.25) {};
		\node [style=black dot] (12) at (-0.75, -0.5) {};
		\node [style=black dot] (13) at (0.75, 0.75) {};
		\node [style=black dot] (14) at (-1.25, 0.5) {};
		\node [style=black dot] (15) at (1.25, 0.25) {};
		\node [style=Writing] (20) at (-0.55, -0.8) {$u_1$};
		\node [style=Writing] (21) at (1.35, -1.275) {$v_1$};
		\draw [style=connected terminals] (4) to (5);
		\draw [style=connected terminals] (5) to (7);
		\draw [style=connected terminals] (8) to (10);
		\draw [style=connected terminals] (14) to (5);
		\draw [style=P] (7) to (8);
		\draw [style=P] (5) to (10);
		\draw [style=Strong witness] (11) to (8);
		\draw [style=Strong witness] (11) to (12);
		\draw [style=Strong witness] (12) to (5);
		% \draw [style=Blue edge] (7) to (13);
		% \draw [style=Blue edge] (13) to (10);
		% \draw [style=Marked Edges] (4) to (6);
		% \draw [style=Marked Edges] (6) to (12);		
\end{tikzpicture}&
                
\begin{tikzpicture}[scale=0.8]
    
    % Node styles
    % \tikzset{black dot/.style={draw=black, very thick, circle,minimum size=0pt, inner sep=1pt, outer sep=1pt,fill=black}}
    \tikzstyle{black dot}=[fill=black, draw=black, circle, minimum size=0pt,inner sep=2pt, outer sep=2pt]
    \tikzset{terminal/.style={draw=black,  thick,minimum size=0pt, inner sep=2.5pt, outer sep=1pt}}
    \tikzset{P node/.style={fill={rgb,255: red,20; green,154; blue,0}, draw={rgb,255: red,20; green,154; blue,0}, circle, minimum size=0pt,inner sep=1pt, outer sep=1pt}}
    \tikzset{Writing/.style={shape=circle} }
    \tikzstyle{empty circle}=[fill=white, draw=black, shape=circle]

    % Edge styles
    \tikzstyle{witness edge}=[-, draw={rgb,255: red,195; green,0; blue,3}, very thick]
    \tikzstyle{T edges}=[-, thick]
    \tikzstyle{Fat edge}=[-, ultra thick]
    \tikzstyle{new witness}=[-, draw={rgb,255: red,195; green,0; blue,3}, dashed, ultra thick]
    \tikzstyle{connected terminals}=[-, draw=black, dashed, very thick]
    \tikzstyle{P}=[-, draw={rgb,255: red,20; green,154; blue,0}, very thick]
    \tikzstyle{Strong witness}=[-, draw={rgb,255: red,195; green,0; blue,3}, line width=4pt]
    \tikzstyle{Blue edge}=[-, draw={rgb,255: red,3; green,0; blue,195}]
    \tikzstyle{arrow}=[<-, very thick]
    \tikzstyle{Marked Edges}=[-, draw={rgb,255: red,255; green,207; blue,14}, very thick]

		\node [style=Writing] (25) at (-2, -1.5) {};
		\node [style=Writing] (26) at (-2, -0.5) {};
		\node [style=Writing] (27) at (-1.5, -1.25) {$X_2$};
		\filldraw [color=yellow!50, bend left=90, looseness=1.25] (-2, -1.5) to (-2, -0.5);
		\filldraw [color=yellow!50, bend right=90, looseness=1.25] (-2, -1.5) to (-2, -0.5);        
    
		\node [style=Writing] (0) at (-2.25, 1.5) {};
		\node [style=Writing] (1) at (2.25, 1.5) {};
		\node [style=Writing] (2) at (2.25, -1.5) {};
		\node [style=Writing] (3) at (-2.5, -1.5) {$(b)$};
		\node [style=black dot] (4) at (-1.75, 1.25) {};
		\node [style=black dot] (5) at (-0.5, 1) {};
		\node [style=black dot] (6) at (-2, -1) {};
		\node [style=black dot] (7) at (0.25, -0.25) {};
		\node [style=black dot] (8) at (2, -0.25) {};
		\node [style=black dot] (10) at (1.5, 1.25) {};
		\node [style=terminal] (11) at (1, -1.25) {};
		\node [style=black dot] (12) at (-0.75, -0.5) {};
		\node [style=black dot] (13) at (0.75, 0.75) {};
		\node [style=black dot] (14) at (-1.25, 0.5) {};
		\node [style=black dot] (15) at (1.25, 0.25) {};
		\node [style=Writing] (22) at (-0.25, 1.25) {$c$};
		\node [style=Writing] (24) at (-2.25, -0.75) {$u_2$};
		\draw [style=connected terminals] (4) to (5);
		\draw [style=connected terminals] (5) to (7);
		\draw [style=connected terminals] (8) to (10);
		\draw [style=connected terminals] (14) to (5);
		\draw [style=Fat edge] (7) to (8);
		\draw [style=Fat edge] (11) to (12);
		\draw [style=Fat edge] (12) to (5);
		\draw [style=Fat edge] (11) to (8);
		\draw [style=Fat edge] (5) to (10);
		\draw [style=Marked Edges] (4) to (6);
		\draw [style=Marked Edges] (6) to (14);
		\draw [style=Marked Edges] (6) to (12);
\end{tikzpicture}&
                
\begin{tikzpicture}[scale=0.8]
    
    % Node styles
    % \tikzset{black dot/.style={draw=black, very thick, circle,minimum size=0pt, inner sep=1pt, outer sep=1pt,fill=black}}
    \tikzstyle{black dot}=[fill=black, draw=black, circle, minimum size=0pt,inner sep=2pt, outer sep=2pt]
    \tikzset{terminal/.style={draw=black,  thick,minimum size=0pt, inner sep=2.5pt, outer sep=1pt}}
    \tikzset{P node/.style={fill={rgb,255: red,20; green,154; blue,0}, draw={rgb,255: red,20; green,154; blue,0}, circle, minimum size=0pt,inner sep=1pt, outer sep=1pt}}
    \tikzset{Writing/.style={shape=circle} }
    \tikzstyle{empty circle}=[fill=white, draw=black, shape=circle]

    % Edge styles
    \tikzstyle{witness edge}=[-, draw={rgb,255: red,195; green,0; blue,3}, very thick]
    \tikzstyle{T edges}=[-, thick]
    \tikzstyle{Fat edge}=[-, ultra thick]
    \tikzstyle{new witness}=[-, draw={rgb,255: red,195; green,0; blue,3}, dashed, ultra thick]
    \tikzstyle{connected terminals}=[-, draw=black, dashed, very thick]
    \tikzstyle{P}=[-, draw={rgb,255: red,20; green,154; blue,0}, very thick]
    \tikzstyle{Strong witness}=[-, draw={rgb,255: red,195; green,0; blue,3}, ultra thick]
    \tikzstyle{Blue edge}=[-, draw={rgb,255: red,3; green,0; blue,195}, ultra thick]
    \tikzstyle{arrow}=[<-, very thick]
    \tikzstyle{Marked Edges}=[-, draw={rgb,255: red,255; green,207; blue,14}, very thick]

		\node [style=Writing] (29) at (0.5, 1) {};
		\node [style=Writing] (30) at (1.5, 0) {};
		\filldraw [color=blue!20, bend right=75, looseness=0.75] (0.5, 1) to (1.5, 0);
		\filldraw [color=blue!20, bend right=90, looseness=0.75] (1.5, 0) to (0.5, 1);
    
		\node [style=Writing] (0) at (-2.25, 1.5) {};
		\node [style=Writing] (1) at (2.25, 1.5) {};
		\node [style=Writing] (2) at (2.25, -1.5) {};
		\node [style=Writing] (3) at (-2.5, -1.5) {$(c)$};
		\node [style=black dot] (4) at (-1.75, 1.25) {};
		\node [style=black dot] (5) at (-0.5, 1) {};
		\node [style=black dot] (6) at (-2, -1) {};
		\node [style=black dot] (7) at (0.25, -0.25) {};
		\node [style=black dot] (8) at (2, -0.25) {};
		\node [style=black dot] (10) at (1.5, 1.25) {};
		\node [style=terminal] (11) at (1, -1.25) {};
		\node [style=black dot] (12) at (-0.75, -0.5) {};
		\node [style=black dot] (13) at (0.75, 0.75) {};
		\node [style=black dot] (14) at (-1.25, 0.5) {};
		\node [style=black dot] (15) at (1.25, 0.25) {};
		\draw [style=connected terminals] (4) to (5);
		\draw [style=connected terminals] (8) to (10);
		\draw [style=connected terminals] (5) to (7);
		\draw [style=connected terminals] (14) to (5);
		\draw [style=Fat edge] (7) to (8);
		\draw [style=Fat edge] (11) to (12);
		\draw [style=Fat edge] (12) to (5);
		\draw [style=Fat edge] (11) to (8);
		\draw [style=Fat edge] (5) to (10);
		\draw [style=Blue edge] (7) to (13);
		\draw [style=Blue edge] (13) to (10);
		\draw [style=Fat edge] (4) to (6);
		\draw [style=Fat edge] (6) to (14);
		\draw [style=Fat edge] (6) to (12);
		\draw [style=Blue edge] (15) to (11);
		\node [style=Writing] (28) at (1.35, -1.275) {$v_1$};
		\node [style=Writing] (31) at (1.4, 0.725) {$X_3$};
		\node [style=Writing] (32) at (0.25, 0.75) {$u_3$};
		\node [style=Writing] (33) at (1, 0) {$v_3$};
\end{tikzpicture}
        \end{tabular}
    \end{center}
    \caption{ A depiction of the edges and vertex sets found by Algorithms~\ref{alg:phase1},~\ref{alg:phase2}, and \ref{alg:phase3} in $V(D)$. Here the unsafe vertices are depicted by black circles. In this example there is only one safe vertex, $v_1$ in the set $V(D)$ that is shown by a square.\\
    (a) The dashed edges are pseudo-edges $P$ found by Lemma~\ref{lem:RainbowForest}. Algorithm~\ref{alg:phase1} first computes good cycle on green edges that merges two large components of pseudo-edges, then it finds the red cycle that merges the new large component and 2 singletons. $X_1 = \{u_1, v_1\}$
    (b) The yellow edges of the second figure are found by Algorithm~\ref{alg:phase2} which cover the cut-vertex $c$ in the component. The interior vertex is $X_2 = \{u_2\}$.
    (c) The blue edges of the third figure are the edges found by Algorithm~\ref{alg:phase3}, which add edges to the solution that bring $u_3$ and $v_3$ into $V(D)$ form a feasible FVC solution. The vertex $v_1$ is a safe vertex so we only add one edge ($x_3v_1$) incident on $v_3$.
     }
    \label{fig:rainbow}
\end{figure}

We finalize our solution for the instance by computing $S_P$, edges with endpoints in $K_{1,1}, K_{1,2}, K_{2,2},$ and $K_{2,3}$, by considering each pseudo-edge $\tilde e= v_1v_2\in P$.
If $\tilde e$ has colour $c_{u}$, then $u\in K_{1,2}$. We add edges $uv_1$ and $uv_2$ to $S_P$. 

If $\tilde e$ has colour $c_{uv}$, then $u,v\in K_{2,3}$. We add edges to $S_P$ in exactly one of the following ways (breaking ties in an arbitrary but fixed manner): 
    (1) we add $uv_1, vv_2$ and $uv$ to $S_P$ if $uv_1,vv_2\in E$; 
    (2) we add $uv_2, vv_1$ and $uv$ to $S_P$ if $uv_2,vv_1\in E$; 
    (3) we add $uv_1$, $uv_2$, and $uv$ to $S_P$ if $u$ is safe, and  $uv_1$, $uv_2$, $uv \in E$; 
    (4) we add $vv_1$, $vv_2$, and $uv$ to $S_P$ if $v$ is safe, and  $vv_1$, $vv_2$, $uv \in E$; 
    (5) we add $vv_1$, $vv_2$, and $uu_1$ to $S_P$ if there exist a safe vertex $u_1\in V(D)$ such that $uu_1\in E$, and  $vv_1$, $vv_2 \in E$, and; 
    (6) we add $uv_1$, $uv_2$, and $vu_1$ to $S_P$ if there exist a safe vertex $u_1\in V(D)$ such that $vu_1\in E$, and  $uv_1$, $uv_2 \in E$.
    
For every $v\in K_{1,1}$, buy an edge $uv\in E$ where $u\in V(D)$ is safe. By definition of $K_{1,1}$, $u$ exists. Also, for every $uv\in E(K_{2,2})$, if $u$ and $v$ are both adjacent to  safe vertices $u'$ and $v'$ in $V(D)$, then we buy the edges $uu'$ and $vv'$. If at least one of $u$ and $v$, (wlog say $u$) is not adjacent to a safe vertex in $V(D)$, then by definition of $K_{1,2}$, $v$ must be safe, and $v$ must be adjacent to a safe vertex $v'\in V(D)$ in $G$. In this case we buy edges $vv',uv$.
Observe that by construction, $|S_P|=|K_{1,1}| + 2| K_{1,2}| +2| K_{2,2}| + \frac{3}{2}|K_{2,3}|$.
The output of our algorithm is $APX_2 \coloneqq S_P \cup S_1 \cup S_2 \cup S_3$. The following Lemma shows that $APX_2$ is feasible, can be computed in polynomial time, as well as an upper bound on $APX_2$. Its proof can be found in Appendix~\ref{sec:FVCupperbound} 
\begin{lemma}
\label{lem:FVCupperbound}
    Our algorithm computes a feasible solution $APX_2 = S_P \cup S_1 \cup S_2 \cup S_3$, in polynomial time and  $|APX_2| \leq |V(D)|-1 + |S_P| + \alpha - 1 - (\frac{\alpha - \alpha'}{2} + \frac{\alpha_{large}}{2} + \alpha_{1}')$.
\end{lemma}

\subsection{Approximation Factor}

We fix an optimal solution, $OPT$, for the instance $G=(V,E)$. The following Lemma, proven in Appendix~\ref{sec:FVC3lowerbounds}, finds a set of lower bounds on $|OPT|$ that depend on terms found by our algorithm, in particular $K_{1,1}, K_{1,2}, K_{2,2}, K_{2,3}, \alpha, \alpha_{large}, \alpha_1',$ and $\alpha_2'$. 
We Recall that $|S_P| = |K_{1,1}| + 2|K_{1,2}| + 2|K_{2,2}|  + \frac{3}{2}|K_{2,3}|$. 
\begin{lemma}
\label{lem:FVC3lowerbounds}
    $|OPT| \geq \max\{|S_P| + \alpha - 1, 2K_{1,2} -2\alpha_{large} + \alpha_{1}' + 2\alpha_{2}',n\}$
\end{lemma}
With Lemma~\ref{lem:FVC3lowerbounds}, Lemma~\ref{lem:earAPX}, and Lemma~\ref{lem:FVCupperbound} we have the tools necessary to prove Theorem~\ref{thm:FVCapx}.
\begin{proof}[Proof of Theorem~\ref{thm:FVCapx}]
    Given instance of $(1,1)$-FVC, $G=(V_S\cup V_U, E)$. We apply Lemma~\ref{lem:removeforiddencycles} to assume without loss of generality that $G$ does not contain any forbidden cycles and $G$ is 2VC.
    We first find solution $APX_1$, and vertex sets $D$, $K_{1,1}, K_{1,2}, K_{2,2},$ and  $K_{2,3}$. By Lemma~\ref{lem:apx1feasible}, $APX_1$ is a feasible solution that we obtain in polynomial time.
    By Lemma~\ref{lem:earAPX}, we have $|APX_1| \leq \frac{4}{3}(|V(D)| -1) + |K_{1,1}| + 2|K_{1,2}| + |K_{2,2}| + \frac{3}{2}|K_{2,3}| = \frac{4}{3}(|V(D)| -1) + |S_P|$.

    Using sets $V(D)$, $K_{1,1}, K_{1,2}, K_{2,2},$ and  $K_{2,3}$ we apply Lemma~\ref{lem:RainbowForest} to compute a set of pseudo-edges $P$ on $V(D)$ with $\alpha$ many components and $\alpha_{large}$ many large components (at least 2 vertices). We then apply Algorithms~\ref{alg:phase1},~\ref{alg:phase2}, and~\ref{alg:phase3} as described in Section~\ref{sec:11/9fvc} to compute edge sets $S_1,S_2$, and $S_3$, as well as $\alpha_1'$ and $\alpha_2'$, where $\alpha' =\alpha_1' + \alpha_2'$. We then find edge set $S_P$ by replacing pseudo-edges with corresponding edges, and let $APX_2 = S_P \cup S_1 \cup S_2 \cup S_3$. By Lemma~\ref{lem:FVCupperbound} computing $APX_2$ in this way takes polynomial time and $|APX_2| \leq |V(D)|-1 + |S_P| + \alpha - 1 - (\frac{\alpha - \alpha'}{2} + \frac{\alpha_{large}}{2} + \alpha_{1}')$.
    
  %  We fix an optimal solution $OPT$ for this instance.
  By Lemma~\ref{lem:FVC3lowerbounds}, we have $|OPT| \geq \max\{|S_P| + \alpha - 1, 2K_{1,2} -2\alpha_{large} + \alpha_{1}' + 2\alpha_{2}',n\}$. Therefore, we have $\frac{\min\{APX_1, APX_2\}}{|OPT|}$ is at most:
  % With Lemma~\ref{lem:FVC3lowerbounds}, Lemma~\ref{lem:earAPX}, and Lemma~\ref{lem:FVCupperbound}, we have $\frac{\min\{APX_1, APX_2\}}{|OPT|}$ is at most:
%    $2K_{1,2} -2\alpha_{large} + \alpha_{1}' + 2\alpha_{2}'$
    \begin{align*}
        \frac{\min\{ \frac{4}{3}(|V(D)| -1) + |S_P| , |V(D)|-2 + |S_P| + \alpha - (\frac{\alpha - \alpha'}{2} + \frac{\alpha_{large}}{2} + \alpha_{1}')\}}{\max\{|S_P| + \alpha - 1, 2K_{1,2} -2\alpha_{large} + \alpha_{1}' + 2\alpha_{2}',n\}} \leq \frac{11}{7}.
    \end{align*}
    Where the last inequality will be proven in Appendix~\ref{sec:11/7}.
    % Finally, by applying Lemma~\ref{lem:11/7}, we have that our solution (that is the smaller set among $APX_1$ and $APX_2$) is a $\frac{11}{7}$-approximate solution.
\end{proof}
\section{FGC Improvement}
\label{sec:FGC}
The goal of this section is to prove Theorem~\ref{thm:FECapx}. We use the algorithm described by \cite{boyd2023approximation}, and provide a tighter analysis. In particular, their algorithm combines two different algorithms, and takes the best solution output among the two. 
Our improvement is based on a better analysis of the second one.

The lemma below comes from  \cite{boyd2023approximation}, in particular using their first algorithm. The authors prove the following (see Claim 6.4 of \cite{boyd2023approximation}).

\begin{lemma}[\cite{boyd2023approximation}]%\todo{I don't know how to put [7] not in "emph"}
\label{lem:FECapx1}
  Given a FGC instance with optimal solution OPT, one can compute in polynomial time a solution $F_1$
  %$J'$ with $|J'| \leq \frac{1}{2}|OPT \cap E_U|$. Hence 
  with $|F_1| \leq |OPT \cap E_S| + \frac{3}{2}|OPT\cap E_U|$.
\end{lemma}

The second algorithm used in \cite{boyd2023approximation} needs to be described fully, as we modify its analysis.

\smallskip
\emph{Algorithm $2$}. Given a FGC instance defined on a graph $G$, consider the graph $G''$ obtained from $G$ by duplicating every safe edge in $E$. 
%Similarly, let $F''$ be the multiedge-set obtained from OPT by duplicating every safe edge of OPT. Clearly, $(V, F'')$ is a 2ECSS of $G''$ consisting of $2|OPT\cap E_S| + |OPT\cap E_U|$ edges. 
Run an $\beta$-approximation algorithm for the 2ECSS problem on $G''$, and let $F_2$ be the output. Drop extra copies of safe edges from $F_2$. 
\smallskip

The authors prove the following claim (see Claim 6.5 of \cite{boyd2023approximation}).
\begin{lemma}[\cite{boyd2023approximation}]
    We have $|F_2| \leq 2\beta|OPT\cap E_S| + \beta |OPT\cap E_U|$.
\end{lemma}
In order to improve this algorithm we show that one can find a better approximation for 2ECSS if the size of the optimal solution is far enough from $n$, the number of vertices in $G$:

\begin{lemma}
\label{lem:2ECSSImprovement}
    Let $G$ be a 2ECSS instance. One can find in polynomial time a solution $APX$ of size $\frac{4}{3}n+\frac{2}{3}(x-1)$, where $x = |OPT|-n$. 
\end{lemma}

Instead of proving Lemma~\ref{lem:2ECSSImprovement},  we here prove the following lemma, which assumes the instance is 2VC. Lemma~\ref{lem:2ECSSImprovement} is proved in Appendix~\ref{sec:jensextension} by considering the blocks of $G$.

\begin{lemma}
\label{lem:jensextension}
    Let $G$ be a 2ECSS instance that is 2VC. One can find in polynomial time a solution $APX$ of size $\frac{4}{3}n+\frac{2}{3}(x-1)$, where $x = |OPT|-n$. 
\end{lemma}
\begin{proof}
    We provide a refined analysis of the algorithm provided in \cite{sebHo2014shorter}. 
    
    The authors in \cite{sebHo2014shorter} first find what they call a "nice" ear decomposition.  
     To define it, let's introduce some terminology. The minimum number of ears of 1 one are called trivial, while  ears of length $2$ and $3$ are called short.
     A vertex is pendant if it is not an endpoint of any non-trivial ear, and an ear is pendant if it is non-trivial and all its internal vertices are pendant. A nice ear decomposition is an ear decomposition with minimum number of even ears, in which there are no trivial ears, all short ears are pendant, and internal vertices of distinct short ears are non-adjacent.  See Figure~\ref{fig:pendant} for an example of a nice ear decomposition, and a clarification of short ears.
     Now  consider the nice ear decomposition found in~\cite{sebHo2014shorter} and  let $\pi$ denote the number of short ears, $\pi_i$ denote the number of ears of length $i$, and $\phi(G)$ denote the number of even length ears.

\begin{figure}[t]
    \begin{center}
        \begin{tabular}{c c c}
            (a)     
\begin{tikzpicture}
    
    % Node styles
    % \tikzset{black dot/.style={draw=black, very thick, circle,minimum size=0pt, inner sep=1pt, outer sep=1pt,fill=black}}
    \tikzstyle{black dot}=[fill=black, draw=black, circle, minimum size=0pt,inner sep=2pt, outer sep=2pt]
    \tikzset{terminal/.style={draw=black,  thick,minimum size=0pt, inner sep=2.5pt, outer sep=1pt}}
    \tikzset{P node/.style={fill={rgb,255: red,20; green,154; blue,0}, draw={rgb,255: red,20; green,154; blue,0}, circle, minimum size=0pt,inner sep=1pt, outer sep=1pt}}
    \tikzset{Writing/.style={shape=circle} }
    \tikzstyle{empty circle}=[fill=white, draw=black, shape=circle]

    % Edge styles
    \tikzstyle{witness edge}=[-, draw={rgb,255: red,195; green,0; blue,3}, very thick]
    \tikzstyle{T edges}=[-, thick]
    \tikzstyle{Fat edge}=[-, ultra thick]
    \tikzstyle{new witness}=[-, draw={rgb,255: red,195; green,0; blue,3}, dashed, ultra thick]
    \tikzstyle{connected terminals}=[-, draw=black, dashed, very thick]
    \tikzstyle{P}=[-, draw={rgb,255: red,20; green,154; blue,0}, very thick]
    \tikzstyle{Marked Edges}=[-, draw={rgb,255: red,255; green,207; blue,14}, very thick]

		\node [style=Writing] (0) at (-1.75, 1.5) {};
		\node [style=Writing] (1) at (3.25, 1.5) {};
		\node [style=Writing] (2) at (3.25, -1.5) {};
		\node [style=Writing] (3) at (-1.75, -1.5) {};
		\node [style=black dot] (4) at (-1.25, 1) {};
		\node [style=black dot] (5) at (0, 1.25) {};
		\node [style=black dot] (6) at (0.75, 0.5) {};
		\node [style=black dot] (7) at (-0.75, 0.25) {};
		\node [style=black dot] (8) at (1.5, 1.5) {};
		\node [style=black dot] (9) at (1.5, 0.5) {};
		\node [style=black dot] (10) at (1.75, -0.25) {};
		\node [style=black dot] (11) at (1, 0) {};
		\node [style=black dot] (12) at (-1.5, -0.5) {};
		\node [style=black dot] (13) at (-0.5, -0.75) {};
		\node [style=black dot] (14) at (0.5, -0.75) {};
		\node [style=black dot] (15) at (1.25, -1.5) {};
		\node [style=black dot] (16) at (2.5, -0.5) {};
		\node [style=black dot] (17) at (2.5, 0.5) {};
		\node [style=Writing] (18) at (2.25, 1.25) {$P_1$};
		\node [style=Writing] (19) at (2.5, -1) {$P_2$};
		\node [style=Writing] (20) at (-0.75, 0.75) {$E_1$};
		\node [style=Writing] (21) at (1, 1) {$E_2$};
		\node [style=Writing] (22) at (-0.25, -0.25) {$E_3$};

        \draw [style=witness edge] (4) to (5);
		\draw [style=witness edge] (5) to (6);
		\draw [style=witness edge] (6) to (7);
		\draw [style=witness edge] (7) to (4);
		\draw [style=new witness] (5) to (8);
		\draw [style=new witness] (8) to (9);
		\draw [style=new witness] (9) to (10);
		\draw [style=new witness, in=345, out=165] (10) to (11);
		\draw [style=new witness] (11) to (7);
		\draw [style=connected terminals] (7) to (12);
		\draw [style=connected terminals] (12) to (13);
		\draw [style=connected terminals] (13) to (14);
		\draw [style=connected terminals] (14) to (11);
		\draw [style=P] (13) to (15);
		\draw [style=P] (15) to (16);
		\draw [style=P] (16) to (10);
		\draw [style=Marked Edges] (10) to (17);
		\draw [style=Marked Edges] (17) to (8);
\end{tikzpicture}&
        \end{tabular}
    \end{center}
    \caption{Here we have an example of a "nice" ear decomposition, consisting of ears $E_1,E_2,E_3,P_1$ and $P_2$, each represented with a different colour or edge shape. Note that the only short ears $P_1$ and $P_2$ are open and "pendant". That is, the internal vertices are \emph{only} on their respective ears.}
    \label{fig:pendant}
\end{figure}
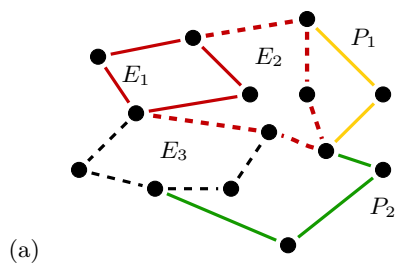

    The authors define two algorithms and take the minimum output of the two.  The first algorithm (see Section 5.3 of~\cite{sebHo2014shorter}) outputs a solution $ALG_1$ which satisfies $|ALG_1| \leq \frac{3}{2}|OPT| - \pi$. The second algorithm simply returns as a solution a nice ear decomposition (which they show exists for a 2VC graph, and can be computed efficiently). Let us call this solution $ALG_2$. We now provide a bound on the size of $ALG_2$.
\begin{align*}
        |ALG_2|& = \sum_{i\geq 2} i \pi_i = 2 \pi_2 + 3 \pi_3 + 4 \pi_4 + \sum_{i\geq 5} i \pi_i \\
        &\leq (\frac{5}{4} +\frac{3}{4}) \pi_2 + (\frac{5}{4} \cdot 2 +\frac{1}{2}) \pi_3 + (\frac{5}{4} \cdot 3 +\frac{1}{4}) \pi_4  +  
        \sum_{i \geq 5} \frac{5}{4} (i-1) \pi_i \\
        &\leq \frac{5}{4}(n-1) + \frac{3}{4}\pi_2 + \frac{1}{2}\pi_3 + \frac{1}{4}\pi_4
        = \frac{5}{4}(n-1) + \frac{1}{4}(\pi_2 + \pi_4) + \frac{1}{2}(\pi_3 + \pi_2)\\
        &\leq \frac{5}{4}(n-1) + \frac{1}{4}\phi(G) + \frac{1}{2}\pi.
\end{align*}

The second to last inequality follows since an ear of length $i$ has $i-1$ internal vertices, and every vertex of the graph but 1 is an internal vertex of exactly one ear.  
   % Consider an ear of length at least 5, thus containing at least 4 internal vertices, and thus there are at most $\frac{5}{4}$ edges per vertex. Thus, the total cost in $ALG_2$ from these ears is at most $\frac{5}{4}(n-1)$. 
   % Consider an ear of length 4, thus containing 3 internal vertices, observe that $4 = 3\frac{5}{4} + \frac{1}{4}$, so we assign a cost of $\frac{5}{4}$ edges per vertex in the ear, leaving an excess cost of $\frac{1}{4}$.    
   % Consider an ear of length 3, thus containing 2 internal vertices, observe that $3= 2\frac{5}{4} + \frac{1}{2}$, so we assign a cost of $\frac{5}{4}$ edges per vertex, with an excess cost of $\frac{1}{2}$. Lastly, consider an ear of length 2, thus containing a single internal vertex, observe that $2 = \frac{5}{4}+ \frac{3}{4}$, so we assign a cost of $\frac{5}{4}$ edges to the single vertex, with an excess cost of $\frac{3}{4}$. Therefore, all vertices that are interior vertices of an ear have a cost of $\frac{5}{4}$ assigned to them, and in the worst case $n-1$ vertices are interior vertices of a nice ear decomposition. For ears of length $2$, $3$, and $4$, we have $\frac{3}{4}\pi_2, \frac{1}{2}\pi_3$,  and $\frac{1}{4}\pi_4$ excess remaining respectively. Therefore, we can see that     
    So the algorithm in \cite{sebHo2014shorter} returns a solution of size $\leq \min\{|ALG_1|,|ALG_2|\} 
        \leq \frac{1}{3}|ALG_1| + \frac{2}{3}|ALG_2| $ which is
    \begin{align*}
     \leq & \frac{1}{3} \left( \frac{3}{2}OPT-\pi \right) + \frac{2}{3} \left(\frac{5}{4}(n-1) + \frac{1}{4}\phi(G) + \frac{1}{2}\pi \right) 
        =\frac{1}{2}OPT + \frac{5}{6}(n-1)+ \frac{1}{6}\phi(G).
    \end{align*}
   To prove our claim it is enough to show that the latter term is bounded by $\frac{4}{3}n+\frac{2}{3}(x-1)$. 
    For this, we need to employ Theorem 5 of \cite{sebHo2014shorter} which states that $n+\phi(G)-1 \leq |OPT|$, and hence 
   $\phi(G)-1 \leq x$. We then get
    \begin{align*}
        &\frac{4}{3}n+\frac{2}{3}(x-1) - \left(\frac{1}{2}OPT + \frac{5}{6}(n-1)+ \frac{1}{6}\phi(G)\right) \\
        =&\frac{4}{3}n+\frac{2}{3}(x-1) - \left(\frac{4}{3}n + \frac{1}{2}x + \frac{1}{6}\phi(G) - \frac{5}{6}\right) 
        =\frac{1}{6}(x -\phi(G) + 1) \ge  0.
    \end{align*}
\end{proof}
We are now ready to prove Theorem~\ref{thm:FECapx}.
% \begin{theorem}
%     $\min\{ALG_1, ALG_2 \} \leq \frac{10}{7}OPT$
% \end{theorem}
\begin{proof}[Proof of Theorem~\ref{thm:FECapx}]
    By Lemma~\ref{lem:FECapx1}, we have that $ALG_1\leq |OPT \cap E_S|+\frac{3}{2}|OPT \cap E_U|$.

    Let $opt(G'')$ be the size of an optimal solution for the 2ECSS instance $G''$ considered by Algorithm 2. It is important to observe that
    $opt(G'') \leq 2|OPT \cap E_S| + |OPT \cap E_U|$. Let us denote $OPT_S :=OPT \cap E_S$ and $OPT_U :=OPT \cap E_U$. 
    Using Lemma~\ref{lem:2ECSSImprovement}, we can see that 
    \begin{align*}
        |ALG_2| &\leq \frac{4}{3}n  + \frac{2}{3}(opt(G'') - n-1) \leq \frac{4}{3}n + \frac{2}{3}(2|OPT_S| + |OPT_U| - n-1)\\
        &=\frac{2}{3}(n-1)+\frac{4}{3}|OPT_S|+\frac{2}{3}|OPT_U|\leq (\frac{4}{3}+\frac{2}{3}) |OPT_S| + \frac{4}{3}|OPT_U|.
    \end{align*}

    Thus our algorithm provides a solutions of size at most $ \min\{|ALG_1|, |ALG_2|\}$ which is 
    \begin{align*}
        \leq &\min \left\{|OPT_S|+\frac{3}{2}|OPT_U|, \left(\frac{4}{3}+\frac{2}{3} \right)|OPT_S| + \frac{4}{3}|OPT_U| \right\} \\
        \leq& \frac{3}{7}(2|OPT_S| + \frac{4}{3}|OPT_U| ) + \frac{4}{7} (|OPT_S|+\frac{3}{2}|OPT_U|) 
        = \frac{10}{7}|OPT_S|  + \frac{10}{7}|OPT_U|.%\approx1.42 OPT.
    \end{align*}
    
\end{proof}

\section{\texorpdfstring{$1+O(1/\sqrt{k})$}{1+O(1/sqrt(k))}-approximation for $k$-FGC}
\label{sec:1kFEC}
The goal of this section is to prove Theorem~\ref{thm:1kFECapx}. $k-$FGC was first tackled in \cite{adjiashvili2022flexible}, where the authors claimed a $(1+O(\frac{1}{k}))$-approximation. Unfortunately, the analysis of their algorithm has a flaw, rendering the proof incorrect, that seems not salvageable. We discuss this in further detail in Appendix~\ref{sec:bug}. We provide in this section an alternate analysis to a similar algorithm used in \cite{adjiashvili2022flexible}, finding a slightly larger approximation factor than the one claimed.

The algorithm employed here is simple. We compute a maximum forest $F_s$ on  $(V, E_S)$. The connected components of $(V, F_s)$ are contracted (remove loops but keep parallel edges) to find $G/F_s \coloneqq G'=(V',E')$. Then, we run one of the known $1+\frac{1}{O(k)}$-approximation algorithms  for $(k+1)$ECSS due to \cite{cheriyan2000approximating,gabow2005improved} on $G'$, finding $F\subseteq E'$. As long as $F$ is not a minimal solution for the $(k+1)$ECSS instance defined on $G'$, we find and remove an edge $e\in F$ such that $F\setminus \{e\}$ is still a solution. We then return ALG:= $F_s \cup F$ as our solution.

% We then apply the best $(k+1)$ECSS approximation algorithm \cite{whoever is best} to compute a $\beta_{k+1}$-approximate $(k+1)$-ECSS solution $F$ on $G/F_s$ and return as our solution $F_s\cup F$.
% \begin{theorem}
% \label{thm:1kapproximation}
%     $|F_s \cup F| \leq \left( 1+O \left( \frac{1}{\sqrt{k}} \right) \right)|OPT|$
% \end{theorem}

% The algorithm presented computes a maximal forest $F_s \subseteq E_s$. The edges of $F_s$ are contracted (this removes loops, but allows parallel edges) to create $G/F_s \coloneqq G'=(V',E')$. Then, we run the best known approximation algorithm for $(k+1)$-ECSS on $G'$, finding $F\subseteq E'$.
% As long as $F$ is not a minimal $(k+1)$-ECSS on $G'$, we find an edge $e\in F$ from $F$ such that $F\setminus \{e\}$ is $(k+1)$-ECSS on $G'$, we remove $e$ from $F$. We then return $F_s \cup F$ as our solution.

To prove Theorem~\ref{thm:1kFECapx}, we need the following lemmas. The first one relies on the fact that a minimal $k$ECSS solution can be partitioned into $k$ forests. As usual,
$n$ denotes $|V|$. 
% {\color{red}It was shown in \cite{gabow2005improved} that a minimal $k$-ECSS solution can be be partitioned into $k$ spanning  forests. Thus $F$ contains at most $(k+1) |V(G)/F_s|$ edges.}
\begin{lemma}[\cite{gabow2005improved}]
\label{lem:kconnectedupperbound}
   For a minimal solution $F$ to a given $k$ECSS instance $G$,  $|F| \leq nk$.
\end{lemma}
In every $k$-edge-connected graph the degree of each vertex is at least $k$, hence:
\begin{lemma}
\label{lem:kconnectedlowerbound}
    For an optimal solution $OPT$ of a kECSS instance $G$, $ \frac{nk}{2}\leq |OPT|$.
\end{lemma}
The following lemma is proven in \cite{adjiashvili2022flexible}.
\begin{lemma}[\cite{adjiashvili2022flexible}]
\label{lem:contractsafeunsafefeasible}
    Let $H\subseteq G$ be a feasible solution to a $k$-FGC instance. Then $H/(E_s \cap E(H))$ is $(k+1)$-edge-connected.
\end{lemma}
% Finally, we have the following useful lemma that will be a useful tool in proving Theorem~\ref{thm:1kapproximation}.
% \begin{lemma}
% \label{claim:optTbound}
%      $n - |OPT_s| \leq \frac{2k}{k-1}\left( n - |F_s| \right)$.
% \end{lemma}
We are now ready to prove Theorem~\ref{thm:1kFECapx}.
\begin{proof}[Proof of Theorem~\ref{thm:1kFECapx}]
    Fix an optimal solution $OPT$, which we decompose as $OPT = OPT_s \cup OPT_u$, where $OPT_s\subseteq E_S$, and $OPT_u \subseteq E_U$. 
    Using Lemma~\ref{lem:kconnectedupperbound} we have $|F| \leq (n-|F_s|)(k+1)$.
    %\begin{align}
     %   |F| \leq (n-|F_s|)(k+1)
    %\end{align}
   This observation, together with Lemmas~\ref{lem:kconnectedlowerbound} and~\ref{lem:contractsafeunsafefeasible} and the fact that $|OPT_s| \leq |F_s|$ gives
    \begin{align}
        |OPT| &\geq |OPT_s| + (n-|OPT_s|)\frac{k+1}{2} \\
        \label{inq:3}
        & \geq |F_s| + (n-|F_s|)\frac{k+1}{2} = n\frac{(k+1)}{2} + \left(1-\frac{k+1}{2}\right)|F_s|.
    \end{align}
    % Using Lemmas~\ref{lem:kconnectedlowerbound} and ~\ref{lem:kconnectedupperbound}, we can show the following useful inequalities.
    % \begin{align}
    %     |ALG| &\leq |F_s| + (n-|F_s|) \beta_k \\
    %     |ALG| &\leq |F_s| + (n-|F_s|) k \\
    % \end{align}
    %%%%%%%%%%%%%%%%%%%%%%%
    % corrected the usage of $k$ up to here
    %%%%%%%%%%%%%%%%%%%%%%%
    Let $\beta_{k+1}$ be the best approximation factor for $(k+1)$ECSS. We distinguish two cases. 
    \begin{itemize}
        \item \textit{Case 1:} $|F_s| \leq n\frac{\sqrt{k}+2}{\sqrt{k}+3}$. Note that
            $|ALG| \leq n\frac{\sqrt{k}+2}{\sqrt{k}+3} + \beta_{k+1}|OPT|$.
            By
        % \todo{missing note that $opt_u$ is a feasible solution on $G/F_s$. or at least that it is larger than the optimal of $G/F_s$.}
        applying inequality~\ref{inq:3} and the assumption of this case, we get
        \begin{align*}
            |OPT| &\geq n\frac{(k+1)}{2} + \left(1-\frac{k+1}{2}\right)|F_s|
            \geq  n\left(\frac{k+1}{2} + \left(1 - \frac{k+1}{2}\right)\frac{\sqrt{k}+2}{\sqrt{k}+3}\right).
            % \\
            % \Rightarrow& n \leq \frac{|OPT|}{\frac{k}{2} + \left(1 - \frac{k}{2}\right)\frac{\sqrt{k}}{\sqrt{k}+1}}
        \end{align*}
        Where the second inequality follows since $1-\frac{k+1}{2} \leq 0$ for $k\geq 1$. Therefore, we have 
        \begin{align*}
            |ALG| 
            &\leq n\frac{\sqrt{k}+2}{\sqrt{k}+3} + \beta_{k+1}|OPT| 
            \leq \frac{|OPT|(\sqrt{k}+2) }{(\sqrt{k}+3) \left(\frac{k+1}{2} + \left(1 - \frac{k+1}{2}\right)\frac{\sqrt{k}+2}{\sqrt{k}+3}\right)} + \beta_{k+1}|OPT|\\
            &= \frac{|OPT|2(\sqrt{k}+2)}{k+2\sqrt{k} +5}  + \beta_{k+1}|OPT|
            < \left( \frac{2}{\sqrt{k}} + \beta_{k+1} \right) |OPT|.
            % & = \frac{|OPT|\sqrt{k}}{(\sqrt{k}+1)\frac{(k+1)(\sqrt{k}+1)+(1-k)\sqrt{k}}{2(\sqrt{k}+1)}}+ \beta_{k+1}|OPT|
            % =\frac{2|OPT|\sqrt{k}}{(k+2\sqrt{k}+1)} + \beta_{k+1}|OPT_u|\\
            % &= \frac{2|OPT|\sqrt{k}}{(\sqrt{k}+1)^2}+ \beta_{k+1}|OPT_u| 
            % < \left( \frac{2}{\sqrt{k}} + \beta_{k+1} \right) |OPT|   %= \left(1 + o(1)\right) |OPT|
        \end{align*}

        \item \textit{Case 2:} $|F_s| > n\frac{\sqrt{k}+2}{\sqrt{k}+3}$.
        We here use a tighter bound on $ALG$. Namely,
            $|ALG| \leq n\frac{\sqrt{k}+2}{\sqrt{k}+3} + \beta_{k+1}|OPT_U|$. This holds as $OPT_U$ is a feasible solution to the $k$ECSS instance we obtain contracting $F_s$: this is because $F_s$ is a maximal forest, hence vertices of each component of $(V,OPT_S)$ are a subset of vertices of some component of $(V,F_S)$. Therefore
        \begin{align*}
            \frac{ALG}{OPT} \leq \frac{|F_s| + \beta_{k+1} |OPT_u|}{|OPT_s| + |OPT_u|} \leq \max \left\{\frac{|F_s|}{|OPT_s|}, \beta_{k+1} \right\}.
        \end{align*}
        We need to bound $\frac{|F_s|}{|OPT_s|}$. 
        By applying Lemmas~\ref{lem:kconnectedupperbound} and \ref{lem:kconnectedlowerbound} we see that 
        \begin{align*}
            & |F_s| + (n-|F_s|)(k+1) 
            \geq |ALG| \geq |OPT| 
            \geq |OPT_s| + (n-|OPT_s|)\frac{k+1}{2} \\ 
            &= n+(n-|OPT_s|)\frac{k-1}{2}
            \Rightarrow  \frac{2k}{k-1}(n-|F_s|)\geq n-|OPT_s|.
        \end{align*}
        
        Using this inequality with the fact that $\frac{2k}{k-1}\le 4$ when $k\ge 2$, as well as the case assumption we find 
        \begin{align*}
            n - |OPT_s| 
            &\leq \frac{2k}{k-1}(n - |F_s|) 
             < n\frac{2k}{k-1} \left( 1- \frac{\sqrt{k}+2}{\sqrt{k}+3}\right)
            \le 4n \left( 1- \frac{\sqrt{k}+2}{\sqrt{k}+3}\right).
        \end{align*}
        Therefore, we have
        \begin{align}
             |OPT_s| &\geq n\left(1 - 4 + \frac{4(\sqrt{k}+2)}{\sqrt{k}+3} \right) = n\left(1 - \frac{4}{\sqrt{k}+3} \right).\label{inq:case2}
        \end{align}
        The last term above is positive for all $k\geq 2$. Using Inequality~\ref{inq:case2} and the fact that $|F_s| < n$:
        \begin{align*}
            \frac{|F_s|}{|OPT_s|} < \frac{n}{n \left( 1- \frac{4}{\sqrt{k} + 3} \right)} = \frac{1}{ 1- \frac{4}{\sqrt{k} + 3} } = \frac{\sqrt{k}+3}{\sqrt{k}-1} = 1 + \frac{4}{\sqrt{k}-1}.
        \end{align*}
        Therefore, in this case we have 
        \begin{align*}
            \frac{|ALG|}{|OPT|}\leq \max \left\{ \frac{|F_s|}{|OPT_s|}, \beta_{k+1} \right\} 
            \leq \max \left\{ 1 + O \left( \frac{1}{\sqrt{k}} \right) , \beta_{k+1} \right\}.
        \end{align*}
    \end{itemize}
    Combining the inequalities from Case 1 and 2 we find,
    \begin{align*}
        \frac{|ALG|}{|OPT|}\leq \max \left\{  \frac{2}{\sqrt{k}} + \beta_{k+1} , 1 + O \left( \frac{1}{\sqrt{k}} \right) , \beta_{k+1}\right\} = 1 + O\left( \frac{1}{\sqrt{k}} \right).
    \end{align*}
    Where the last inequality above follows by applying the algorithm in \cite{cheriyan2000approximating,gabow2009approximating}, for which  we have that $\beta_{k+1} \leq 1 + \frac{1}{O(k)} $ and our claim follows.
\end{proof}

\bibliography{refs}

\begin{appendix}
    \section{Proof of useful Lemmas}
% \subsection{Proof of Lemma~\ref{lem:blockbound} }
% \label{sec:blockbound}
% \begin{proof}
%     The following arguments are based arguments in . 
%     \begin{enumerate}
%         \item Assume for contradiction that distinct blocks $B_1$ and $B_2$ have at least two common vertices. By definition of a block, when a vertex $v\in V(B_i)$, $i=1,2$ is deleted, $B_i\setminus \{v\}$ is still connected. Therefore, there is a path from every vertex that remains to every vertex of $B_1\cap B_2$ that remains. Since the blocks have at least two vertices in common, $|V(B_1\cap B_2) \setminus\{v\}| \geq 1$, therefore, $B_1\cup B_2$ is connected after the deletion of a vertex, hence is a block. Contradicting the maximality of $B_1$ and $B_2$.
        
%         \item This claim immediately since blocks are defined as maximal sets. 
%         \item This claim immediately follows since if there is no cycle containing both $e_1$ and $e_2$, then $B_i$ contains a cut vertex.
%         \item First, observe that adding an edge to a connected graph will never increase the number of blocks. To see this, observe that if an edge $uv$ is added to a connected graph, every edge between $u$ and $v$ are now part of a block.

%         Consider a spanning tree $T$ of $G$. It is not hard to see that $T$ has $n-1$ blocks, where $n\coloneqq |V|$. We can add every edge of $E(G) \backslash E(T)$ to $T$ one by one, and observe after each edge is added, we do not increase the number of blocks, thus the number of blocks in $G$ is at most the number of blocks in $T$.
%     \end{enumerate}
% \end{proof}

\subsection{Proof of Lemma~\ref{lem:reduceblocks}}
\label{sec:reduceblocks}
\begin{proof}
    $H$ is a subgraph of $G$, therefore every block of $H$ is either a block in $G$, or a subgraph of a block in $G$. Since $B(H) > B(G)$, there must exist a block $B$ of $G$ that contains at least two blocks of $H$. Let $B_1,\dots, B_k$ be blocks of $H$ that are subgraphs of $B$ that all share the same vertex, $c \in V(B)$.

    We wish to show that there are blocks $B_i$ and $B_j$, $i,j\in \{1,\dots, k\}$, and edge $e\in E(G)\backslash E(H)$ between $B_i$ and $B_j$, that does not contain $c$ as an endpoint. That is, $B_i \cup B_j \cup \{e\}$ is a block. 

    Pairwise, any edges with $c$ as an endpoint will be part of a cycle in $B$ by Property 2 of Lemma~\ref{lem:blockbound}. So we pick edges $e_1$ and $e_2$ arbitrarily of these. By construction there is block $B_1$ that contains $e_1$ and block $B_2$ that contains $e_2$. Using Property 2 of Lemma~\ref{lem:blockbound}, there is a cycle $C$ in $B$ that contains both $e_1$ and $e_2$. Furthermore, the path $C-\{e_1,e_2\}$ starts in $B_1$ and ends in $B_2$. Therefore, there is an edge connecting a pair of blocks in $B_1,\dot, B_k$.
\end{proof}
    \section{Missing Proofs of Section~\ref{sec:11fvc}}

\subsection{Proof of Lemma~\ref{lem:removeforiddencycles}}
\label{sec:removeforbiddencycles}

\begin{lemma}
\label{lem:2-vertex-connected-node}
    Let $G=(V,E)$ be a graph and let $x$ be a safe cut-vertex of $G$. Let $V_1$ be the vertex set of one connected component of $G[V\setminus\{x\}]$ and let $V_2=V\setminus (\{V_1\}\cup x)$. 
    Now consider an arbitrary subset $E'$ of $E$. For $i\in \{1,2\}$, we define $E'_i$  as the set of edges in $E'$ that have both endpoints in $V_i\cup \{x\}$.
    Then $E'$ is a feasible FVC solution for $G$ if and only if $E'_1$ and $E'_2$ are feasible FVC solutions for $G_1:=G[V_1\cup\{x\}]$ and  $G_2:=G[V_1\cup\{x\}]$, respectively.
\end{lemma}
\begin{proof}
    Clearly, we can observe that $E'_1$ and $E'_2$ partition $E'$.
    Assume $E'_i$ is not a feasible FVC solution for $G_i$ for some $i\in \{1,2\}$. 
    First, suppose $(V_i\cup \{x\},E'_i)$ is disconnected. Then, there exists a vertex $v$ in $V_i$ that has no paths to $x$ in  $(V_i\cup \{x\},E'_i)$. Therefore, $v$ has no paths to $x$ in  $(V,E')$ and hence $E'$ is not a feasible FVC solution for $G$. So suppose that $(V_i\cup \{x\},E'_i)$ is connected, thus, as $E'_i$ is not a feasible FVC solution for $G_i$, then $(V_i\cup \{x\},E'_i)$ must contain an unsafe cut-vertex, say $u\in V_i$. Now $u$ is also a cut-vertex in $(V,E')$, hence the claim.

    Now, we assume for each $i\in \{1,2\}$, $E'_i$ is a feasible FVC solution for $G_i$. Clearly $(V,E')$ is connected. Furthermore, if we remove any unsafe vertex $v_i\in V_i$ from $(V_i\cup \{x\},E'_i)$, it still remains connected. Therefore, by removing any unsafe vertex from $(V,E')$, it will remain connected. Hence $E'$ is a feasible FVC solution for $G$.
    %Therefore, to find a $\beta$-approximation for the more general instance, we can find an  $\beta$-approximation for the $2$VC instance. 
\end{proof}
For convenience we repeat the definition of forbidden cycles here.
\begin{definition}[Forbidden Cycle]
    We say that a $4$-cycle $C$ in $G$ is a forbidden cycle if $C$ has two vertices $w$ and $z$ such that $wz\notin E(C)$ and $\deg_G(w)=\deg_G(z)=2$.
\end{definition}

\begin{lemma}
\label{lem:forbiddencyclesUnsafe}
    Assume $G$ is $2$VC and has at least five vertices. Let $C:=uwvz$ be a forbidden cycle of $G$ as described above such that $\deg_G(w)=\deg_G(z)=2$ and $u$ and $v$ are unsafe. Then any optimal solution for $G$ can be decomposed into an optimal solution for $G\setminus\{w\}$ and the edges $\{uw,wv\}$.
\end{lemma}
\begin{figure}[t]
    \begin{center}
        \begin{tabular}{c c}
            (a) \begin{tikzpicture}
    
    % Node styles
    % \tikzset{black dot/.style={draw=black, very thick, circle,minimum size=0pt, inner sep=1pt, outer sep=1pt,fill=black}}
    \tikzstyle{black dot}=[fill=black, draw=black, shape=circle]
    \tikzset{terminal/.style={draw=black,  thick,minimum size=0pt, inner sep=2.5pt, outer sep=1pt}}
    \tikzset{P node/.style={fill={rgb,255: red,20; green,154; blue,0}, draw={rgb,255: red,20; green,154; blue,0}, circle, minimum size=0pt,inner sep=1pt, outer sep=1pt}}
    \tikzset{Writing/.style={shape=circle} }
    \tikzstyle{empty circle}=[fill=white, draw=black, shape=circle]

    % Edge styles
    \tikzstyle{witness edge}=[-, draw={rgb,255: red,195; green,0; blue,3}, very thick]
    \tikzstyle{T edges}=[-, thick]
    \tikzstyle{Fat edge}=[-, ultra thick]
    \tikzstyle{new witness}=[-, draw={rgb,255: red,195; green,0; blue,3}, dashed, ultra thick]
    \tikzstyle{connected terminals}=[-, draw=black, dashed, very thick]
    \tikzstyle{P}=[-, draw={rgb,255: red,20; green,154; blue,0}, very thick]
    
		\node [style=black dot] (1) at (-0.75, -0.5) {};
		\node [style=black dot] (2) at (-1.75, 0.5) {};
		\node [style=terminal] (3) at (-0.75, 1.5) {};
		\node [style=empty circle] (4) at (0.25, 0.5) {};
		\node [style=black dot] (5) at (0.25, 1.75) {};
		\node [style=Writing] (13) at (0.5, 0.5) {$w$};
		\node [style=Writing] (16) at (0.15, 2) {$x$};
		\node [style=Writing] (16) at (-2.25, 0.5) {$z$};
		\node [style=Writing] (16) at (-1.25, 1.45) {$u$};
		\node [style=Writing] (16) at (-1.25, -0.35) {$v$};
		\node [style=Writing] (17) at (-1.25, -0.75) {};
		\node [style=Writing] (18) at (-0.75, -0.75) {};
		\node [style=Writing] (19) at (-0.25, -0.75) {};
		\draw [style=connected terminals] (4) to (1);
		\draw [style=T edges] (1) to (2);
		\draw [style=connected terminals] (3) to (4);
		\draw [style=T edges] (2) to (3);
		\draw [style=T edges] (3) to (5);
		\draw [style=connected terminals] (17) to (1);
		\draw [style=connected terminals] (1) to (18);
		\draw [style=connected terminals] (1) to (19);
\end{tikzpicture}
            &
            (b)
\begin{tikzpicture}
    
    % Node styles
    % \tikzset{black dot/.style={draw=black, very thick, circle,minimum size=0pt, inner sep=1pt, outer sep=1pt,fill=black}}
    \tikzstyle{black dot}=[fill=black, draw=black, shape=circle]
    \tikzset{terminal/.style={draw=black,  thick,minimum size=0pt, inner sep=2.5pt, outer sep=1pt}}
    \tikzset{P node/.style={fill={rgb,255: red,20; green,154; blue,0}, draw={rgb,255: red,20; green,154; blue,0}, circle, minimum size=0pt,inner sep=1pt, outer sep=1pt}}
    \tikzset{Writing/.style={shape=circle} }
    \tikzstyle{empty circle}=[fill=white, draw=black, shape=circle]

    % Edge styles
    \tikzstyle{witness edge}=[-, draw={rgb,255: red,195; green,0; blue,3}, very thick]
    \tikzstyle{T edges}=[-, thick]
    \tikzstyle{Fat edge}=[-, ultra thick]
    \tikzstyle{new witness}=[-, draw={rgb,255: red,195; green,0; blue,3}, dashed, ultra thick]
    \tikzstyle{connected terminals}=[-, draw=black, dashed, very thick]
    \tikzstyle{P}=[-, draw={rgb,255: red,20; green,154; blue,0}, very thick]
    
		\node [style=black dot] (1) at (-0.75, -0.5) {};
		\node [style=terminal] (2) at (-1.75, 0.5) {};
		\node [style=black dot] (3) at (-0.75, 1.5) {};
		\node [style=empty circle] (4) at (0.25, 0.5) {};
		\node [style=black dot] (5) at (0.25, 1.75) {};
		\node [style=Writing] (13) at (0.5, 0.5) {$w$};
		\node [style=Writing] (16) at (0.15, 2) {$x$};
		\node [style=Writing] (16) at (-2.25, 0.5) {$z$};
		\node [style=Writing] (16) at (-1.25, 1.45) {$u$};
		\node [style=Writing] (16) at (-1.25, -0.35) {$v$};
		\node [style=Writing] (17) at (-1.25, -0.75) {};
		\node [style=Writing] (18) at (-0.75, -0.75) {};
		\node [style=Writing] (19) at (-0.25, -0.75) {};
		\draw [style=connected terminals] (4) to (1);
		\draw [style=T edges] (1) to (2);
		\draw [style=connected terminals] (3) to (4);
		\draw [style=T edges] (2) to (3);
		\draw [style=T edges] (3) to (5);
		\draw [style=connected terminals] (17) to (1);
		\draw [style=connected terminals] (1) to (18);
		\draw [style=connected terminals] (1) to (19);
\end{tikzpicture}
        \end{tabular}
    \end{center}
    \caption{As in the statement and proof of Lemma~\ref{lem:forbiddencyclesUnsafe}, we are given a forbidden cycle $C$, with unsafe vertices $u$ and $v$, and degree 2 vertices $w$ and $z$. The edges of $F'$  are shown with solid edges, and $F$ is both the solid lines and the dashed edges incident to $w$. Since $G\geq 5$, there is an additional vertex $x\notin V(C)$. 
    We wish to show that $F'$ neither $u$ nor $z$ is a cut vertex of $(V\setminus\{w\}, F')$ showing that if so, then that vertex will be cut vertex separating $v$ from $x$ in .
    % In both figures we show parts of the forbidden cycle $C$ (as in the statement of Lemma~\ref{lem:forbiddencyclesUnsafe} ) in both $G\setminus \{w\}$ (represented by the dashed lines), and $\{uw,wv\}$ (represented by the solid lines). Since $G$ has at least $5$ vertices, there is an additional vertex $x$ not in $C$ that is adjacent to $u$. 
    We consider cases if $u$ or $z$ are cut vertices of $(V\setminus\{w\}, F')$:
    \\
    (a) $u$ is a cut vertex (shown as a square), separating $x$ from $z$ and $v$, and clearly even with $vw$, $x$ is still separated from these vertices.\\
    (b) $z$ is a cut vertex(shown as a square), separating $v$ from $u$, and in particular, separating $v$ from $x$. It is clear that in this case $u$ is again a cut vertex.}
    \label{fig:forbiddencycleunsafe}
\end{figure}
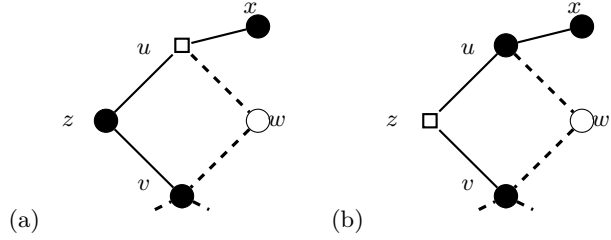
\begin{proof}
    First, we observe that since $G$ has at least $5$ vertices, and is $2$VC, there must be a path from $u$ to $v$ that does not contain any edge of $C$. Furthermore, since $u$ and $v$ are unsafe, it is clear that any feasible solution must contain every edge of $C$, otherwise, one of $u$ or $v$ would be an unsafe cut vertex. We wish to show that for every minimal feasible solution $F$,
    % such that $uw\in F$, we can either find another feasible solution $F'$ not containing $uw$ with $|F| = |F'|$, or 
    $F' \coloneqq F \backslash\{uw,vw\}$ is a feasible solution for the instance $G\backslash\{w\}$.

    Since $w$ is adjacent only to $u$ and $v$ in $G$,  and $F$ already contains edges $uz$ and $vz$, the only vertices that can become cut vertices when we remove $w$ from $F$ are vertices in $C$, namely, $u,v,$ and $z$. We demonstrate by cases that none of these can become cut vertices of $G':=(V\setminus\{w\},F')$. See Figure~\ref{fig:forbiddencycleunsafe}.

    \textbf{Case: }Assume for contradiction that at least one of $u$ or $v$ are cut vertices of $(V\setminus\{w\},F')$ (without loss of generality, say $u$ is the cut-vertex). That is, $G'\backslash\{u\}$ has at least two connected components, where one component contains $v$ and another contains some vertex $x\in V\backslash V(C)$. Clearly, this means that in $(V,F)$, every path from $v$ to $x$ contains $u$. Hence $F$ is not feasible, a contradiction.

    \textbf{Case: }So assume for contradiction that $z$ is a cut vertex in $G'$. Denote the components of $G'\backslash\{z\}$ as $F_u'$ and $F_v'$, which are the components containing $u$ and $v$, respectively. Since $G$ has at least five vertices, without loss of generality we can assume that there is vertex $x\in F_u'$ where $x\notin V(C)$. As $z$ has degree two, we can see that any path in $F'$ from $x$ to $v$ contains $u$. Therefore, $u$ is a cut vertex in $F'$, which leads us to the previous case to derive a contradiction.
\end{proof}

\begin{lemma}\label{lem:forbiddencycles}
    Assume $G$ is $2$VC and has at least five vertices. Let $C:=uwvz$ be a forbidden cycle of $G$ as described above such that $\deg_G(w)=\deg_G(z)=2$ and $v$ is safe. W.l.o.g assume that if $w$ is a safe vertex then $z$ is also a safe vertex. Then  there exists an optimal solution that does not contain the edge $uw$.
\end{lemma}
\begin{proof}
    Again, we observe that since $G$ has at least $5$ vertices, and is $2$VC, there must be a path from $u$ to $v$ that does not contain any edge of $C$. We wish to show that for every feasible FVC solution $F$ such that $uw\in F$, we can find a feasible FVC solution $F'$ not containing $uw$ with $|F| \ge |F'|$. In particular, given an optimal solution, we will be able to find an optimal solution that does not contain $uw$.
    % Assume that every optimal solution in $G$ contains the edge $uv$. Consider a fixed optimal solution, $OPT$. If $OPT$ does not contain $uw$ we are done. So assume that $OPT$ contains $uw$.
    
    Assume $vw\notin F$. Since $F$ is a feasible FVC solution, $uw\in F$. In this case  $F':=F\cup\{vw\}\setminus\{uw\}$ is a feasible FVC solution as $v$ is a safe vertex. Thus we can assume that $vw\in F$. 
    
    Also $F$ must have an edge incident on $z$. So far, we realized that $F$ has $uw,vw$ and at least one of the edges in $\{zu,zv\}$.
     See Figure~\ref{fig:forbiddencyclessafe}.(a)
    
    Now, let $F':=F\cup C\setminus\{uw\}$. Since $F$ contains at least three edges of $C$, we have $|F|-1\leq |F'|\le |F|$. Furthermore for any vertex $x\in V\setminus\{u,z\}$, $x$ is a cut-vertex in $F'$ if $x$ is a cut-vertex in $F$.
    Therefore $F'$ is a feasible FVC solution satisfying our requirements, unless one of $z$ or $u$ is an unsafe cut-vertex of $(V,F')$. 
    
    Assume that $z$ is an unsafe cut-vertex of $(V,F')$. 
    Then there must be exactly one path in $(V,F')$ from $u$ to $v$, namely, the path $u,z,$ and $v$. Furthermore, as $z$ is unsafe then $w$ is also unsafe by the assumption of this Lemma on $z$ and $w$.

    % Let $F_z'$ be the set of edges in $F'$ that are incident to $z$
    As any path from $u$ to $v$ in $(V,F')$ (and hence in $F$) contains at least one edge of $C$, $F$ must contain all edges of $C$; otherwise, since $z$ and $w$ are unsafe then $F$ is not a feasible FVC solution, a contradiction. This implies that $|F|=|F'|+1$. Let $F_u'$ and $F_v'$ denote the two connected components of $(V\backslash\{z\},F'\{uz,vz\})$.

    Since $G$ has at least $5$ vertices, there is a vertex $x\notin V(C)$ in a component of $(V\backslash\{z\},F'\{uz,vz\})$, without loss of generality say $x\in F_u'$. Since $G$ is $2$VC, there is a path $P$ from $x$ to $v$ that does not contain $u$, as otherwise $u$ would be a cut vertex of $G$. Since there are only two components of $(V\backslash\{z\},F'\{uz,vz\})$, there is an edge $e\in E(P)$ with one endpoint in $F_u'$ and the other in $F_v'$. See Figure~\ref{fig:forbiddencyclessafe}.(b). Therefore, $F'\cup\{e\}$ is a feasible FVC solution that contain $uw$.

    Lastly, suppose $u$ is an unsafe cut vertex of $(V,F')$, we only need to observe that since $u$ is a cut vertex, there is some vertex $x$ that is in a separate component from $v$. And the argument is similar to the case where $z$ is an unsafe cut vertex.
\end{proof}
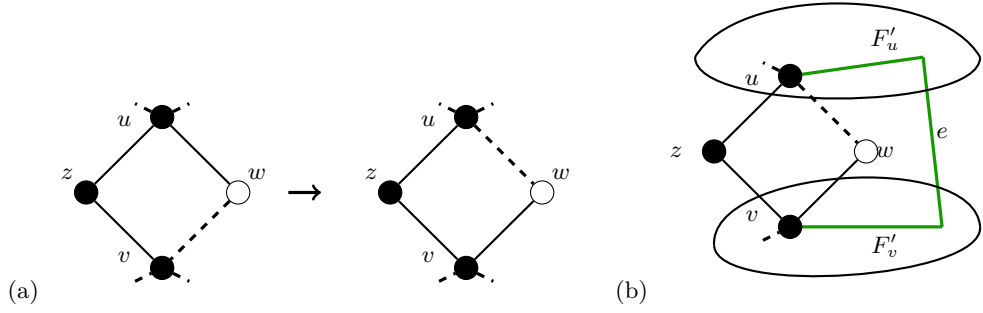
\begin{figure}[t]
    \begin{center}
        \begin{tabular}{c c}
            (a) \begin{tikzpicture}
    
    % Node styles
    % \tikzset{black dot/.style={draw=black, very thick, circle,minimum size=0pt, inner sep=1pt, outer sep=1pt,fill=black}}
    \tikzstyle{black dot}=[fill=black, draw=black, shape=circle]
    \tikzset{terminal/.style={draw=black,  thick,minimum size=0pt, inner sep=2.5pt, outer sep=1pt}}
    \tikzset{P node/.style={fill={rgb,255: red,20; green,154; blue,0}, draw={rgb,255: red,20; green,154; blue,0}, circle, minimum size=0pt,inner sep=1pt, outer sep=1pt}}
    \tikzset{Writing/.style={shape=circle} }
    \tikzstyle{empty circle}=[fill=white, draw=black, shape=circle]

    % Edge styles
    \tikzstyle{witness edge}=[-, draw={rgb,255: red,195; green,0; blue,3}, very thick]
    \tikzstyle{T edges}=[-, thick]
    \tikzstyle{Fat edge}=[-, ultra thick]
    \tikzstyle{new witness}=[-, draw={rgb,255: red,195; green,0; blue,3}, dashed, ultra thick]
    \tikzstyle{connected terminals}=[-, draw=black, dashed, very thick]
    \tikzstyle{P}=[-, draw={rgb,255: red,20; green,154; blue,0}, very thick]
    \tikzstyle{arrow}=[<-, very thick]
    
		\node [style=black dot] (1) at (-0.75, -0.5) {};
		\node [style=black dot] (2) at (-1.75, 0.5) {};
		\node [style=black dot] (3) at (-0.75, 1.5) {};
		\node [style=empty circle] (4) at (0.25, 0.5) {};
		\node [style=Writing] (13) at (0.5, 0.75) {$w$};
		\node [style=Writing] (16) at (-2, 0.75) {$z$};
		\node [style=Writing] (16) at (-1.25, 1.45) {$u$};
		\node [style=Writing] (16) at (-1.25, -0.35) {$v$};
		\node [style=Writing] (17) at (-1.25, -0.75) {};
		\node [style=Writing] (18) at (-0.75, -0.75) {};
		\node [style=Writing] (19) at (-0.25, -0.75) {};
		\node [style=Writing] (20) at (-0.75, 1.75) {};
		\node [style=Writing] (21) at (-1.25, 1.75) {};
		\node [style=Writing] (22) at (-0.25, 1.75) {};
		\node [style=black dot] (23) at (3.25, -0.5) {};
		\node [style=black dot] (24) at (2.25, 0.5) {};
		\node [style=black dot] (25) at (3.25, 1.5) {};
		\node [style=empty circle] (26) at (4.25, 0.5) {};
		\node [style=Writing] (27) at (4.5, 0.75) {$w$};
		\node [style=Writing] (28) at (2, 0.75) {$z$};
		\node [style=Writing] (29) at (2.75, 1.45) {$u$};
		\node [style=Writing] (30) at (2.75, -0.35) {$v$};
		\node [style=Writing] (31) at (2.75, -0.75) {};
		\node [style=Writing] (32) at (3.25, -0.75) {};
		\node [style=Writing] (33) at (3.75, -0.75) {};
		\node [style=Writing] (34) at (3.25, 1.75) {};
		\node [style=Writing] (35) at (2.75, 1.75) {};
		\node [style=Writing] (36) at (3.75, 1.75) {};
		\node [style=Writing] (37) at (0.75, 0.5) {};
		\node [style=Writing] (38) at (1.5, 0.5) {};
		\draw [style=connected terminals] (4) to (1);
		\draw [style=T edges] (1) to (2);
		\draw [style=T edges] (3) to (4);
		\draw [style=T edges] (2) to (3);
		\draw [style=connected terminals] (17) to (1);
		\draw [style=connected terminals] (1) to (18);
		\draw [style=connected terminals] (1) to (19);
		\draw [style=connected terminals] (3) to (20);
		\draw [style=connected terminals] (3) to (22);
		\draw [style=connected terminals] (3) to (21);
		\draw [style=connected terminals] (1) to (18);
		\draw [style=T edges] (26) to (23);
		\draw [style=T edges] (23) to (24);
		\draw [style=connected terminals] (25) to (26);
		\draw [style=T edges] (24) to (25);
		\draw [style=connected terminals] (31) to (23);
		\draw [style=connected terminals] (23) to (32);
		\draw [style=connected terminals] (23) to (33);
		\draw [style=connected terminals] (25) to (34);
		\draw [style=connected terminals] (25) to (36);
		\draw [style=connected terminals] (25) to (35);
		\draw [style=connected terminals] (23) to (32);
		\draw [style=arrow] (38) to (37);
\end{tikzpicture}
            &

            (b)
\begin{tikzpicture}
    
    % Node styles
    % \tikzset{black dot/.style={draw=black, very thick, circle,minimum size=0pt, inner sep=1pt, outer sep=1pt,fill=black}}
    \tikzstyle{black dot}=[fill=black, draw=black, shape=circle]
    \tikzset{terminal/.style={draw=black,  thick,minimum size=0pt, inner sep=2.5pt, outer sep=1pt}}
    \tikzset{P node/.style={fill={rgb,255: red,20; green,154; blue,0}, draw={rgb,255: red,20; green,154; blue,0}, circle, minimum size=0pt,inner sep=1pt, outer sep=1pt}}
    \tikzset{Writing/.style={shape=circle} }
    \tikzstyle{empty circle}=[fill=white, draw=black, shape=circle]

    % Edge styles
    \tikzstyle{witness edge}=[-, draw={rgb,255: red,195; green,0; blue,3}, very thick]
    \tikzstyle{T edges}=[-, thick]
    \tikzstyle{Fat edge}=[-, ultra thick]
    \tikzstyle{new witness}=[-, draw={rgb,255: red,195; green,0; blue,3}, dashed, ultra thick]
    \tikzstyle{connected terminals}=[-, draw=black, dashed, very thick]
    \tikzstyle{P}=[-, draw={rgb,255: red,20; green,154; blue,0}, very thick]
    
		\draw [style=P] (-0.75, -0.5) to (1.25, -0.5);
		\draw [style=P] (-0.75, 1.5) to (1, 1.75);
		\draw [style=P] (1.25, -0.5) to (1, 1.75);
		\node [style=black dot] (1) at (-0.75, -0.5) {};
		\node [style=black dot] (2) at (-1.75, 0.5) {};
		\node [style=black dot] (3) at (-0.75, 1.5) {};
		\node [style=empty circle] (4) at (0.25, 0.5) {};
		\node [style=Writing] (13) at (0.5, 0.5) {$w$};
		\node [style=Writing] (16) at (-2.25, 0.5) {$z$};
		\node [style=Writing] (16) at (-1.25, 1.45) {$u$};
		\node [style=Writing] (16) at (-1.25, -0.35) {$v$};
		\node [style=Writing] (17) at (-1.25, -0.75) {};
		\node [style=Writing] (18) at (-0.75, -0.75) {};
		\node [style=Writing] (19) at (1.25, -0.5) {};
		\node [style=Writing] (20) at (-0.75, 1.75) {};
		\node [style=Writing] (21) at (-1.25, 1.75) {};
		\node [style=Writing] (22) at (1, 1.75) {};
		\node [style=Writing] (23) at (-2, 1.75) {};
		\node [style=Writing] (24) at (1.75, 1.75) {};
		\node [style=Writing] (25) at (-1.75, -0.75) {};
		\node [style=Writing] (26) at (1.75, -0.5) {};
		\node [style=Writing] (27) at (1.25, 0.75) {$e$};
		\node [style=Writing] (28) at (0.5, 2) {$F_u'$};
		\node [style=Writing] (29) at (0.5, -0.75) {$F_v'$};
		\draw [style=T edges] (4) to (1);
		\draw [style=T edges] (1) to (2);
		\draw [style=connected terminals] (3) to (4);
		\draw [style=T edges] (2) to (3);
		\draw [style=connected terminals] (17) to (1);
		\draw [style=connected terminals] (1) to (18);
		\draw [style=connected terminals] (3) to (20);
		\draw [style=connected terminals] (3) to (21);
		\draw [style=connected terminals] (1) to (18);
		\draw [style=T edges, bend left=60, looseness=0.75] (-2, 1.75) to (1.75, 1.75);
		\draw [style=T edges, bend right=90, looseness=0.50] (-2, 1.75) to (1.75, 1.75);
		\draw [style=T edges, bend left=90, looseness=0.75] (-1.75, -0.75) to (1.75, -0.5);
		\draw [style=T edges, bend right=90, looseness=0.50] (-1.75, -0.75) to (1.75, -0.5);
\end{tikzpicture}
        \end{tabular}
    \end{center}
    \caption{As in the statement and proof of Lemma~\ref{lem:forbiddencycles}, we are given a forbidden cycle $C$, with unsafe vertices $u$ and $v$, and degree 2 vertices $w$ and $z$. \\
    (a) in the first figure, the solid edges represent edges in $F$, and the dashed represent the edges not in $F$ but in $C$. In the second figure we have the solid edges representing $F' \coloneqq F\backslash\{uw\} \cup \{vw\}$. 
    Since $v$ is safe, we can replace $uw$ with $vw$, and not create an unsafe cut vertex. So $F'$ is a feasible FVC solution.\\
    (b) Here the solid edges again represent edges of $F'$. The green edges show a path from $u$ to $v$ that does not contain an edge of $C$. Here we depict the case that $z$ is an unsafe cut vertex of $F'$. Since $G$ is $2$VC, we can pick edge $e$ of this path that connects the two components of $F'\backslash\{z\}$.}
    \label{fig:forbiddencyclessafe}
\end{figure}
    
It remains to show that the reduction described in Lemmas~\ref{lem:forbiddencyclesUnsafe}~and~\ref{lem:forbiddencycles} is approximation preserving and can be performed in polynomial time.

\begin{proof}[Proof of Lemma~\ref{lem:removeforiddencycles}]
    We solve the problem recursively. First, observe that we can check the feasibility of the instance in polynomial time. Moreover, if our instance has less than five vertices we can solve it in polynomial time by enumeration.
    If there exists a cut-vertex $x$ in $G$, then in polynomial time we can find sets $V_1,V_2$ such that $V_1\cup V_2\cup \{x\}$ is a partition of $V$ described by Lemma~\ref{lem:2-vertex-connected-node}. Thus computing a $\beta$-approximation on $G_1:=G[V_1\cup \{x\}]$ and $G_2:=G[V_2\cup \{x\}]$ and taking the union leads to a $\beta$-approximation for $G$. 
    
    Furthermore, if $G$ has a forbidden cycle that satisfies conditions of Lemma~\ref{lem:forbiddencyclesUnsafe} if we find a $\beta$-approximation for $G\setminus\{w\}$ and then add $\{uw,wv\}$ to it we obtain a solution of cost $\beta (opt(G)-2)+2<\beta opt(G)$.
    If $G$ has a forbidden cycle that satisfies conditions of Lemma~\ref{lem:forbiddencycles} then we can find an edge ($uw$) and remove it from the instance.
    
    Let $T(m)$ be the running time for an instance on $m$ edges.  By the above argument we have form $m\ge 3$:
    $$T
        (m)\le F(m) + \max\{T(m_1)+T(m-m_1),T(m-2),T(m-1)\},
    $$ 
    where $1\le m_1 \le m-1$ and $F(m)$ is a polynomial function of $m$, that upper bounds the running time of any of the operations listed above. Note that $T(1)=T(2)=O(1)$ as this small instances can be solved by constant number of operation. Therefore our algorithm has polynomial running time.
\end{proof}

\subsection{Proof of Lemma~\ref{lem:optnottree}}
\label{sec:optnottree}
\begin{proof}
    As $(V,OPT)$ is connected, $|OPT|\ge n-1$.
    If $|OPT|=n-1$, then by the feasibility of $OPT$, $(V, OPT)$ is a tree in which the unsafe vertices are leaves. Therefore, the only vertices of $OPT$ with degree at least 2 are safe, and thus form a connected subtree of $G$.
    
    To find a solution $OPT'$ with $|OPT'|=|OPT|$, we compute a spanning tree on the subgraph induced by the safe vertices. For every unsafe vertex $u\in V$, we buy edge $uv\in E$, where $v$ is a safe vertex $v$ that is adjacent to $u$. $OPT'$ is a tree, and thus $|OPT'|=n-1=|OPT|$.
\end{proof}

\subsection{Proof of Lemma~\ref{lem:decompositioninvariant}}
\label{sec:decompositioninvariant}
\begin{proof}
    We prove that the claim is true by showing that the inequality is invariant at every step of the construction of $D$ by using induction on the number of steps.

    In the first iteration, $D$ is a cycle of length at least 4. Thus $4\leq |E(D)|=|V(D)|\le \frac{4}{3}(|V(D)|-1)$.
    
    Now assume the statement holds in the $i^{th}$ iteration. In the $i+1^{st}$ iteration, the number of edges that we add to $D$ is at most $\frac{4}{3}$ the number of vertices added at this step since we add potential open ears of length at least $4$. Therefore the above inequality holds also at step $i+1$.
\end{proof}

\subsection{Proof of Lemma~\ref{lem:apx1feasible}}
\label{sec:apx1feasible}
\begin{proof}
    Recall, $D$ can be computed in polynomial time by Claim~\ref{lem:eardecompinpolytime}. 
    
    Computing and identifying $K_{1,1},K_{1,2},K_{2,1},K_{2,2}$ can also be done in polynomial time. Therefore, the edges we add to $APX_1$ that share an endpoint with $K_{1,1}\cup K_{1,2}\cup K_{2,2}, \cup K_{2,3}$ can be computed in polynomial time.
    
    We now show that $APX_1$ is feasible by showing that every time we add a set of edges to $APX_1$, the subgraph $(V(APX_1),APX_1)$ is a feasible FVC solution for the instance defined by $G[V(APX_1)]$, namely $(V(APX_1),APX_1)$ is connected and has no unsafe cut-vertices ($(V(APX_1),APX_1)$ is not a feasible solution for the original instance).
    Let us start with $APX_1\leftarrow D$. The subgraph $D$ is 2VC, and hence $E(D)$ is a feasible FVC solution on $V(D)$. Observe that the edges selected for components of $K_{1,2}$ and $K_{2,3}$ are potential open ears for $D$ of length two and three, respectively. Therefore, as $D$ is $2$VC, when we add these edges to $APX_1$, we have $(V(D)\cup K_{1,2}\cup K_{2,3},APX_1)$ is $2$VC.  
    
    For every $v\in K_{1,1}$, we buy edge $uv$, where $u\in V(D)$ is safe. As by adding these edges to $APX_1$,  $v$ is only adjacent to $u$ then $(V(D)\cup K_{1,2}\cup K_{2,3}\cup K_{1,1},APX_1)$ is connected and has no unsafe cut-vertex. A similar argument holds also for $K_{2,2}$ and hence this implies $APX_1$ is feasible.
\end{proof}

\subsection{Proof of Lemma~\ref{lem:UpperboundSimple}}
\label{sec:UpperboundSimple}
\begin{proof}
    We have $|OPT|\geq n$ by Lemma~\ref{lem:optnottree}. We now show $|OPT| \ge |K_{1,1}| + 2|K_{1,2}| + |K_{2,2}| + \frac{3}{2}|K_{2,3}|$. Recall by Lemma~\ref{lem:matching}, that $G[V\setminus V(D)]$ is a matching. %, and that $G$ is $2$VC, and $|V| \geq 5$ by Lemma~\ref{lem:removeforiddencycles}.  
    Consider a connected component $F$ of $G[V\setminus V(D)]$. We consider cases of $F$.
    If $V(F) = \{v\} \subseteq K_{1,1}$, $OPT$ has at least one edge incident on $v$.
    If $V(F) = \{v\} \subseteq K_{1,2}$, then $v$ is not adjacent to a safe vertex, hence $OPT$ has at least two edges incident on $v$.
    If $V(F) = \{u,v\} \subseteq K_{2,2}$, then as $(V,OPT)$ is connected, $OPT$ has at least two edges incident on at least one vertex in $\{u,v\}$.
    If $V(F) = \{u,v\} \subseteq K_{2,3}$, then as $(V,OPT)$ is connected, $OPT$ has at least two edges incident on at least one vertex in $\{u,v\}$. However as $(V,OPT)$ has no unsafe cut-vertex then by definition of $K_{2,3}$, $OPT$ must have at least three such edges.
\end{proof}

\subsection{Proof of Lemma~\ref{lem:5/3apx}}
\label{sec:5/3apx}
\begin{proof} 
    Our approximation factor is upper-bounded by
    \begin{align*}
        S&:=\frac{ \frac{4}{3}(|V(D)|-1) + |K_{1,1}| + 2|K_{1,2}| + |K_{2,2}| + \frac{3}{2}|K_{2,3}|}{\max\{ n, |K_{1,1}| + 2|K_{1,2}| + |K_{2,2}| + \frac{3}{2}|K_{2,3}|\}}\\ 
        &= \frac{ \frac{4}{3}(|V(D)|-1) + |K_{1,1}| + 2|K_{1,2}| + |K_{2,2}|+ \frac{3}{2}|K_{2,3}| }{\max\{ |V(D)|+|K_{1,1}| + |K_{1,2}| + |K_{2,2}| + |K_{2,3}|, |K_{1,1}| + 2|K_{1,2}| + |K_{2,2}| + \frac{3}{2}|K_{2,3}|\}}\\
        &\le \frac{\frac{4}{3}|V(D)|-\frac{4}{3} + 2|K_{1,2}| + \frac{3}{2}|K_{2,3}|}{\max\{|V(D)|+  |K_{1,2}|  + |K_{2,3}|, 2|K_{1,2}| + \frac{3}{2}|K_{2,3}|\}}.
    \end{align*}
    % $$
    %     S=\frac{ \frac{4}{3}(|V(D)|-1)+K_{1,1}+2K{1,2}+2K_{2,2}+3K_{2,3}}{\max\{ n, |K_{1,1}| + 2|K_{1,2}| + |K_{2,2}| + \frac{3}{2}|K_{2,3}|\}}
    % $$
    % As  $n\leq |V(D)|+K_{1,1}+2K_{1,2}$, then:
    % $$
    %     S\leq \frac{n+(|V(D)|/3)+K_{1,1}+K_{1,2}-4/3}{\max\{n,n-|V(D)|+K_{1,1}+K_{1,2}\}}
    % $$
    Now we consider two case:
    \begin{itemize}
        \item If $|K_{1,2}|+\frac{1}{2}|K_{2,3}|\ge |V(D)|$. In this case we have:
        $$S=\frac{\frac{4}{3}|V(D)|-\frac{4}{3} + 2|K_{1,2}| + \frac{3}{2}|K_{2,3}|}{ 2|K_{1,2}| + \frac{3}{2}|K_{2,3}|}.$$
        Therefore by our condition, we have: $$S\le \frac{\frac{4}{3}(|K_{1,2}|+\frac{1}{2}|K_{2,3}|)-\frac{4}{3} + 2|K_{1,2}| + \frac{3}{2}|K_{2,3}|}{ 2|K_{1,2}| + \frac{3}{2}|K_{2,3}|}<\frac{5}{3}.$$
        % Thus $$S\le 1+\frac{\frac{4}{3}(|V(D)|-1)}{n}\le 1+\frac{\frac{4}{3}(n/2-1)}{n}\le 1+2/3.$$
         \item Otherwise we have:
        $$S=\frac{\frac{4}{3}|V(D)|-\frac{4}{3} + 2|K_{1,2}| + \frac{3}{2}|K_{2,3}|}{ |V(D)|+|K_{1,2}| + |K_{2,3}|}=\frac{4}{3}+\frac{-\frac{4}{3}+ 2|K_{1,2}| + \frac{3}{2}|K_{2,3}|-\frac{4}{3}(|K_{1,2}| + |K_{2,3}|)}{|V(D)|+|K_{1,2}| + |K_{2,3}|}$$\\
        $$\le \frac{4}{3}+\frac{-\frac{4}{3}+\frac{2}{3}|K_{1,2}| + \frac{1}{6}|K_{2,3}| }{|V(D)|+|K_{1,2}| + |K_{2,3}|}<\frac{4}{3}+\frac{\frac{2}{3}|K_{1,2}| + \frac{1}{6}|K_{2,3}| }{2|K_{1,2}| + \frac{2}{3}|K_{2,3}|}<\frac{5}{3}.$$
        % We show that $K_{1,1}+K_{1,2}<|V(D)|$ and $K_{1,1}+2K_{1,2}+|V(D)|=n$ implies that $K_{1,1}+K_{1,2}+|V(D)|/3\le 2/3n$. To prove this observe that if $|V(D)|\le n/2$ then $K_{1,1}+K_{1,2}\le n/2$ and thus $K_{1,1}+K_{1,2}+|V(D)|/3\le n/2+n/6=2n/3$ and if $|V(D)|> n/2$, then $K_{1,1}+K_{1,2}+|V(D)|/3\le n-|V(D)|+|V(D)|/3=n-2|V(D)|/3\le n-n/3=2n/3$.
        % Thus 
        % $$S\le 1 + \frac{2/3n-4/3}{n}<5/3.$$
    \end{itemize}
    So the first claim is proven. So now we consider the case when $|K_{2,3}|+|K_{1,2}|\le 2$.
    \begin{align*}
        S&=\frac{ \frac{4}{3}(|V(D)|-1) + |K_{1,1}| + 2|K_{1,2}| + 2|K_{2,2}| + \frac{3}{2}|K_{2,3}|}{\max\{ n, |K_{1,1}| + 2|K_{1,2}| + |K_{2,2}| + \frac{3}{2}|K_{2,3}|\}}\\ 
        &\le \frac{ \frac{4}{3}(|V(D)|-1) + |K_{1,1}| + 2|K_{1,2}| + 2|K_{2,2}|+ \frac{3}{2}|K_{2,3}| }{ |V(D)|+ |K_{1,1}| + |K_{1,2}| + |K_{2,2}| + |K_{2,3}|}
        \le \frac{\frac{4}{3}|V(D)|-\frac{4}{3} + 2|K_{1,2}| + \frac{3}{2}|K_{2,3}|}{|V(D)|+  |K_{1,2}|  + |K_{2,3}|}\\
        &\leq\frac{\frac{4}{3}|V(D)|-\frac{4}{3} + 4}{|V(D)| + 2}=\frac{4}{3}.
    \end{align*}
        
    % \begin{align*}
    %     \\
    %     &\le \frac{\frac{4}{3}|V(D)|-\frac{4}{3}+\frac{2}{\varepsilon}}{|V(D)|+\frac{1}{\varepsilon}}
    % \end{align*}
\end{proof}

\subsection{Proof of Claim~\ref{lem:eardecompinpolytime}}
\label{sec:eardecompinpolytime}

\begin{proof}
    We can check in polynomial time for $x\in V(D)$ and $y_1,y_2,y_3\notin V(D)$ if there exists such a path that starts with $x, y_1, y_2$ and $y_3$ (in this order).
    
    To do this we first check if edges $xy_1, y_1y_2, y_2y_3 \in E$ (if not pick different vertices until all have been scanned). Then we remove $y_1$, $y_2$ and $x$ from $G$ and see if there is a path from $y_3$ to a vertex in $V(D)\setminus x$. As this can be done in polynomial time and there are at most $n^4$ such tuples, then the claim follows.
\end{proof}

\subsection{Proof of Lemma~\ref{lem:RainbowForest}}
\label{sec:RainbowForest}
\begin{proof}
    For graph $G=(V,E)$, we  first define two matroids with ground set $E$:
    \begin{itemize}
        \item Graph matroid: $M_{graphic}$ over the edges $E$ with independent sets $\mathcal{I}_{graphic} = \{F\subseteq E | F \text{ is acyclic}\}$.
        \item Partition matroid: partition edges $E = E_1 \cup \dots \cup E_k$ into $k$ colour classes, $M_{partition}$ over the edges $E$ with independent sets $\mathcal{I}_{partition} = \{F\subseteq E | |F\cap E_i|\leq 1, i=1,\dots k\}$.
    \end{itemize}
    The intersection matroid of $M_{graphic}$ and $M_{partition}$ has independent sets that are acyclic and contain each edge colour at most once. 
    Thus, the intersection matroid $M_{graphic}\cap M_{partition}$  can be solved in polynomial time to find a solution $P'$. 
    For each colour class that is not in $P'$, we select an arbitrary edge and add it to $P'$ to find $P$. 
    Note that $P'$ is a maximum forest that uses each colour class at most once, that is, $P'$ has a minimum number of components. Furthermore, by construction $P$ has the same set of components as $P'$ since if it decreased the number of components, $P'$ would not be optimal. 
    
    To find a minimal number of singletons with respect to replacing edges of a colour class, we consider the edges of $P$ and if there is an edge that can be swapped to reduce the number of singletons, we make that swap until there are no more improving swaps. 
    Formally, for each edge $e\in P$,  we can find all edges in $\tilde E$ with colour $c$, and identify if there is an edge $e'$ with colour $c$ such that $(V, P \cup \{e'\} \setminus \{e\} )$ has one fewer singleton. If there is such an edge,  we replace $P$ with $P\cup \{e'\} \setminus \{e\} $.
    Note that by adding $e'$ we reduce the number of connected components in $(V,P)$ by exactly one, and by removing $e$ we increase the number of connected components by at most one. Thus the number of connected components in this process will never increase.

    This process terminates in polynomial time, as it reduces the number of singletons by one in each step, and each step takes polynomial time to compute.
\end{proof}

% \begin{lemma}[Local Swaps]
% \label{lem:localswap}
%     We can find in polynomial time a Maximum Rainbow Forest that has a minimal number of singletons.
    
%     We can assume without loss of generality that there is no edge of $P$ that can replaced with another edge in $\tilde E$ of the same colour to reduce the number of singletons.
% \end{lemma}\todo{we should clarify somewhere, what we mean by singletons in the way we use them in this section.}

%     % for any rainbow tree component with at least 2 edges, if that component has a pseudo-edge $e$ of colour $c$ that is adjacent to a pseudo-edge $f$ with colour $c$ that is adjacent to a singleton, then we replace $e$ with $f$. So long as this does not create a new singleton. 

\subsection{Proof of Lemma~\ref{lem:findgoodcycles}}
\label{sec:findgoodcycles}
\begin{proof}
    To prove this lemma, we use a similar construction due to  \cite{garg2023improved}, which they show can be found in polynomial time. 
    \begin{definition}[\cite{garg2023improved}]
        (Nice Cycle and Path). Let \(\Pi = \{ V_1, \ldots, V_k \}\), \(k \geq 2\), be a partition of the vertex-set of a graph \(G\). A \textit{nice cycle} (resp. \textit{nice path}) \(N\) of \(\Pi\) is a subset of edges with endpoints in distinct subsets of \(\Pi\) such that: (1) \(N\) induces one cycle of length at least 2 (resp., one path of length at least 1) in the graph obtained from \(G\) by by contracting each $V_i$ into a single vertex one by one; (2) given any two edges of \(N\) incident to some \(V_i\), these edges are incident to distinct vertices of \(V_i\) unless \(|V_i| = 1\).
    \end{definition}
    
    \begin{lemma}[\cite{garg2023improved}]
    \label{lem:findnicecycles}
        Let \(\Pi = (V_1, \ldots, V_k)\), \(k \geq 2\), be a partition of the vertex-set of a \(2VC\) simple graph \(G\). In polynomial time one can compute a nice cycle \(N\) of \(\Pi\).
    \end{lemma}
    
    Each set $V_i\in \Pi$ induces a connected graph in $G$, and we call these sets \textit{components} for simplicity. For $V_i\in \Pi$ with $|V_i| \geq 2$, we say that $V_i$ is a large component. For $V_i\in \Pi$ with $|V_i| = 1$, we say that $V_i$ is a singleton of $\Pi$. Using the sets of $\Pi$, we create a new vertex partitioning $\Pi''$ in the following way.

    We define $A_1 \subseteq \Pi$ as the singletons of $\Pi$ that are adjacent to another singleton in $\Pi$. Denote by $F_1$ the set of maximal connected components of $A_1$ with edges in $E$. Define $\Pi' \coloneqq \Pi \cup F_1 \backslash A_1$. So basically $\Pi'$ is a partition obtained from $\Pi$ by replacing  singletons of $\Pi$ by the vertex sets of the connected components of graph $G[A_1]$.

    We define $A_2\subseteq \Pi'$ as the singletons of $\Pi'$ that are incident to two distinct edges with endpoint in the same large component of $\Pi$. Starting with $F_2 \coloneqq \emptyset$ and $F_2^- \coloneqq \emptyset$. For each $D\in \Pi$, that is adjacent to a singleton in $A_2$, let $N_{A_2}(D)$ denote the vertices of $A_2$ that are adjacent to $D$. Add $D\cup N_{A_2}(D)$ to $F_2$, add $D$ and every vertex of $N_{A_2}(D)$ to $F_2^-$, and finally remove $N_{A_2}(D)$ from $A_2$. When $A_2$ is empty, define $\Pi'' \coloneqq \Pi' \cup F_2 \backslash F_2^-$. So basically we replace large components and any singletons (that haven't already been processed) that have two edges between them with a single large component.
    
    % add $D\cup (A_2 \cap N(D))$ let $D$ be a large component of $\Pi$ be such that $a$ is incident to two distinct edges with endpoints in $D$. 
    % Add $D\cup\{a\}$ to $F_2$, and add $D$ and $\{a\}$ to $F_2^-$. Once each vertex in $A_2$ is added to a large component in $F_2$, define $\Pi'' \coloneqq \Pi' \cup F_2 \backslash F_2^-$. 

    We define $F_0 \coloneqq \Pi \backslash F_2^-$, the large components that are not subsets of components in $F_1$. Finally, we define $A_0\subseteq \Pi''$ the set of singletons of $\Pi''$.
    % We define by $F_1$ the of singletons of $\Pi$ that are adjacent in $G$, we define $\Pi'$ by treating these singletons as a single component. For every singleton in $\Pi'$ that has two incident edges between it and a large component of $\Pi$, we add these edges to the large component to define $\Pi''$, denote by $F_2$ these large components. Denote by $F_0$ the remaining large components of $\Pi$. The singletons of $\Pi''$ are denoted by $A$. 
    The following observation is clear from the definition of these sets.
    \begin{observation}
        $V(F_0)\cup V(F_1) \cup V(F_2) \cup V(A_0) = V$, and $V(F_0), V(F_1) , V(F_2) , V(A_0)$ are pairwise disjoint.
    \end{observation}
    As we are always merging components of $\Pi$ to create $\Pi'$ the following observation is clear.
    \begin{observation}
        Every component of $\Pi$ is a (not necessarily strict) subset of a component of $\Pi'$.
    \end{observation}
    Thus, $\Pi''$ partitions $V$. To see that $|\Pi''| \geq 2$, we consider cases. If there are at least two large components of $\Pi$ then $|\Pi''| \geq 2$ since large components of $\Pi$ are not part of the same component of $\Pi''$. If there is exactly one large component of $\Pi$ then at least one of $F_0$ or $F_2$ is non-empty, and by our assumption on $\Pi$, there are at least two singletons that are adjacent in $G$, which shows that $F_1$ is non-empty. 
    Therefore, $\Pi''$ satisfies the conditions of Lemma~\ref{lem:findnicecycles}. 
    
    We compute a nice cycle $C$ using Lemma~\ref{lem:findnicecycles} on $\Pi''$. It remains to find a good cycle $C'$ on $\Pi$. We will find sure $C'$ by using the edges of $C$, as well as additional edges depending the components that are incident to $C$. In particular, if $C$ incident to a component $D$ of $F_0$ or $F_2$, then $C'$ will be incident to the large component that is contained in $D$. If $C$ is incident to a component $D \in F_1$, then we guarantee that $C'$ will be incident to at least two singletons of $\Pi$ that are contained in $D$. We begin by adding every edge of $C$ to $C'$. For every component $D\in \Pi''$ that has endpoints of edges of $C$, we consider by cases which set among $F_0, F_1 ,F_2$, and $A_0$ contain $D$.

    First, if $D\in A_0$ or $D\in F_0$, we do nothing.

    Next, if $D\in F_1$. Let $e_1,e_2\in E(C)$ be the edges of $C$ incident to $D$.  Observe, by definition of a nice cycle, $e_1$ and $e_2$ have distinct endpoints $v_1,v_2\in V(D)$ respectively. By construction, $G[D]$ is connected, so there is a path in $G[D]$ between $v_1$ and $v_2$. We add the edges of this path to $C'$.  

    Finally, if $D\in F_2$. Let $e_3,e_4\in E(C)$ be the edges of $C$ incident to $D$. By definition of a nice cycle $e_3$ and $e_4$ have distinct endpoints $v_3,v_4\in V(D)$ respectively. We consider cases for $v_3$ and $v_4$. By construction, there is exactly one large component $L\in \Pi$ such that $L \subsetneq D$. 
    
    \textbf{Case:} If $v_3, v_4\in V(L)$. We are done, as $C'$ now has two edges with distinct endpoints in a large component of $\Pi$.

    \textbf{Case:} If $v_3 \in V(L)$ but $v_4 \notin V(L)$. ( the opposite case where only $v_4\in V(L)$ is similar). By definition of $F_2$, $v_3$ is incident to at least two distinct edges, with endpoints in $L$, say $e_3',e_3''$. At most one of these two may  have $v_4$ as an endpoint, without loss of generality say $e_3''$ has $v_4$ as an endpoint in this case. We add $e_3'$ to $C'$ and we are done, as $C'$ now has two edges with distinct endpoints in $L$, a large component of $\Pi$.

    \textbf{Case:} If $v_3, v_4\notin V(L)$. By definition of $F_1$, $v_3$ and $v_4$ are incident to at least two distinct edges, with endpoints in $L$. Denote edges incident to $v_3$ by $e_3',e_3''$ and the edges incident to $v_4$ by $e_4',e_4''$. We add $e_3'$ to $C'$. At most one of $e_4'$ and $e_4''$ can share an endpoint with $e_3'$, without loss of generality say $e_4'$. Add $e_4''$ to $C'$, and we are done as $C'$ now has two edges with distinct endpoints in $L$, a large component of $\Pi$.

    By construction, we have the following: (1) $C'$ induces a cycle of length at least 2 in the graph obtained from $G$ by collapsing each $V_i\in \Pi$ into a single vertex; (2) given any two edges of $C'$ incident to some $V_i$, these edges are incident to distinct vertices of $V_i$ unless $|V_i|=1$.
    
    It remains to show the remaining conditions of a good cycle: (3) $C'$ has an edge incident to at least one $V_i \in \Pi$ with $|V_i| \geq 2$, and;  (4) $|C'|=2$ only if both $V_i$ and $V_j$ incident to $C'$ have $|V_i|,|V_j| \geq 2$. To see point (3), by construction of the sets $F_0,F_1,F_2$ and $A_0$, any component in $F_1$ and $A_0$ are adjacent only to components in $F_0$ and $F_2$. Thus, since a nice cycle has at least two edges, it will contain a component from $F_0$ or $F_2$. And by construction of $C$, $C$ will be incident to a large component of $\Pi$.

    To see point (4), observe that since $C \subseteq C'$, so if $|C'|=2$ then $|C|=2$. So we consider the cases where $|C|=2$. If $C$ is incident to a component $D$ in $F_1$, then $C'$ will have at least three edges by construction, as $C'$ also contains edges of a path in $G[D]$. If $C$ is incident to a component in $A_0$, then by definition of $A_0$, $C$ is incident to at least 2 other components in $\Pi''$, hence $|C| \geq 3$. Therefore, $C$ can be adjacent only to components in $F_0$ and $F_2$. If $C$ is incident to a component in $F_2$ at a vertex that is a singleton in $\Pi$, then $|C'| \geq 3$. Thus, $C'$ is incident to vertices of exactly two large components of $\Pi$. And $C'$ satisfies the conditions of a good cycle.
\end{proof}

\subsection{Proof of Lemma~\ref{lem:FVCupperbound}}
\label{sec:FVCupperbound}
\begin{proof}
Let us begin by the following simple claim.
    \begin{claim}
    \label{lem:s0atmostalpha}
        $(V(D), P)$ has at most $|V(D)|-\alpha$ blocks.
    \end{claim}
    \begin{proof}
        Assume that the initial connected components are of size $a_1,...,a_\alpha$. Then the total number of blocks by Lemma~\ref{lem:reduceblocks} is at most $\sum_{i=1}^\alpha (a_i-1)=n-\alpha$.
    \end{proof}
    Now we will use the next  observation for the analysis of Algorithm~\ref{alg:phase1}.
    \begin{observation}\label{obs:goodCycleBlockDecrease}
        Consider good cycle $C$ computed by line $4$ in Algorithm~\ref{alg:phase1} and suppose $C$ contains $l$ large connected components. By adding $C$ to $S_1$ the number of blocks in $(V(D),P\cup S_1)$ is reduced by at least $l-1$.
        % Consider a good cycle $C$ that is merging $l$ large connected components and $s$ singletons into one connected component. Then after adding $C$ to the graph the number of blocks decreases by at least $l-1$.
    \end{observation}
    \begin{proof}
        Let $T_1,\dots ,T_{l}$ be the set of large connected components of $(V(D),P\cup S_1)$ that are part of cycle $C$, and let $u_i$ and $v_i$ be the vertices of $T_i$ that are incident on edges of $C$ for $i=1,\dots,l$. Now, consider all the edges that belong to paths from $u_i$ to $v_i$ in $T_i$, which we denote by $F_i$. Let   $F'_i\supseteq F_i$ denote the union of the edge sets of blocks in $T_i$ that contain at least one edge of $F_i$.
        % set of edges of $E(C_i)$ that are in a block with at least one edge of $F_i$. 
        % and let $F_i$ be the set of all edges that belongs to a block in $C_i$. 
        Observe that by adding $C$, the set $E(C)\cup F_1' \cup F_2' \dots \cup F_{l}'$ merge into a single block $B$. Thus at least $l$ blocks merge into a single block $B$ and hence the number of blocks decreases by at least $l-1$.
    \end{proof}
    In the next two observations we analyze the running time of Algorithm~\ref{alg:phase1} and the cost of the solution obtained by this algorithm.
    \begin{observation}\label{obs:goodcycleFeasibility}
        Algorithm~\ref{alg:phase1} terminates in polynomial time. Furthermore, upon termination of  Algorithm~\ref{alg:phase1}, $(V(D),P\cup S_1)$ has one large connected component, namely $A$, and $V(D)\setminus V(A)$ forms an independent set of vertices in graphs $(V(D),P\cup S_1)$ and $G[V(D)]$.
    \end{observation}

    \begin{proof}
        Note that initially $P\neq \emptyset$ as we are applying Lemma~\ref{lem:5/3apx} and assuming at least one of $K_{1,2}$ and $K_{2,3}$ is not an empty set.  Furthermore at each step of Line 4 to 5 in Algorithm~\ref{alg:phase1}, we are adding edges of a good cycle to $S_1$ to merge at least two connected components of  $(V(D),P\cup S_1)$ into a single component. 
        
        Therefore (1) Line 4 to 5 iterates at most $|V(D)|$ times, and; (2) once the algorithm reaches line 6, we must have at least one large connected component. Furthermore, using Lemma~\ref{lem:findgoodcycles} we can always find a good cycle as long as we have more than one large connected components in $(V(D),P\cup S_1)$, or at least two singletons are adjacent.
        Thus, once we reach line 6 of Algorithm~\ref{alg:phase1}, we have precisely one large connected component, $A$ in $(V(D),P\cup S_1)$. Additionally, at this stage all the vertices of $V(D)\setminus V(A)$ form an independent set of vertices in $G\cup P$, as otherwise there is another good cycle.

        To see that Algorithm~\ref{alg:phase1} terminates in polynomial time, observe that it has at most $|V(D)|$ iterations, and in each iteration we compute a good cycle, which can be done in polynomial time by Lemma~\ref{lem:findgoodcycles}.
    \end{proof}
    
    \begin{observation}\label{obs:goodcycleanalysis}
         At the end of Algorithm~\ref{alg:phase1},
        (1) $ (V(D), S_1  \cup P)$ has at most $|V(D)|-\alpha - \alpha_{large}+1$ blocks, and;
        (2) $|S_1|\le 2\alpha_{large} + \frac{3}{2}|X_1|-2$.
    \end{observation}
    \begin{proof}
        We started with at most $|V(D)|-\alpha$ blocks by Claim~\ref{lem:s0atmostalpha}. Now let $a_1,\ldots,a_t$ be the size of these good cycles, and let $l_i$ be the number of large connected components that appeared in $i^{th}$ good cycle. By Observation~\ref{obs:goodCycleBlockDecrease}, the total decrease on the number of blocks after these iterations is at least $\sum_{i=1}^{t}(l_i-1)$. However, $l_i-1$ is equal to the decrease in the number of large connected components of $(V(D),S_1\cup P)$ when we add the $i^{th}$ cycle to $S_1$. Hence, as we start with $\alpha_{large}$ many large connected components then $\sum_{i=1}^{t}(l_i-1)=\alpha_{large}-1$. Altogether, these imply that the total number of the blocks at the end is at most $|V(D)|-\alpha-(\alpha_{large}-1)$, and claim (1) of the observation is proven.
        
        Furthermore, we observe that $\sum_{i=1}^t (a_i-1)=\alpha_{large}+|X_1|-1$ (This is due to two facts: 1) adding the $i^{th}$ good cycle to $S_1$ decreases the number of connected components of $(V(D),S_1\cup P)$  by $a_i-1$, and; 2) adding all these cycles merge the initial $\alpha_{large}$ many large connected components and $|X_1|$ many singleton components into a single (large) component).
        
        Therefore the total number of edges (i.e. $\sum_{i=1}^t a_i$) is equal to $\alpha_{large}+|X_1|-1+t$. Assume that among these $t$ good cycles, we have $t'\le t$ cycles of length $2$. 
        Then  since each $a_i$ is at least $2$, we we can see that $t'+2(t-t')\le \sum_{i=1}^t (a_i-1)=\alpha_{large}+|X_1|-1$.
        Furthermore, since every cycle of length two merges two large connected components into one connected component, we have $t'\le \alpha_{large}-1$,. 
        
        Thus,  $t=\frac{(t'+2(t-t'))+t'}{2}\le \frac{1}{2}(\alpha_{large} + |X_l|-1+\alpha_{large}-1)$. Hence $\sum_{i=1}^t a_i= \sum_{i=1}^t (a_i-1) +t\le (\alpha_{large} + |X_1|-1) + \alpha_{large} -1+ \frac{|X_l|}{2} = 2\alpha_{large} +\frac{3}{2}|X_1|-2$, hence the claim.
    \end{proof}
    Now we analyze the second phase of our algorithm (Algorithm~2).
    \begin{observation}\label{obs:SecondPhaseComponentsAnd blocks}
        At the end of Algorithm~\ref{alg:phase2}, $(V(D) ,S_1\cup S_2 \cup P)$ has $\alpha'$ connected components that are singletons (which is set $X_3$) and one large connected component that has exactly one block. Furthermore, Algorithm~\ref{alg:phase2} terminates in polynomial time.
    \end{observation}
    \begin{proof}
        First recall that since $D$ is defined as an open ear decomposition, $G[V(D)]$ is $2$VC by Lemma~\ref{lem:2VCandOpenEars}. Second, by Observation~\ref{obs:goodcycleFeasibility}, $V(D)\setminus V(A)$ is an independent set.
        
        We show that, as long as $G[V(A)\cup X_2]\cup P$ has more than one block, then there exists a vertex $v\in V(D)\setminus(V(A)\cup X_2)$ such that $G[V(A)\cup X_2]\cup P$ has more blocks than $G[V(A)\cup X_2 \cup \{v\}]\cup P$.
        
        Assume otherwise for sake of contradiction. That is, for every vertex $v\in V(D)\setminus(V(A)\cup X_2)$ there exists a block $B_v$ of $G[V(A)\cup X_2]\cup P$ such that every neighbor of $v \in V(D)$ belongs to $V(B_v)$. 
        Hence any two edges of $G[V(A)\cup X_2]\cup P$ that are not in the same block in $G[V(A)\cup X_2]\cup P$ are not in the same block $G[V(D)]\cup P$. Therefore $G[V(D)]\cup P$ is also not $2$VC, a contradiction. 
        
        So Line 4 of Algorithm~\ref{alg:phase2} can find such a $v$. Moreover, $v$ can be found in polynomial time.
        
        Now we show that we can find edges $e_1$ and $e_2$ as indicated by line $5$ of Algorithm~\ref{alg:phase2}. Assume for contradiction this is not the case. Consider any $x\neq y\in V(A) \cup X_2$,  such that $vx$ and $vy$ are edges of $G$. Let $E_{xy}$ denote the edges of $(V(A)\cup X_2 \cup \{v\},P\cup S_1\cup S_2)$ that belong to a path from $x$ to $y$ in $(V(A)\cup X_2 \cup \{v\},P\cup S_1\cup S_2)$. Then as $vx$ and $vy$ do not satisfy conditions of line $5$ of Algorithm~\ref{alg:phase2} then all edges of $E_{xy}$ are in a unique block $B_{xy}$ of $G[V(A)\cup X_2 \cup \{v\}]\cup P$. 
        
        There must exist a pair  $x',y'\in V(A) \cup X_2$, $x'\neq y'$, such that $vx'$ and $vy'$ are edges of $G$. Otherwise, every neighbour of $v$ is in a single block, thus $G[V(A)\cup X_2 \cup \{v\}]\cup P$ has the same number of blocks as $G[V(A)\cup X_2 ]\cup P$, which is a contradiction. 
        
        Therefore, there must exist $x,y,x',y' \in V(A) \cup X_2$ such that $B_{xy}\neq B_{x'y'}$. Furthermore, we can assume that $y=y'$, since if $y\neq y'$, then $B_{xy'}$ can be equal to at most one of $B_{xy}$ and $B_{x'y}$, and without loss of generality we can assume that $B_{xy} \neq B_{x'y}$. There is a path in $B_{xy}$ from $x$ to $y$ with at least one edge $e$, and a path $B_{x'y}$ from $x'$ to $y$ with at least on edge $e'$. Therefore, there is a cycle in $G[V(A)\cup X_2]\cup P \cup\{vx, vx'\}$ that contains $e$ and $e'$, and by Lemma~\ref{lem:blockbound}, $e$ and $e'$ are in the same block, contradicting our assumption.

        Note that as $A$ is the only large connected component of $(V(D),P\cup S_1)$, then the vertices of $A\cup X_2$ is the vertex set of the only connected component of $(V(D),P\cup S_1\cup S_2)$.

        Finally using Lemma~\ref{lem:reduceblocks} while there is an edge $e_3 \in E\backslash S$ such that $( V(A)\cup X_2, P\cup S_1 \cup S_2 \cup\{e_3\} )$ that has fewer blocks than $( V(A)\cup X_2, P\cup S_1 \cup S_2 )$, then it can be found in polynomial time. 
        % when $G[V(A)\cup X_2]\cup P$ is $2$VC line 8 and line 9 Algorithm~\ref{alg:phase2}, can successfully find edges to add to $S_2$ until $(V(A)\cup X_2 ,P\cup S_1\cup S_2)$ is $2$VC.

        Therefore we have one connected component that is $2$VC. Notice that all the vertices not in this component, i.e. the set $X_3$, are isolated vertices in this graph and hence we have $\alpha'$ small components. 

        Observe that the while loop on line 3 iterates at most $|V(D)|$, since in every iteration of this loop we increase the size of $X_2$ and since $X_2\subseteq V(D)$. Furthermore at every iteration of the while loop on line 8, we increase the size of $S_2$ by adding a new edge of $E(G)$. Thus we will have at most $|E(G)|$ iteration of this while loop. Therefore the algorithm runs in poly-time.
       % Proof by construction. Here we have one large connected component and $\alpha'$ singletons.
    \end{proof}

    \begin{claim}\label{claim:SecondPhaseEdges}
         $|S_2|\le |X_2|+|V(D)|-\alpha- \alpha_{large}$.
    \end{claim}
    \begin{proof}
        We start with  $b$ blocks at the start of the Algorithm~\ref{alg:phase2}, where $b\le |V(D)|-\alpha-\alpha_{large} +1$ by Observation~\ref{obs:goodcycleanalysis}. We will consider what happens to the number of blocks of $V(A\cup X_2, P\cup S_1\cup S_2)$ as we add each edge. Let us partition the edges of $S_2$ into $F_1,...,F_{|X_2|},F'$, where  $F'$ is the edges added by line of Algorithm~\ref{alg:phase2} and $F_i$ is the pair of edges added as we add the $i$-th vertex of $X_2$ by line 7 of the algorithm. %and before adding the $(i+1)$-th vertex of $X_2$.
        In that case, every time we add an edge of $F'$, we decrease the number of blocks of $V(A\cup X_2, P\cup S_1\cup S_2)$ by at least one. Furthermore the edges of $F_i$ decrease the number of blocks by at least one. As we have precisely one block in the end, then the total number of edges is:
        $$
            |S_2|=\sum_{i=1}^{|X_2|}|F_i|+|F'|\le (b-1)+|X_2|\le |X_2|+|V(D)|-\alpha-\alpha_{large}.
        $$
    \end{proof}

    By construction in Algorithm~\ref{alg:phase3} we add $|S_3|=\alpha'_{1}+2\alpha'_{2}$ edges. Hence $|S_1|+|S_2|+|S_3|$ is upper-bounded by: 
    \begin{align*}
        &(2\alpha_{large} +\frac{3}{2}|X_1|-2)+(|X_2|+|V(D)|-\alpha-\alpha_{large} )+(\alpha'_1+2\alpha'_2)\\
        =&|V(D)|+\alpha_{large} +|X_2|+\frac{3}{2}|X_1|-2-\alpha+2\alpha'-\alpha'_1 \\
        =&|V(D)|+\frac{|X_1|}{2}+(\alpha_{large} +|X_2|+|X_1|+\alpha')-2-\alpha+\alpha'-\alpha'_1 \\
        =&|V(D)|+\frac{|X_1|}{2}-2+\alpha'-\alpha'_1.
    \end{align*}
  %  where the last equality follows as $\alpha=\alpha_{large}+|X_1|+|X_2|+\alpha'$ and $\alpha'=\alpha'_1+\alpha'_2$.

    Therefore the total number of edges that we use is $|S_1|+|S_2|+|S_3| +|S_P|$, which is at most:
    \begin{align*}
        &|S_P| + |V(D)|+\frac{|X_1|}{2}+\alpha'-\alpha'_1-2 \\
        =& |S_P| +|V(D)|-2-\alpha'_1+\alpha'+ \frac{\alpha-\alpha'-\alpha_{large} -|X_2|}{2} \\
        \le& |V(D)|+|S_P|-2+\alpha -\frac{\alpha-\alpha'}{2} -\frac{\alpha_{large}}{2}-\alpha'_1.
    \end{align*}
    Furthermore as $(V(D)\setminus X_3, S_1\cup S_2\cup P)$ is $2$VC, then $(V(D), S_1\cup S_2\cup S_3\cup P)$ is connected and has no unsafe vertex that is a cut-vertex.
    Finally we show that $S_1\cup S_2\cup S_3\cup S_P$ is feasible. 
    Let us start by showing that  $H:=(V,S_1\cup S_2\cup S_3\cup S_P)$ is connected. Note that by our choice of $S_P$, every vertex $v$ in $K:=K_{1,1}\cup K_{1,2}\cup K_{2,1}\cup K_{2,2}$ is either an endpoint of an edge $uv \in S_P$ such that $u\in V(D)$ or is incident on edges $wv',v'v\in S_P$ such that $w\in V(D)$ and $v'\in K$. 
    
    Thus, it suffices to show that any two vertices in $V(D)$ are in the same connected component of $H$. Now, $(V(D), S_1\cup S_2\cup S_3\cup P)$ is connected, and every edge $xy\in P$, $S_P$ contains a path from $x$ to $y$ then all the vertices of $V(D)$ are in the same connected component of $H$ and hence $H$ is connected.
    
    Now we show that $H$ has no unsafe cut-vertex. Assume this is not true and such a cut-vertex $u$ exists. If $u\in K_{1,1}$, it's a leaf of $H$ adjacent to a safe vertex, thus, it cannot be a cut-vertex. If $u \in K_{2,2}$, since $u$ is a cut-vertex, then the degree of $u$ is at least 2 in $H$, then by our choice of $S_P$, $u$ is safe. 

    If $u\in K_{1,2}$, the degree of $u$ in $H$ is two, thus $H\backslash\{u\}$ has two components. Let $w,z \in V$ be the vertices $u$ is adjacent to. Observe by definition of $K_{1,2}$ that $w,z\in V(D)$. Since $u$ is a cut vertex, there is no path in $H\backslash\{u\}$ from $w$ to $z$. (Note that this implies that there is exactly one copy of the edge $wz$ in $P$).
    Therefore, there is no path from $w$ to $z$ in $(V(D), S_1\cup S_2\cup S_3\cup P\setminus\{wz\})$. So at least one of $w$ and $z$ are cut-vertices of $(V(D), S_1\cup S_2\cup S_3\cup P)$. But since $u\in K_{1,2}$, $u$ is not adjacent to any safe vertices, therefore, $w$ and $z$ are unsafe, contradicting our conclusion that $(V(D), S_1\cup S_2\cup S_3\cup P)$ has no unsafe cut-vertices.

    If $u\in K_{2,3}$, the degree of $u$ in $H$ is at least 2. There are two cases to consider by our choice of $S_P$: 
    (1) $u$ is degree 3, and 
    (2) $u$ is degree 2. If $u$ is degree 3 then by our choice of $S_P$, $u$ is safe vertex, which is a contradiction. Therefore, we assume $u$ is degree 2. 
    
    Let $v$ be the unique vertex in $K_{2,3}$ such that $uv\in E$.
    If $uv\notin S_P$, then there are vertices $u_1,u_2\in V(D)$ such that $uu_1,uu_2\in S_P$, and a safe vertex $v_1\in V(D)$ such that $vv_1\in S_P$. This follows a similar argument to the previous case, where $u\in K_{1,2}$.
    
    Lastly, if $uv\in S_P$, there exist vertices $u',v'\in V(D)$ such that $uu',vv'\in S_P$. Since $u$ is a cut-vertex of $H$, there is no path from $u'$ to $v'$ in $H\setminus \{u\}$. Again this implies that there is exactly one copy of the edge $u'v'$ in $P$.
    Therefore there is no path from $u'$ to $v'$ in $(V(D), S_1\cup S_2\cup S_3\cup P\setminus\{u'v'\})$.  Recall that by Observation~\ref{obs:SecondPhaseComponentsAnd blocks}, $(V(D),S_1\cup S_2 \cup P)$ has a unique large connected component that has exactly one  block $B$. Furthermore as we assume $|K_{1,2}|+|K_{2,3}|> 2$ by Lemma~\ref{lem:5/3apx}, then $|P|=|K_{1,2}|+\frac{1}{2}|K_{2,3}|>1$. Therefore, as $B$ is a block that has at least two edges (as $|P|\ge 2$), there is a path from $u'$ to $v'$ in $(V(D), S_1\cup S_2\cup S_3\cup P\setminus\{u'v'\})$, a contradiction. 
    %However as $|K_{1,2}|+|K_{2,3}|>2$, then $|P|>2$, a contradiction.
    
    % So at least one of $u'$ and $v'$ are cut-vertices of $(V(D), S_1\cup S_2\cup S_3\cup P)$, but by definition of $K_{2,3}$, at least one of $u'$ and $v'$ are unsafe. If both are unsafe, then this contradicts our conclusion that $(V(D), S_1\cup S_2\cup S_3\cup P)$ has no unsafe cut-vertices. 
    
    % Therefore, at least one must be safe, furthermore, exactly one is safe by definition of $K_{2,3}$, without loss of generality $u'$ is the safe vertex. So $u'$ is a safe cut-vertex, and $v'$ is unsafe and not a cut-vertex.

    % {\color{red}If there are two pseudo-edges, then $A$ is 2vC, so there is a path from $u'$ to $v$ when we delete $u$.}

    Therefore we can assume that $u\in V(D)$. Now let $H'$ be the graph obtained from $H$ by removing $u$ and let $C$ be a connected components of $H'$. Observe that by our choice of $S_P$ if $V(C)\cap V(D)=\emptyset$, then $u$ must be a safe vertex, a contradiction. Therefore $V(C)$ has a vertex of $D$. Now as $u$ is a cut-vertex of $H$, then by the above argument there must be two vertices $v$ and $v'$ of $D$ such that there is no path in $H'$ from $v$ to $v'$. Let $R$ be a path from $v$ to $v'$ in $(V(D),S_1\cup S_2 \cup P)$ that does not contain $u$. Now we obtain a path $R'$  from $v$ to $v'$ using $R$, which is a contradiction. We start by setting $R':=R$. Then for any pseudo-edge $xy$ in $R'$ we remove $xy$ from $R'$ and add the unique path of length maximum three from $x$ to $y$ that contains edges of $S_P$. %This is a path in $H'$, a contradiction.
\end{proof}

\subsection{Proof of Lemma~\ref{lem:FVC3lowerbounds}}
\label{sec:FVC3lowerbounds}
\begin{proof}
    We are given graph $G=(V,E)$, for every subset of edges $F\subseteq \binom{V}{2}$, we let $\kappa(F)$ denote the number of connected components of graph $(V,F)$.
\begin{itemize}
    \item We start by showing $|OPT|\ge |K_{1,2}|+2|K_{1,2}|+|K_{2,2}|+\frac{3}{2}|K_{2,2}|$.\\ 
    
    Consider the set of edges $E^*\subseteq OPT$ that have an endpoint in $K_{1,2} \cup K_{2,3}$.
    We show the following useful claim.
    \begin{claim}
    \label{lem:eprimebound}
        $(V,E^*)$ has at least $\alpha + 2K_{1,2}+\frac{3}{2}K_{2,3} - |E^*| + K_{2,2} + K_{1,1}$ many connected components.
    \end{claim}
    \begin{proof}
        % The primary tool for proving this claim is the fact that any subset of pseudo-edges of $\tilde E$ that is a feasible solution to the Maximum Rainbow Connection problem, has at least $\alpha$ connected components. 
        % To apply this fact, we first want to use $E^*$ to create a new set of edges $E_1^*\subseteq \binom{V}{2}$ of size $2|K_{1,2}|+\frac{3}{2}|K_{2,3}|$, and a set of pseudo-edges $P'\subseteq \tilde E$. Where $E_1^*$ satisfies $\kappa(E_1^*)-\kappa(E^*)\le |E^*|-|E_1^*|$. And each edge in $P'$ has a distinct colour.

        % We will use $E_1^*$ and $P'$ to show that $\kappa(E_1^*)\le \alpha+|K_{1,1}|+|K_{2,3}|$ and conclude claim.
    
        The primary tool for proving this claim is the fact that any subset of pseudo-edges of $\tilde E$ that is a feasible solution to the Maximum Rainbow Connection problem, has at least $\alpha$ connected components.
        % To prove this the main tool is to use the fact that by definition of $\alpha$, every rainbow subset of pseudo edges of $\tilde E$ induces at least $\alpha$ connected components. 
        To use this fact we first use $E^*$ to create a set $E_1^*\subseteq \binom{V}{2}$, and a set of pseudo-edges. We will guarantee that $|E_1^*| = 2|K_{1,2}|+\frac{3}{2}|K_{2,3}|$ and $\kappa(E_1^*)-\kappa(E^*)\le |E^*|-|E_1^*|$, and every pseudo-edge of $P'$ has a unique colour. We will use $P'$, and the above fact to show that $\kappa(E_1^*)\le \alpha+|K_{1,1}|+|K_{2,3}|$, and then show the claim holds.

       % We wish to show that 

        Note, every time we decrease $|E_1^*|$ by $k$, $\kappa(E_1^*)$ increases by at most $k$. This is clear since every time we remove an edge from $E_1^*$, we increase $\kappa(E_1^*)$ by at most one.
        %This is clear in the cases that we modify $E_1^*$ by deleting $k$ edges from it as removal of an edge of a graph can increase the number of connected components by at most one.
        
        We begin $P'\leftarrow \emptyset$ and $E_1^* \leftarrow E^*$. 
        For every $v\in K_{1,2}$, we remove all but two edges incident to $v$ from $E_1^*$. 
        For every $uv\in E(G[K_{2,3}])$, if there are at least three edges between $u$ (or $v$) and $V(D)$, but $uv\notin E_1^*$, then we remove all but two of these edges, namely $vv_1,vv_2$ from $E_1^*$. Thus, there exists a pseudo-edge $v_1v_2$ in $\tilde E$ of colour $c_v$. We add pseudo-edge $v_1v_2$ to $P'$. %Note that we decrease 
        
        %%%%%%%%%%%%%%%%

        For every $u,v\in K_{2,3}$ such that $uv\in E$, let $E_{uv}$ be the set of edges of $E^*$ that have at least one endpoint in $\{u,v\}$. We have one of the following cases:
        \begin{itemize}
            \item If $uv\in E^*$;
            We distinguish two sub-cases:\\
            \textbf{Case:} there exists three edges of $E_{uv}$ that forms a path $R=u'uvv'$, then we remove all the edges of $E_{uv}\setminus R$. We add the pseudo-edge $u'v'\in \tilde{E}$ with colour $c_{uv}$ to $P'$.

            \textbf{Case:} Assume there is no such path in $E_{uv}$. Recall $OPT$ is a feasible solution, and, by definition of $K_{2,3}$ at least one of $u$ or $v$ does not have a safe neighbor in $V(D)$. Then, (w.l.o.g.) $u$ is not incident on any edge of $E_{uv}\setminus\{uv\}$. By feasibility of $OPT$, $v$ must be safe and $E_{uv}$ must have at least two edges $vv_1$ and $vv_2$ such that $v_1,v_2\in V(D)$. We remove all edges of $E_{uv}\setminus \{uv,vv_1,vv_1\}$ from $E^*_1$. Finally we observe that by construction $\tilde{E}$ has an edge $v_1v_2$ with colour $c_{uv}$. We add this pseudo-edge to $P'$.
            
            \item If $uv\notin E^*$; Again we distinguish following sub-cases:\\ 
            \textbf{Case:} If one of $u$ or $v$ is incident on exactly one of the edges of $E_{uv}$, w.l.o.g. let $uu'$ be that edge. Since $OPT$ is a feasible FVC solution, $u'$ must be safe. Furthermore as $uv\in E$ and $u,v\in K_{2,3}$, then the other endpoint of the edges of $E_{uv}$ that are incident on $v$ must be unsafe. Hence there are at least two edges incident on $v$. We fix edges $vv_1,vv_2\in E_{uv}$. Now we remove all the edges of $E_{uv}$ from $E_1^*$ except for $\{uu',vv_1vv_2\}$. Finally observe that by construction $\tilde{E}$ has the pseudo-edge $v_1v_2$ with colour $c_{uv}$, which we  add to $P'$.

            \textbf{Case:} Now assume each of $u$ and $v$ is incident to at least two edges of $E_{uv}$.  If there is a path from $u$ to $v$ that only contains edges of $E_1^*$, let $uu_1$ and $vv_1$ be the first and last edge of this path. We update $E_1^*$ by removing all the edges of $E_{uv}$ from it except for two edge $uu_2,vv_2\in E_{uv}$ such that $u_2\neq u_1$ and $v_2\neq u_2$. Then we add $uv$ to $E_1^*$. Note that by this operation we decrease the number of edges of $E_1^*$ by $|E_{uv}|-3$ units and the number of connected components of $(V,E_1^*)$ increases by at most $|E_{uv}|-3$. Furthermore we have a  pseudo-edge $u_2v_2$ of colour $c_{uv}$ in $\tilde{E}$ by construction, we add $u_2v_2$ to $P'$. 
            
            If such a path does not exist then we update $E_1^*$ by removing all the edges of $E_{uv}$ from it except for the two edges $uu_1,vv_1\in E_{uv}$. Then we add the edge $uv$ to $E_1^*$. Observe that $|E_1^*|$ decreased by $|E_{uv}|-3$ and  the number of connected components of $(V,E_1^*)$ increases by at most $|E_{uv}|-3$. This time we add the pseudo-edge $u_1v_1$ of colour $c_{uv}$ to $P'$.

           % \item 
        \end{itemize}
        Observe that by construction $|E_1^*|=2|K_{1,2}|+\frac{3}{2}|K_{2,3}|$. Furthermore by the above argument: 
        $$
            \kappa(E_1^*)-\kappa(E^*)\le |E^*|-|E_1^*|.
        $$
        Furthermore, by construction for every pseudo-edge $p=xy\in P'$ of colour $c_v$, the edges $vx$ and $vy$ are in $E_1^*$, and for every pseudo-edge $p=xy \in  P'$ of colour $c_{uv}$, there exists a path consisting only of edges of $E_1^*$ from $x$ to $y$ with inner vertices in $\{u,v\}$. 
        Therefore, any two vertices in $D$ that are in the same connected component of $(V,E_1^*)$ are also in the same connected component of $(V,P')$. Moreover any vertex $v\in K_{1,2}\cup K_{2,3}$ has a path to at least one vertex of $(V,E_1^*)$, therefore:
        $$
            \kappa(P')=\kappa(E_1^*)-(|K_{1,2}|+|K_{1,2}|).
        $$ 
        Altogether, we have:
        \begin{align*}
            &\kappa(E^*)
            \ge \kappa(E_1^*)-|E^*|+|E_1^*|
            \ge \kappa(P')-(|K_{1,2}|+|K_{1,2}|)-|E^*|+|E_1^*|\\
            =& \kappa(P')+(|K_{1,2}|+\frac{1}{2}|K_{1,2}|)-|E^*|\\
            \ge& \alpha+|K_{1,1}|+|K_{1,2}|+|K_{2,2}|+|K_{2,3}|+(|K_{1,2}|+\frac{1}{2}|K_{1,2}|)-|E^*|.
        \end{align*}
    \end{proof}
    By Claim~\ref{lem:eprimebound}, as $OPT$ is connected, then 
    \begin{align*}
        &|OPT| - |E^*| \geq \alpha + 2K_{1,2}+\frac{3}{2}K_{2,3} - |E^*| + K_{2,2} + K_{1,1} - 1,\\
         \Rightarrow &|OPT| \geq \alpha + 2K_{1,2}+\frac{3}{2}K_{2,3} + K_{2,2} + K_{1,1} - 1.
     \end{align*}
    
    \item $|OPT| \geq n$ follows by our assumption from applying Lemma~\ref{lem:optnottree}.
    
    \item
    It remains to show that $|OPT|\ge 2K_{1,2}-2\alpha_{large}+ \alpha_{1}' + 2\alpha_{2}'$.
    
    Recall that in Algorithm~\ref{alg:phase3}, we compute a subset of singletons in $(V(D), P)$, which we denote by $X_3$, and that $|X_3| = \alpha_1' + \alpha_2'$. Let $K_{1,2}'\subseteq K_{1,2}$ be subset of all vertices in $K_{1,2}$  that have at least one neighbor in $X_3$. Consider $I\coloneqq K_{1,2} \cup X_3\setminus K_{1,2}'$. As $X_3$ and $K_{1,2}$ are independent sets, by definition of $K_{1,2}'$, we see that $I$ is independent.
    
    In any feasible solution, for any $v\in I$, $v$ either has a safe neighbor, or at least two neighbors. Note, there are precisely $\alpha'_1$ vertices in $I$ that have a safe neighbor in $G$, since no vertices of $K_{1,2}$ are adjacent to safe vertices, by definition. Thus, we see 
    $|OPT|\ge 2|I| - \alpha'_1 = 2|K_{1,2}| + \alpha'_1 + 2\alpha'_2 -2|K_{1,2}'|$.

    Thus it suffices to prove that $|K_{1,2}'|\le \alpha_{large}$. If $|K_{1,2}'| = 0$, we are done. We assume that $|K_{1,2}'| > 0$. 
    % For this purpose we use the Lemma~\ref{lem:localswap}. 
    Consider a vertex $v'\in K_{1,2}'$ with neighbour $u\in X_3$. Let $P$ be the set of pseudo-edges computed by Algorithm~\ref{alg:phase1}. In $(V(D), P)$, the degree of $u$ is zero (recall that $X_1\cup X_2 \cup X_3$ were isolated vertices in $(V(D),P)$). 
    
    % Take any other neighbor $w\neq u$ of $v'$ in $V(D)$.
    % {\color{red}If $P$ has no pseudo-edge with colour $c_{v'}$ then we can replace the two edges incident on $v'$ and replace it with edges $v'w$ and $v'u$. This would lead to a larger rainbow forest, a contradiction. }\todo{this line seems odd, is $P$ a forest in general or the rainbow solution we compute?}
    Note that $P$ has a pseudo-edge $xy$ with colour $c_{v'}$. Observe that as $x$ and $y$ are neighbors of $v'$, then we have the option to replace pseudo-edge $xy$ with either $xu$ or $uy$ (as both will exist since pseudo-edges whose  colour is represented by a vertex in $K_{1,2}$ form a clique). 
    Note that this would not change $\kappa(P)$, however if $x$ (or $y$) has degree two or more in $P$, then replacing $xy$ with $uy$ (or $ux$) decreases the number of singletons, a contradiction to the definition of $P$ due to Lemma~\ref{lem:RainbowForest}.
    Therefore as such a swap is not possible, then $x$ and $y$ are of degree one in $(V(D),P)$ (i.e. $x$ and $y$ are only incident on the edge $xy$ in $P$). Therefore $x$ and $y$ form a component of size $2$ in $P$ that is not a singleton (and hence large). We associate this component to $v'$, as the unique pseudo-edge of this component has color $c_{v'}$. Therefore every vertex in $K_{1,2}'$ has a unique large connected component of $P$ associated to it and thus $|K_{1,2}'| \leq \alpha_{large}$.

\end{itemize}
\end{proof}

\subsection{Proof of Lemma~\ref{lem:11/7}}
\label{sec:11/7}
% \begin{align*}
%     \frac{\min\{ \frac{4}{3}(|V(D)| -1) + |K_{1,1}| + 2|K_{1,2}| + |K_{2,2}| + \frac{3}{2}|K_{2,3}| , |V(D)|-1 + K + \alpha - 1 - (\frac{\alpha - \alpha'}{2} + \frac{\alpha_{large}}{2} + \alpha_{1}')\}}{\max\{|K_{1,1}| + 2|K_{1,2}| + 2|K_{2,2}|  + \frac{3}{2}|K_{2,3}| + \alpha - 1, 2K_{1,2} -2\alpha_{large} + \alpha_{1}' + 2\alpha_{2}',n\}} \leq \frac{11}{7}
% \end{align*}

%{\color{red}Using the above claim and Lemma~\ref{lem:FVC3lowerbounds} we have $|OPT|\ge 2K_{1,2}+\alpha'_{1}+2\alpha'_{2}$. Recall that we define $S_P \subseteq S$  as the set of edges of $S$ that have an endpoint in $K_{1,1}\cup K_{1,2} \cup K_{2,2} \cup K_{2,3}$. }

The following Lemma will be a useful tool for proving our approximation factor. 
% It's proof can be found in Appendix~\ref{sec:nextlemma}.
\begin{lemma}
\label{lem:nextlemma}
   $$S\coloneqq \frac{\min \{\frac{4|V(D)|-4}{3} + |S_P|, |S_P| + |V(D)| - 2 + \frac{ 3\alpha}{4} \}}{\max\{ |S_P| + \alpha - 1, |V(D)|+|K| \}}\le \frac{11}{7},$$
   where $\alpha,|V(D)|,|K|$ and $|S_P|$ are real numbers such that $0\le \alpha \le |V(D)|$, $1\le |K|\le |S_P| \le 2|K|$. 
\end{lemma}

\begin{proof}
    First suppose that $|S_P| \leq \frac{11}{7}|K|$, then $\frac{\frac{4|V(D)|-3}{3}+|S_P|}{|K|+|V(D)|}\le \frac{11}{7}$ and the claim holds.

    Therefore, we can assume that $|S_P| > \frac{11}{7}|K|$. We consider cases for the values of $|V(D)|,|K|,|S_P|,\alpha$ maximizing $S$. For simplicity we define the following $UB_1 \coloneqq \frac{4|V(D)|-4}{3} + |S_P| $, $UB_2 \coloneqq |S_P| + |V(D)| - 3 + \frac{ 3\alpha}{4}$, $LB_1 \coloneqq |S_P| + \alpha - 1$, and $LB_2 \coloneqq |V(D)|+|K|$. Let $r \coloneqq |S_P| / |K|$. Observe that $\frac{11}{7} < r \leq 2$.
    \begin{enumerate}
        \item \textbf{Case:} First consider the case that $LB_1=LB_2$ and $UB_1=UB_2$. Since $UB_1 = UB_2$, we can conclude that $|V(D)| = 2 + \frac{9}{4}\alpha$. Since $LB_1 = LB_2$, we can conclude that $|K|  = \frac{5\alpha + 12}{4(r-1)}$. Therefore,
        \begin{align*}
            S 
            &= \frac{UB_1}{LB_2} 
            = \frac{ \frac{4}{3}( \frac{9}{4}\alpha + 1) + r(\frac{5\alpha + 12}{4(r-1)})}{2 + \frac{9}{4}\alpha + \frac{5\alpha + 12}{4(r-1)} } 
            \leq \max \left\{ \frac{3\alpha  + \frac{5\alpha r}{4(r-1)}}{\frac{9}{4}\alpha + \frac{5\alpha}{4(r-1)}}, \frac{\frac{4}{3} + \frac{3}{r-1}}{2 + \frac{3}{r-1}} \right\}\\
            &\leq \max \left\{ \frac{12(r-1) + 5r}{9(r-1) + 5}, \frac{\frac{4}{3}(r-1) + 3r}{2r+1} \right\}
            = \max \left\{ \frac{17r-12}{9r-4}, \frac{\frac{13}{3}r - \frac{4}{3}}{2r+1}\right\}\\
            &\leq \max \left\{ \frac{34-12}{18-4}, \frac{\frac{26}{3} - \frac{4}{3}}{5}\right\} \leq \frac{11}{7},
        \end{align*}
        where the second to last inequality follows as $\frac{7}{11}\le r\le 2$, and both $\frac{17r-12}{9r-4}$ and $ \frac{\frac{13}{3}r - \frac{4}{3}}{2r+1}$  are increasing functions in terms of $r$.
        \item \textbf{Case:} If $\alpha=0$, then:
        \begin{align*}
            S= \frac{\min\{\frac{4|V(D)|-4}{3}+|S_P|,|S_P|+|V(D)|-2 \}}{\max\{ |S_P|-1, |V(D)| + |K| \}} < \frac{|V(D)|+|S_P|-1}{\max\{|S_P|-1, |V(D)| + |K| \}},
        \end{align*}
        Thus, by applying the definition of $r$ we can see
        $$
            S < \frac{|V(D)| + r|K|-1}{\max \{ r|K|-1, |V(D)| + |K| \} }
            \le \frac{|V(D)|+r|K|-1}{ \frac{4r}{11}(|K|-1) + \frac{7}{11}(|V(D)|+|K|) }
        $$
        $$
            =\frac{|V(D)|+r|K|-1}{\frac{7}{11}|V(D)|+(\frac{4r}{11}+1)|K|-\frac{4}{11}}\le \frac{7}{11}.
        $$ Again, the last inequality follows as $\frac{7}{11}\le r\le 2$.    
        \item \textbf{Case:}  If $\alpha=|V(D)|$, then:
        \begin{align*}
            S= \frac{ \min\{ \frac{4|V(D)|-4}{3}+|S_P|,|S_P|+|V(D)|-2+\frac{3|V(D)|}{4} \}}{ \max \{ |S_P|+|V(D)|-1,|V(D)|+|K| \}} \le \frac{\frac{4|V(D)|-4}{3}+|S_P|}{|S_P|+|V(D)|-1}\le \frac{4}{3}.
        \end{align*}
        % Thus,
        % $$S'=\frac{|V(D)|+r|K|}{r|K|,|V(D)|+|K|}\le \frac{|V(D)|+r|K|}{\frac{4}{1}r|K|,\frac{7}{11}(|V(D)|+|K|)}$$
        % $$=\frac{|V(D)|+r|K|}{\frac{7}{11}|V(D)|+(\frac{4r}{11}+1)|K|}\le \frac{7}{11}.$$ Again, the last inequality follows as $\frac{7}{11}\le r\le 2$.
        % \\
        \item \textbf{Case:} If $LB_1=LB_2$ but $UB_1\neq APX$. Assume for the sake of contradiction that $S > \frac{11}{7}$, and consider cases for $UB_1>UB_2$ or not.
        
        \textbf{Case 4.1)} if $UB_1>UB_2$. We define $\Delta\coloneqq (UB_1-UB_2)/(\frac{7}{4}-\frac{4}{3}) = \frac{12(UB_1 - UB_2))}{5}$, $D': = |V(D)| + \Delta$, and $\alpha' \coloneqq  \alpha + \Delta$. Define 
        \begin{align*}
            S' &\coloneqq \frac{\min\{ \frac{4D' - 4}{3} + |S_P|, |S_P| + D' - 2 + \frac{ 3\alpha'}{4}\}}{\max\{ |S_P| + \alpha' - 1, D'+|K| \}}\\
            &= \frac{\min\{ \frac{4D - 4}{3} + |S_P| + \frac{4\Delta}{3}, |S_P| + D - 2 + \frac{ 3\alpha}{4} + \frac{7\Delta}{4}\} }{\max\{ LB_1  + \Delta ,LB_2 + \Delta\}}\\
            & = \frac{\min\{ UB_1 + \frac{4}{3}\Delta, UB_2  + \frac{7}{4}\Delta\}}{\max\{LB_1 + \Delta, LB_2 + \Delta\}}.
        \end{align*}
        By definition of $\Delta$, we have $ UB_1 + \frac{4}{3}\Delta =  UB_2  + \frac{7}{4}\Delta$, and by the assumption of Case (4) we have $LB_1 + \Delta = LB_2 + \Delta$. Therefore, we can apply Case (1) to see that $S' \leq \frac{11}{7}$.
        
        \begin{align*}
            \frac{11}{7} 
            &\geq S' 
            =\frac{\min\{ UB_1 + \frac{4}{3}\Delta, UB_2  + \frac{7}{4}\Delta\}}{\max\{LB_1 + \Delta, LB_2 + \Delta\} }
            = \frac{UB_2 + \frac{7}{4}\Delta}{LB_2 + \Delta} 
            \geq \min \left\{\frac{UB_2 }{LB_2}, \frac{\frac{7}{4}\Delta}{\Delta} \right\} \\
            &=\min \left\{S, \frac{7}{4} \right\} > \frac{11}{7}.
        \end{align*}
        And thus we find  a contradiction.
        % In this case if $S'$ is the new value for $S$ then we have $S'=\frac{UB_2 + \frac{7}{4}\Delta}{LB+\Delta}$. Thus since $S > \frac{11}{7}$, then $S' > \frac{11}{7}$. However in the new $S'$ we have $UB_1=UB_2$, contradicting case 2). 
        
        \textbf{Case 4.2)} if $UB_1<UB_2$. We define $\Delta \coloneqq \min\{(UB_2 - UB_1) / (\frac{7}{4} - \frac{4}{3}), \alpha \}$,  $D' \coloneqq |V(D)| - \Delta$, and $\alpha' \coloneqq \alpha - \Delta$. Now define
        \begin{align*}
            S' \coloneqq \frac{\min\{ \frac{4D' - 4}{3} + |S_P|, |S_P| + D' - 2 + \frac{ 3\alpha'}{4}\}}{\max\{ |S_P| + \alpha' - 1, D'+|K| \}}
        \end{align*}
        If  $\alpha \leq (UB_2 - UB_1) / (\frac{7}{4} - \frac{4}{3}) $, then $\Delta = \alpha$, then $S'$ satisfies the conditions for Case (2), so $S' \leq \frac{11}{7}$. Now suppose that $\alpha > (UB_2 - UB_1) / (\frac{7}{4} - \frac{4}{3}) $.
        \begin{align*}
            &= \frac{\min\{ \frac{4D - 4}{3} + |S_P| - \frac{4\Delta}{3}, |S_P| + D - 2 + \frac{ 3\alpha}{4} - \frac{7\Delta}{4}\} }{\max\{ LB_1  - \Delta ,LB_2 - \Delta\}}\\
            & = \frac{\min\{ UB_1 - \frac{4}{3}\Delta, UB_2  - \frac{7}{4}\Delta\}}{\max\{LB_1 - \Delta, LB_2 - \Delta\}}.
        \end{align*}
        By construction, we have $UB_1 - \frac{4}{3}\Delta = UB_2  - \frac{7}{4}\Delta$. Thus, by Case (1), we have that $S' \leq \frac{11}{7}$. By a similar argument to Case (4), we have $\frac{11}{7} < S \leq S' \leq \frac{11}{7}$, which is a contradiction.
        % Now we decrease $|V(D)|$ and $\alpha$ by value $\Delta=\min\{(UB_2-UB_1)/(7/4-4/3),\alpha-1\}??$.
        % Note that by these changes $S$ would increase as $S'??=\frac{APX-\frac{4}{3}\Delta}{LB-\Delta}$ and $S >\frac{11}{7}$. So either $\alpha$ becomes one which is impossible by case 4) or $UB_1=UB_2$ which contradicts case 2). 
        
        \item \textbf{Case:} If $LB_1<LB_2$. Define $\alpha' = \alpha + \max\{|V(D)|-\alpha,LB_2-LB_1\}$. Define 
        \begin{align*}
            S' = \frac{\min \{UB_1, |S_P| + |V(D)| - 2 + \frac{ 3\alpha'}{4} \}}{\max\{ |S_P| + \alpha' - 1, LB_2\}}.
        \end{align*}
        If $|V(D)|-\alpha \leq LB_2-LB_1$, then $S' = \frac{\min\{UB_1, |S_P| + |V(D)| - 2 + \frac{ 3\alpha'}{4} \}}{\max\{LB_1 + LB_2 - LB_1, LB_2\}}$. So $S'$ satisfies the conditions of Case (4), and $S' \leq \frac{11}{7}$. However, by construction we have that $\frac{11}{7}< S \leq S' \leq \frac{11}{7}$. A contradiction.
    
        If $|V(D)|-\alpha > LB_2-LB_1$, then $S' = \frac{\min\{UB_1, |S_P| + |V(D)| - 2 + \frac{ 3|V(D)|}{4} \}}{\max\{ |S_P| + |V(D)| - 1, LB_2\}}$. In this case, we have $\alpha = |V(D)|$, so $S'$ satisfies the conditions of Case (3), so $S' \leq \frac{11}{7}$. Once again, we have $S \leq S'$, which is a contradiction.
        
        % Note that by this operation $S$ can only increase and we will either have $\alpha=|V(D)|$ which contradicts case (3) or $LB_1=LB_2$ which contradicts case 5).
            
        \item \textbf{Case:} If $LB_1>LB_2$. We define $D' = |V(D)| + LB_1 - LB_2$, and we define $S' \coloneqq \frac{\min \{\frac{4D'-4}{3} + |S_P|, |S_P| + D' - 2 + \frac{ 3\alpha}{4} \}}{\max\{ LB_1, D'+|K| \}} = \frac{\min \{\frac{4D'-4}{3} + |S_P|, |S_P| + D' - 2 + \frac{ 3\alpha}{4} \}}{\max\{LB_1, LB_2 + LB_1 - LB_2\}}$. $S'$ satisfies the conditions of Case (4), and clearly $S' \geq S$, thus we have $\frac{11}{7} < S \leq S' \leq \frac{11}{7}$.
    \end{enumerate}

\end{proof}

\begin{lemma}
\label{lem:11/7tool}
    $$
        S\coloneqq  \frac{ \min\{\frac{4|V(D)| - 4}{3} + |S_P|, |S_P| + |V(D)| - 2 + \alpha - x\}}{ \max \{ |S_P| + \alpha - 1, |V(D)| + |K|, 2K_{1,2} + 2\alpha - 4x\}} \le \frac{11}{7},
    $$
    where $\alpha,|V(D)|,|K|$ and $|S_P|$ are real numbers such that $0\le x\le \alpha \le |V(D)|$, and $1\le |K|\le |S_P| \le 2|K|$. 
\end{lemma}
\begin{proof}
    For notational simplicity we define the following terms $UB_1 = \frac{4|V(D)| - 4}{3} + |S_P|$, $UB_2 = |S_P| + |V(D)| - 2 + \alpha - x$, $LB_1 = |S_P| + \alpha - 1$, $LB_2 = |V(D)| + |K|$, and $LB_3 = 2K_{1,2} + 2\alpha - 4x$.

    \begin{enumerate}
        \item \textbf{Case:} if $x\geq \frac{\alpha}{4}$, then 
        \begin{align*}
            S \leq \frac{ \min\{\frac{4|V(D)| - 4}{3} + |S_P|, |S_P| + |V(D)| - 2 + \frac{3}{4}\alpha \}}{ \max \{ |S_P| + \alpha - 1, |V(D)| + |K|\}} \leq \frac{11}{7}.
        \end{align*}
        Where the first inequality follows by applying the Case assumption and the fact that $\max\{LB_1,LB_2,LB_3\} \geq \max\{LB_1,LB_2\}$ and the second inequality holds using Lemma~\ref{lem:nextlemma}.

        \item \textbf{Case:} If $UB_2\le UB_1$. Define $\Delta \coloneqq \frac{\alpha-4x}{5} > 0$, $x' = x + \Delta$, and $\alpha' = \alpha - \Delta$. Thus, $x' = \frac{\alpha'}{4}$. Furthermore, we can define
        \begin{align*}
            S' &\coloneqq \frac{ \min\{\frac{4|V(D)| - 4}{3} + |S_P|, |S_P| + |V(D)| - 2 + \frac{3}{4}\alpha'\}}{ \max \{ |S_P| + \alpha' - 1, |V(D)| + |K|, 2K_{1,2} + \alpha' \}}\\
            &\leq   \frac{ \min\{\frac{4|V(D)| - 4}{3} + |S_P|, |S_P| + |V(D)| - 2 + \frac{3}{4}\alpha'\}}{ \max \ |S_P| + \alpha' - 1 ,|V(D)|+|K|\}} \leq \frac{11}{7}.
        \end{align*}
        Where the second inequality follows by application of Lemma~\ref{lem:nextlemma}.
        
        Note that $ |S_P| + \alpha' - 1 \leq LB_1$, $ |V(D)| + |K|\leq LB_2$, $2K_{1,2} + 2\alpha' - 4x' \leq LB_3$, and $UB_2 = |S_P| + |V(D)| - 2 + \alpha - \Delta - x + \Delta = |S_P| + |V(D)| - 3 + \alpha' - x'$. Thus, we can see that $S\leq S'\leq \frac{11}{7}$.
        
        \item \textbf{Case:} If $UB_2> UB_1$. Define $\Delta \coloneqq \min\{\alpha - 4x, UB_2 - UB_1\}$, and $\alpha' = \alpha - \Delta$ and 
        \begin{align*}
            S' \coloneqq  \frac{ \min\{\frac{4|V(D)| - 4}{3} + |S_P|, |S_P| + |V(D)| - 2 + \alpha' - x\}}{ \max \{ |S_P| + \alpha' - 1, |V(D)| + |K|, 2K_{1,2} + 2\alpha' - 4x\}}.
        \end{align*}
        First observe that as $\max|S_P| + \alpha' - 1, |V(D)| + |K|, 2K_{1,2} + 2\alpha' - 4x\}\le \max \{ |S_P| + \alpha - 1, |V(D)| + |K|, 2K_{1,2} + 2\alpha - 4x\}$ and $\frac{4|V(D)| - 4}{3} + |S_P|\le |S_P| + |V(D)| - 2 + \alpha' - x<UB_2$ then $S'\ge S$. Thus it is sufficient to show that $S'\le \frac{11}{7}$. We consider two sub-cases.
        If $\alpha - 4x \leq UB_2 - UB_1$. In which case, $\alpha' = \alpha - \alpha + 4x = 4x$, by applying Case (1) $S'$ one sees that $S'\leq \frac{11}{7}$.

        Else, $\alpha - 4x \geq UB_2 - UB_1$. Then we have that $\frac{4|V(D)| - 4}{3} + |S_P|= |S_P| + |V(D)| - 3 + \alpha'-x$, and therefore, Case (2) shows that $S'\le \frac{11}{7}$.
        % Note that as $\max|S_P| + \alpha' - 1, |V(D)| + |K|, 2K_{1,2} + 2\alpha' - 4x\}\le \max \{ |S_P| + \alpha - 1, |V(D)| + |K|, 2K_{1,2} + 2\alpha - 4x\}$ and $\frac{4|V(D)| - 4}{3} + |S_P|=|S_P| + |V(D)| - 2 + \alpha' - x<UB_2$, then $S'\ge S$ and thus $S\le \frac{11}{7}$.
        % \begin{align*}
        %     S' &= \frac{ \min\{\frac{4|V(D)| - 4}{3} + |S_P|, |S_P| + |V(D)| - 3 + \frac{3}{4}\alpha'\}}{ \max \{ |S_P| + \alpha' - 1, |V(D)| + |K|, 2K_{1,2} + 2\alpha' - 4x\}}\\
        %     &\leq \frac{ \min\{\frac{4|V(D)| - 4}{3} + |S_P|, |S_P| + |V(D)| - 3 + \frac{3}{4}\alpha'\}}{ \max \{ |S_P| + \alpha' - 1, |V(D)| + |K|\}}
        % \end{align*}
        % And thus, we can apply Lemma~\ref{lem:nextlemma} to $S'$ to see $S'\leq \frac{11}{7}$.
        
        % Note that $UB_1 < UB_2 <  |S_P| + |V(D)| - 3 + \alpha' - x$. Moreover, $LB_1 \leq |S_P| + \alpha' - 1, LB\leq |V(D)| + |K|,$ and $LB_3 \leq  2K_{1,2} + 2\alpha' - 4x.$ Thus, we can see that $S\leq S'\leq \frac{11}{7}$.
        
        % Note that by these operations $LB_1,LB_2,LB_3$ will not increase, $UB_1$ will not change and $UB_1\le UB_2$. Hence $S$ does not decrease. However now we will either have $x=\frac{\alpha}{4}$ which we can handle using Lemma~\ref{lem:nextlemma}, or we will have $UB_1=UB_2$, which is done by case 1).
    \end{enumerate}
    
   % we decrease $\alpha$ by $\min \{ \alpha -4x$. In this case $S$ will not decrease as $LB_1$, 
\end{proof}
Lastly, we have the following useful claim.
\begin{claim}\label{claim:SimplifyingLB3}
    Define $x\coloneqq \frac{\alpha-\alpha'}{2} + \frac{\alpha_{large} }{2} +\alpha'_1$. Then $\alpha'_1+ 2\alpha'_2 - 2\alpha_{large} \ge 2\alpha-4x$

\end{claim}

We have the ingredients necessary to prove the main Lemma of this section
\begin{proof}
    $2\alpha-4x=2\alpha-4(\frac{\alpha-\alpha'}{2}+\frac{\alpha_{large} }{2}+\alpha'_1)=2\alpha'-4\alpha'_1-\alpha_{large} $\\
    $=2\alpha'_2-2\alpha'_1-2\alpha_{large}  \le \alpha'_1+2\alpha'_2-2\alpha_{large} $
\end{proof}
\begin{lemma}
\label{lem:11/7}
    It holds that 
    \begin{align*}
        \frac{\min\{ \frac{4}{3}(|V(D)| -1) + |S_P| , |V(D)|-2 + |S_P| + \alpha - (\frac{\alpha - \alpha'}{2} + \frac{\alpha_{large}}{2} + \alpha_{1}')\}}{\max\{|S_P| + \alpha - 1, 2K_{1,2} -2\alpha_{large} + \alpha_{1}' + 2\alpha_{2}',n\}} \leq \frac{11}{7}.
    \end{align*}
\end{lemma}

\begin{proof}[Proof of Lemma~\ref{lem:11/7}]
    To simplify notation, we define $K \coloneqq K_{1,1} \cup K_{1,2} \cup K_{2,2} \cup K_{2,3}$. Observe that $|K| \leq |S_P| = |K_{1,1}| + 2|K_{1,2}| + 2|K_{2,2}| + \frac{3}{2}|K_{2,3}| \leq 2(|K_{1,1}| + |K_{1,2}| + |K_{2,2}| + |K_{2,3}|) = 2|K|$.

%      With Lemma~\ref{lem:FVC3lowerbounds}, Lemma~\ref{lem:earAPX}, and Lemma~\ref{lem:FVCupperbound}, we have $\frac{\min\{APX_1, APX_2\}}{|OPT|}$ is at most:

% %    $2K_{1,2} -2\alpha_{large} + \alpha_{1}' + 2\alpha_{2}'$
%     \begin{align*}
%         \frac{\min\{ \frac{4}{3}(|V(D)| -1) + |S_P| , |V(D)|-2 + |S_P| + \alpha - (\frac{\alpha - \alpha'}{2} + \frac{\alpha_{large}}{2} + \alpha_{1}')\}}{\max\{|S_P| + \alpha - 1, 2K_{1,2} -2\alpha_{large} + \alpha_{1}' + 2\alpha_{2}',n\}}.
%     \end{align*}
    Define $x\coloneqq \frac{\alpha-\alpha'}{2} + \frac{\alpha_{large} }{2} +\alpha'_1$. Then, by using Claim~\ref{claim:SimplifyingLB3} we have $\alpha'_1+ 2\alpha'_2 - 2\alpha_{large} \ge 2\alpha-4x$. Noting that $n = |V(D)| + |K|$, we have $\frac{\min\{APX_1, APX_2\}}{|OPT|}$  is at most 
    \begin{align*}
        \frac{\min\{ \frac{4}{3}(|V(D)| -1) + |S_P| , |V(D)|-2 + |S_P| + \alpha - x)\}}{\max\{|S_P| + \alpha - 1, 2K_{1,2} +2\alpha - 4x,|V(D)| + |K|\}}.
    \end{align*}
    And by Lemma~\ref{lem:11/7tool} this is at most $\frac{11}{7}$.
\end{proof}

    \section{Proof of Lemma~\ref{lem:2ECSSImprovement}}
\label{sec:jensextension}
\begin{proof}
     If $G$ is $2$VC then one can solve the problem using Lemma~\ref{lem:jensextension}. 
     So assume $V$ is not $2VC$. Then in polynomial time we can decompose it into blocks $E_1,...,E_b$ (see \cite{west2001introduction} for further details).
     Let $G_1,...,G_b$ be the induced subgraph on $E_i$ (i.e. $G_i:=G[E_i]$). Furthermore let $n_i$ be the number of vertices of $G_i$ and $opt_i$ be the size of the optimal $2$-ECSS of $G_i$. We have:
     $$
        \sum_{i=1}^b n_i=n+b-1.
     $$
     Now let $E'\subseteq E$. One observes that $E'$ is a $2$-ECSS of $G$ if and only if $E'\cap E_i$ is a $2$-ECSS of $G_i$ for every $i\in \{1,...,b\}$.
     Therefore now we can use Lemma~\ref{lem:jensextension} for each $i$ to get an approximate solution $APX_i$ of size at most $\frac{4}{3}n_i+\frac{2}{3}(opt_i-n_i)-\frac{2}{3}$ for $G_i$. Thus if we set $APX:=APX_1\cup\dots\cup APX_b$, we obtain a feasible solution for $G$ of size:
     $$
        |APX|\le \frac{4}{3}(n+b-1)+\frac{2}{3}(opt-(n+b-1))-\frac{2}{3}b
    $$
    $$
        =\frac{4}{3}n+\frac{2}{3}(opt-n)-\frac{2}{3}.
    $$
  
\end{proof}
    \section{Previous Analysis of \texorpdfstring{$k-$}{k-}FGC}
\label{sec:bug}
In \cite{adjiashvili2022flexible} the authors address the problem of $k$-FGC with Theorem 3, which is proven in Section 2 of that paper. We now discuss the relevant elements of the proof of Theorem 3 provided in \cite{adjiashvili2022flexible}. The authors provide two approximation algorithms for the problem, and by taking the minimum, they find the desired approximation of $1 + O(\frac{1}{k})$. Here we focus on the second algorithm and its analysis, and we will point out an inequality that is incorrect. We provide here a close depiction of the proof found in \cite{adjiashvili2022flexible} in order it identify this issue. 

% \todo{ what follows is directly copied from the \cite{adjiashvili2022flexible}}

In the construction the authors of \cite{adjiashvili2022flexible} provide, they let $I' = (G, F, k)$ be an instance of unweighted $k$-FGC, where $G$ is the input graph, and $F$ is the set of all safe edges of $G$. Their approximation algorithm proceeds as follows. First, we compute an edge-set $X$ that forms a maximum forest restricted to the safe edges $F$ of $G$. Then we compute a $(k + 1)$-edge connected spanning subgraph $Y'$ of $G/X$ using the algorithm given in \cite{cheriyan2000approximating}.% 
Clearly, the algorithm above runs in polynomial time.

\begin{lemma}
\label{lem:LemmaOneAdjishvali}
    Let $H \subseteq G$ be a feasible solution to $I'$. Then $H / (F \cap E(H))$ is $(k+1)$-edge connected.
\end{lemma}

Let $Z^* \subseteq G$ be an optimal solution to $I'$ and let $X \cup Y'$ be a solution computed by the algorithm described above. To prove Theorem 3 it remains to bound the size of the solution $X \cup Y'$.
Let $\ell := |X|$ and let $Y^*$ be a minimum $(k+1)$-edge connected spanning subgraph of $G/X$. 
By Lemma 1, we have that $X \cup Y'$, and $X \cup Y^*$ are feasible solutions to $I'$.

Recall from Section 1.3 that the optimal solution $Z^*$ consists of 2-edge-connected components that are joined by safe edges $E' \subseteq Z^* \cap F$ in a tree-like fashion. The maximum forest $X$ also joins these 2-edge-connected components of $Z^*$, that is, $(Z^* \backslash E') \cup X$ forms a connected graph. Hence, $Z^* \backslash X$ forms a 2-edge connected graph of $G/X$. 
Therefore we have that $\text{OPT}(I') = |Z^*| \geq |Y^*|$.
%and the size of the solution $X \cup Y$ output by the algorithm can be bounded as follows
% \[
% |X| + |Y| \leq \ell + \theta_{k+1} \cdot |Y^*| \leq \ell + \theta_{k+1} \cdot \text{OPT}(I')
% \]

On the other hand, it was shown in \cite{nagamochi1992linear} that $Y'$ can be partitioned into $k$ spanning forests. Thus $Y'$ contains at most $(k+1)n'$ edges, where $n' := |V(G/X)|=n-\ell$. By Lemma~\ref{lem:kconnectedlowerbound}, we have  that any feasible solution to a $(k+1)$ECSS instance on $n'$ vertices needs to contain at least $\frac{k+1}{2}n'$ edges, since each vertex has to have degree at least $k + 1$. Therefore, we have that
\[
|X| + |Y'| \leq \ell + (k + 1)(n - \ell) \leq 2\text{OPT}(I') - k\ell ,
\]
where $n = |V(G)|$.

Now we show that the second inequality does not necessarily hold. In fact $2\text{OPT}(I') - k\ell$ can potentially be much smaller than $\ell + (k + 1)(n - \ell)$. For instance, consider the case that  $G$ has a spanning tree containing solely of safe edges and $k\ge 3$. In this case $\text{OPT}(I') = |X| = n-1$. Therefore $2\text{OPT}(I') - k\ell = (2-k)\ell <0$. However $\ell + (k + 1)(n - \ell)=\ell+(k+1)\ge n$. Thus in this case:
$$
    \ell + (k + 1)(n - \ell)\ge n > 0 > 2\text{OPT}(I') - k\ell.
$$

The correct inequality should be:
    \[
|X| + |Y'| \leq \ell + (k + 1)(n - \ell) \leq 2\text{OPT}(I') - \ell ,
\]
where the second inequality holds by applying Lemma~\ref{lem:LemmaOneAdjishvali} to see that $|X|+\frac{(n-|X|)k}{2} \le \text{OPT}(I')$ and the fact that  every $(k+1)$-edge connected subgraph of $G/ X$ must have at least $\frac{(n-|X|)}{2}$ edges.

As previously mentioned, the authors of \cite{adjiashvili2022flexible} provide a second approximation algorithm. Which simply computes the best $(k+1)$ECSS solution $Y$ instead of $Y'$.
\begin{lemma}[\cite{adjiashvili2022flexible}]
    There is a polynomial time algorithm that has cost at most $\ell + \beta_{k+1}OPT(I')$.
\end{lemma}

Taken together, we find an approximation algorithm with approximation factor
\begin{align*}
    &\min\{\ell + \beta_{k+1}OPT(I'), 2OPT(I') - 2\ell\} 
    \leq \frac{1}{2}(\ell + \beta_{k+1}OPT(I')) + \frac{1}{2}(2OPT(I') - \ell) \\
    =& \frac{1}{2}(\beta_{k+1} + 2)OPT(I') 
    = \left( \frac{3}{2} + O \left( \frac{1}{k} \right) \right) OPT(I').
\end{align*}
% Where $\frac{2}{3}(\beta_{k+1}+1) = \frac{4}{3} + O(\frac{1}{k})$.
% \[
%     \min\{|X| + |Y|, |X| + |Y'|\} \leq \min\{\ell + \theta_{k+1} \cdot \text{OPT}(I'), 2\text{OPT}(I') - k\ell\}
% \]
% \[
%     \leq \text{OPT}(I') \cdot \frac{2 + k + \theta_{k+1}}{k + 1}
% \]

\end{appendix}

\end{document}